\documentclass{article}
\pdfoutput=1

\newif\ifbiblatex
    \biblatexfalse %
\newif\ifvfull
    \vfullfalse %
\newif\ifchecklist
    \checklisttrue

\ifbiblatex\else
     \PassOptionsToPackage{numbers}{natbib} \fi

\usepackage[final,\ifbiblatex nonatbib\fi]{neurips_2024}

\ifbiblatex\else
\fi

\usepackage[utf8]{inputenc} %
\usepackage[T1]{fontenc}    %
\usepackage{url}            %
\usepackage{booktabs}       %
\usepackage{amsfonts}       %
\usepackage{nicefrac}       %
\usepackage{microtype}      %

\relax %
    \usepackage[dvipsnames]{xcolor}
    \usepackage{tikz}
        \usetikzlibrary{decorations.pathmorphing}
    \usepackage{mathtools}
    \usepackage{amssymb}
		\DeclareMathSymbol{\shortminus}{\mathbin}{AMSa}{"39}
    \usepackage{bbm}
    \usepackage{graphicx}
    \usepackage{wrapfig}
    \usepackage{scalerel}
    \usepackage{verbatim} %
    \usepackage{enumitem}

    \usepackage{caption,subcaption}
    \usepackage{hyperref} %
        \hypersetup{colorlinks=true, linkcolor=blue!75!black, urlcolor=magenta, citecolor=green!50!black}

\relax %
	\usetikzlibrary{positioning,fit,calc, decorations, arrows, shapes, shapes.geometric}
	\usetikzlibrary{patterns,backgrounds}
	\usetikzlibrary{cd}
	\tikzset{is bn/.style={background rectangle/.style={fill=blue!35,opacity=0.3, rounded corners=5},show background rectangle}}

\usepackage{amsthm,thmtools} %
	\usepackage[noabbrev,nameinlink,capitalize]{cleveref}
    \theoremstyle{plain}
    \newtheorem{theorem}{Theorem}
	
    \newtheorem{prop}[theorem]{Proposition}
    \newtheorem{claim}{Claim}
    \newtheorem{conj}[theorem]{Conjecture}
    
    \newtheorem{lemma}[theorem]{Lemma}
    \theoremstyle{definition}
    \declaretheorem[name=Definition, qed=$\square$]{defn}
    \declaretheorem[name=Construction, sibling=defn]{constr}
    \declaretheorem[name=Example, qed=$\triangle$]{example}

	\crefname{defn}{Definition}{Definitions}
	\crefname{prop}{Proposition}{Propositions}
    \crefname{issue}{Issue}{Issues}

\relax %
    
	\let\H\relax
	\DeclareMathOperator{\H}{\mathrm{H}} %
	\DeclareMathOperator{\I}{\mathrm{I}} %
	\DeclareMathOperator*{\Ex}{\mathbb{E}} %

    \newcommand{\mat}[1]{\mathbf{#1}}
    \DeclarePairedDelimiterX{\infdivx}[2]{(}{)}{%
		#1\;\delimsize\|\;#2%
	}
	\newcommand{\thickD}{I\mkern-8muD}
	\newcommand{\kldiv}{\thickD\infdivx}
    
	\newcommand{\tto}{\rightarrow\mathrel{\mspace{-15mu}}\rightarrow}
    \newcommand{\CI}{\mathbin{\bot\!\!\!\bot}}

	\newcommand\numberthis{\addtocounter{equation}{1}\tag{\theequation}}
    
	\makeatletter
	\newcommand{\subalign}[1]{%
	  \vcenter{%
	    \Let@ \restore@math@cr \default@tag
	    \baselineskip\fontdimen10 \scriptfont\tw@
	    \advance\baselineskip\fontdimen12 \scriptfont\tw@
	    \lineskip\thr@@\fontdimen8 \scriptfont\thr@@
	    \lineskiplimit\lineskip
	    \ialign{\hfil$\m@th\scriptstyle##$&$\m@th\scriptstyle{}##$\hfil\crcr
	      #1\crcr
	    }%
	  }%
	}    
    \newcommand{\hathat}[1]{%
    \begingroup%
      \let\macc@kerna\z@%
      \let\macc@kernb\z@%
      \let\macc@nucleus\@empty%
      \hat{\mathchoice%
        {\raisebox{.2ex}{\vphantom{\ensuremath{\displaystyle #1}}}}%
        {\raisebox{.2ex}{\vphantom{\ensuremath{\textstyle #1}}}}%
        {\raisebox{.16ex}{\vphantom{\ensuremath{\scriptstyle #1}}}}%
        {\raisebox{.14ex}{\vphantom{\ensuremath{\scriptscriptstyle #1}}}}%
        \smash{\hat{#1}}}%
    \endgroup%
    }
    \makeatother
    
    \newcommand{\derind}[2]{ {#1}^{+\tilde{#2}(\X^{\U})}}

\newcommand{\TODO}[1][INCOMPLETE]{{\centering\Large\color{red}$\langle$~\texttt{#1}~$\rangle$\par}}

\relax %
    \newcommand{\Pa}{\mathop{\mathbf{Pa}}}
	\newcommand{\SIMInc}{\mathrm{\SQIM}\mathit{\mskip-2.5muI\mskip-3.5mun\mskip-1.7muc}} %

\usepackage[nobox]{restatelinks}
\relax%
	\usetikzlibrary{positioning,fit,calc, decorations, arrows, shapes, shapes.geometric}
	\usetikzlibrary{cd}

	\tikzset{AmpRep/.style={ampersand replacement=\&}}
	\tikzset{center base/.style={baseline={([yshift=-.8ex]current bounding box.center)}}}
	\tikzset{paperfig/.style={center base,scale=0.9, every node/.style={transform shape}}}

	\tikzset{dpadded/.style={rounded corners=2, inner sep=0.7em, draw, outer sep=0.3em, fill={black!50}, fill opacity=0.08, text opacity=1}}
	\tikzset{dpad1/.style={outer sep=0.1em, inner sep=0.4em, draw=gray!50, rounded corners=2, fill=black!08, fill opacity=1, align=center}}
	\tikzset{dpad0/.style={outer sep=0.05em, inner sep=0.3em, draw=gray!75, rounded corners=4, fill=black!08, fill opacity=1, align=center}}
	\tikzset{dpadinline/.style={outer sep=0.05em, inner sep=2.5pt, rounded corners=2.5pt, draw=gray!75, fill=black!08, fill opacity=1, align=center, font=\small}}

 	\tikzset{dpad/.style args={#1}{every matrix/.append style={nodes={dpadded, #1}}}}
	\tikzset{light pad/.style={outer sep=0.2em, inner sep=0.5em, draw=gray!50}}

	\tikzset{arr/.style={draw, ->, thick, shorten <=3pt, shorten >=3pt}}
	\tikzset{arr0/.style={draw, ->, thick, shorten <=0pt, shorten >=0pt}}
	\tikzset{arr1/.style={draw, ->, thick, shorten <=1pt, shorten >=1pt}}
	\tikzset{arr2/.style={draw, ->, thick, shorten <=2pt, shorten >=2pt}}

	\newcommand\cmergearr[5][]{
		\draw[arr, #1, -] (#2) -- (#5) -- (#3);
		\draw[arr, #1, shorten <=0] (#5) -- (#4);
		}
	\newcommand\mergearr[4][]{
		\coordinate (center-#2#3#4) at (barycentric cs:#2=1,#3=1,#4=1.2);
		\cmergearr[#1]{#2}{#3}{#4}{center-#2#3#4}
		}
	\newcommand\cunmergearr[5][]{
		\draw[arr, #1, -, shorten >=0] (#2) -- (#5);
		\draw[arr, #1, shorten <=0] (#5) -- (#3);
		\draw[arr, #1, shorten <=0] (#5) -- (#4);
		}
	\newcommand\unmergearr[4][]{
		\coordinate (center-#2#3#4) at (barycentric cs:#2=1.2,#3=1,#4=1);
		\cunmergearr[#1]{#2}{#3}{#4}{center-#2#3#4}
		}

\relax %
    \newcommand{\nhphantom}[2]{\sbox0{\kern-2%
    \nulldelimiterspace$\left.\delimsize#1\vphantom{#2}\right.$}\hspace{-.97\wd0}}

    \makeatletter
    \newsavebox{\abcmycontentbox}
    \newcommand\DeclareDoubleDelim[5]{
    \DeclarePairedDelimiterXPP{#1}[1]%
        {%
            \sbox{\abcmycontentbox}{\ensuremath{##1}}%
        }{#2}{#5}{}%
        {%
            \nhphantom{#3}{\usebox\abcmycontentbox}%
            \hspace{1.2pt} \delimsize#3%
            \mathopen{}\usebox{\abcmycontentbox}\mathclose{}%
            \delimsize#4\hspace{1.2pt}%
            \nhphantom{#4}{\usebox\abcmycontentbox}%
        }%
    }
    \makeatother

\relax %
	\newcommand{\X}{\mathcal X}
	\newcommand{\V}{\mathcal V}
	\newcommand{\N}{\mathcal N}
	\newcommand{\Ar}{\mathcal A}

    \newcommand{\balpha}{\boldsymbol\alpha}
    \newcommand{\bbeta}{\boldsymbol\beta}

	\def\p_#1{\mathbb P_{\!#1\mskip-2mu}}

	\newif\ifsuba \subatrue
	\newcommand\Src[1]{\ifsuba S\mskip-2mu\vphantom{|}_{{#1}} \else S \fi}
	\newcommand\Tgt[1]{\ifsuba T\mskip-3mu\vphantom{|}_{{#1}} \else T \fi}

	\DeclareMathAlphabet{\mathdcal}{U}{dutchcal}{m}{n}
	\DeclareMathAlphabet{\mathbdcal}{U}{dutchcal}{b}{n}
	\newcommand{\dg}[1]{\mathbdcal{#1}}

	\newcommand\Inc{\mathit{Inc}}
	\newcommand{\IDef}{\mathit{I\mskip-3.0muD\mskip-2.5mue\mskip-1.5muf}}
	\newcommand\SDef{\mathit{S\mskip-3.5muD\mskip-2.5mue\mskip-1.5muf}} %
	\newcommand{\ed}[3]{#2%
	 {\overset{\smash{\mskip-5mu\raisebox{-1pt}{$\scriptscriptstyle
	        #1$}}}{\rightarrow}} #3}

	\DeclareDoubleDelim
		\SD\{\{\}\}
	\DeclareDoubleDelim
		\bbr[[]]
	\makeatletter
	\newsavebox{\aar@content}
	\newcommand\aar{\@ifstar\aar@one@star\aar@plain}
	\newcommand\aar@one@star{\@ifstar\aar@resize{\aar@plain*}}
	\newcommand\aar@resize[1]{\sbox{\aar@content}{#1}\scaleleftright[3.8ex]
		{\Biggl\langle\!\!\!\!\Biggl\langle}{\usebox{\aar@content}}
		{\Biggr\rangle\!\!\!\!\Biggr\rangle}}
	\DeclareDoubleDelim
		\aar@plain\langle\langle\rangle\rangle
	\makeatother

\renewcommand{\V}{\mathrm V}

\newcommand\U{\mathcal U}

\newcommand\emodels{\mathbin{{\models}^{\mathllap{e\,}}_{\vphantom{l}}}}
\newcommand\dotmodels{\mathbin{{\models}^{\mathllap{\bullet\,}}_{\vphantom{l}}}}

\newcommand\PSEMs{\mathrm{PSEMs}}
\newcommand\PSEMsA{\mathrm{PSEMs}_{\!\Ar}}
\newcommand\Wits{\mathrm{Wits}}

\newcommand\enV{\mathcal V} %

\newcommand\infl{\Pa}
\newcommand\muxor{\mu_{\mathit{xor}}}

\newcommand\hyperarc{hyperarc}
\newcommand\arc{hyperarc}

\newcommand\partl{partitional}
\newcommand\subpartl{sub\partl}
\newcommand\unipartl{uni\partl}
\newcommand\suppartl{super\partl}

\newcommand\SQIM{QIM}%
\newcommand\scibility{\SQIM-compatibility}
\newcommand\Scibility{\SQIM-compatibility}
\newcommand\SCibility{\SQIM-Compatibility}

\newcommand\scible{\SQIM-compatible}
\newcommand\escible{E\scible}
\newcommand\cible{compatible}
\newcommand\cibility{compatibility}

\newcommand\vjoe[1]{{\color{red!80!black}#1}}
\newcommand\voli[1]{{\color{orange!80!black}#1}}

\newenvironment{vnew}{\color{green!80!blue!40!black}}{}
\newcommand\dhygraph{directed hypergraph}
\newcommand\hgraph{hypergraph}
\newcommand\DHygraph{Directed Hypergraph}

\newcommand\commentout[1]{}
\ifvfull
\newcommand\vfull[1]{{\color{gray}\hypersetup{linkcolor=blue!75!black!40!white, urlcolor=magenta!40!white,citecolor=green!50!black!40!white}#1}}
\else
\newcommand\vfull[1]{}
\fi

\relax%
    \newsavebox\cyclebox
    \sbox\cyclebox{\begin{tikzpicture}[center base]
        \node[dpad0] (X) at (0:.8) {$X$};
        \node[dpad0] (Y) at (120:.8) {$Y$};
        \node[dpad0] (Z) at (-120:.8) {$Z$};
        \draw[arr2] (X) to 
            (Y);
        \draw[arr2] (Y) to
            (Z);
        \draw[arr2] (Z) to 
            (X);
    \end{tikzpicture}}
\title{
Qualitative Mechanism Independence
}
\author{%
    Oliver E.~Richardson\\
    Dept of Computer Science\\
    Cornell University\\
    Ithaca NY 14853 \\
    \texttt{oli@cs.cornell.edu} \\
\And
    Spencer Peters \\
    Dept of Computer Science\\
    Cornell University\\
    Ithaca NY 14853 \\
    \texttt{speters@cs.cornell.edu}
\And
    Joseph Y.~Halpern \\
    Dept of Computer Science\\
    Cornell University\\
    Ithaca NY 14853 \\
    \texttt{halpern@cs.cornell.edu}
}

\begin{document}
\maketitle

\begin{abstract}
We define what it means for a joint probability distribution
to be 
\emph{(QIM-)compatible}
with a set of independent causal mechanisms,
at a qualitative level---%
or, more precisely, with a directed hypergraph $\Ar$,
which is the qualitative structure of
a probabilistic dependency graph (PDG). 
When $\Ar$ represents a qualitative Bayesian network, 
{\scibility} with $\Ar$
reduces to satisfying the appropriate conditional independencies. 
But giving semantics to hypergraphs using \scibility\ lets us do much more.
For one thing, we can capture functional \emph{dependencies}.
For another, \scibility\ captures important aspects of causality:
we can use compatibility to understand cyclic causal graphs, and
to demonstrate compatibility is essentially to produce a causal model. 
Finally, compatibility has deep connections to 
information theory.
Applying compatibility to cyclic structures helps to clarify a
longstanding conceptual issue 
in information theory.
\end{abstract}

\commentout{
\begin{abstract}
We present a simple yet expressive graphical language for describing qualitative structure
such as dependence and independence in probability distributions. 
At its core is a definition of what it means for a joint probability distribution
to be \emph{(QIM)-compatible} with a set of independent mechanisms, as represented by a (directed) hypergraph $\mathcal A$.
(Our definition is inspired by probabilistic dependency graphs (PDGs), which have hypergraphs as their qualitative structure.)
It is closely related to causality---demonstrating compatibility is essentially equivalent to producing a causal model.
In the special case where $\mathcal A$ represents a qualitative Bayesian network or acyclic causal model, 
QIM-compatibility with $\mathcal A$ reduces to satisfying the appropriate conditional independencies. 
But QIM-compatibility lets us do much more. 
For one thing, it can capture functional dependencies; for another, it gives qualitative meaning to cyclic structures, which has been largely an open problem.
We show how to derive concrete consequences of this qualitative meaning by proving a generalization of the Bayesian network special case: for \emph{any} hypergraph $\mathcal A$, QIM-compatibility with it implies information-theoretic constraints (of a form that generalizes both conditional independencies and functional dependencies). 
As an application, we give a novel formalization of the notion of 
``only pairwise interactions''
 as QIM-compatibility with a certain cyclic structure, that leads us to a surprising resolution of a longstanding conceptual issue with Shannon information.
 \end{abstract}}

\section{Introduction}
The 
structure of a 
(standard)
probabilistic graphical model
(like a Bayesian Network or Markov Random Field)
    encodes a set of conditional independencies among variables.
This is useful because it enables a compact description of
    probability distributions that have those independencies;
it also lets us use graphs as a visual language for
    describing important qualitative
    properties of a probabilistic world.
Yet these kinds of independencies are not the only important
    qualitative aspects of a probability measure.
In this paper, we study a natural generalization of standard graphical model structures that can describe far more than conditional independence. 

For example, another qualitative aspect of a probability distribution is that of
functional \emph{dependence}, 
    which is also exploited across computer science to enable compact representations and simplify probabilistic 
    analysis.
\commentout{
    An acyclic causal model, for example, specifies a distribution  
    via a probability over \emph{contexts} (the values of variables whose
    cause is determined outside the model) 
    together with a collection of equations (i.e., functional dependencies)
    describing how each non-exogenous variable can be determined
        as a function of other variables \cite{pearl:2k}.
}%
Acyclic causal models, for instance, specify a distribution  via a
probability over \emph{contexts}
(the values of variables whose causes are viewed as outside the model),
and a collection of equations (i.e., functional dependencies) \citep{pearl:2k}.
And in deep learning, a popular class of models called 
    \emph{normalizing flows}
    \citep{tabak2010density,Kobyzev_2021}
    specify a distribution
    by composing a fixed
    distribution
    over some latent space,
    say a standard normal distribution,
    with a function (i.e., a functional dependence) fit to observational data.
\vfull{%
Similarly, complexity theorists often regard
    a probabilistic Turing machine as a deterministic function
    that takes as input a uniformly random string \citep{??}.
}%
Functional dependence and independence 
are deeply related and interacting notions.
For instance,
    if $B$ is a function of $A$ (written $A \tto B$) and $A$ is independent of $C$ (written $A \CI C$), then $B$ and $C$ are also independent ($B \CI C$).
\unskip\footnote{This well-known fact (\cref{lem:indep-fun}) is formalized
    and proved in \cref{appendix:proofs},
    where all proofs can be found.}
Moreover, dependence can be written in terms of independence:
    $Y$ is a function of $X$ if and only if
    $Y$ is conditionally independent of itself given $X$ 
    (i.e., $X \tto Y$ iff $Y \CI Y \mid X$).
\commentout{
\begin{equation}
    Y \CI Y \mid X 
    \quad\iff\qquad 
    X \tto Y .
        \label{eq:conditional-self-independence-det}
\end{equation}}
Traditional graph-based languages such as Bayesian Networks (BNs) and
    Markov Random Fields (MRFs) cannot capture 
these relationships.
Indeed, the graphoid axioms (which describe BNs and MRFs) \cite{pearl1987graphoids} and axioms for conditional independence \citep{naumov2013re},
do not even consider statements like $A \CI A$ to be syntactically valid. 
Yet such statements are perfectly meaningful,
and reflect a deep relationship between independence, dependence, and generalizations of both notions
(grounded in information theory, a point we will soon revisit).

\commentout{%
Yet another critical qualitative aspect of a distribution is the
causal connections it describes. As we shall see, this too can be
captured in our framework, and doing so allows us to capture
connections between causality, independence, and interventions.
Indeed, our approach brings causal reasoning to the forefront.
}%

This paper
provides
a simple yet expressive graphical
    language for describing qualitative
    structure 
    such as dependence and independence
    in probability distributions.
\commentout{%
\emph{Probabilistic dependency graphs (PDGs)}, introduced in
\citep{pdg-aaai}), are a 
graphical 
    language for describing structure such as dependence and independence
    in probability distributions.
They are a generalization of
BNs and factor graphs, that can capture inconsistent beliefs \citep{pdg-aaai},
Roughly speaking, a PDG is a graphical model where,
like a Bayesian network (BN), nodes are associated with variables, and
are related by \hyperarc s.  
In a (quantitative) PDG, these \hyperarc s
are associated conditional probability tables, in the spirit of a BN.
In this paper, we focus on qualitative PDGs, which can be viewed as
\emph{\dhygraph s}.
We provide a novel semantics for qualitative PDGs that allows us to
capture dependence and independence in a way that generalizes the
semantics of qualitative BNs and, as we show, has a variety of useful
applications (including, among other things, allowing us to capture
the reasoning dependency described above).
}%
The idea
is to specify 
the inputs and outputs of  
a set of \emph{independent mechanisms}: processes by which
some 
target variables $T$ are determined as a
(possibly randomized) function of  
some
source variables $S$.
This idea generalizes 
    intuition going back to \citet{pearl:2k}
    by allowing, for example, 
    two mechanisms to share a target variable.
So at a qualitative level, 
the modeler specifies not a (directed)
graph, but a 
(directed)
\emph{hypergraph}.
\commentout{%
---which is the structure of
another type of probabilistic graphical model:
a \emph{probabilistic dependency graph (PDG)} \citep{pdg-aaai,one-true-loss,pdg-infer}.
}%

If we were interested in a
    concrete probabilistic model, we would also need to annotate this hypergraph with 
    quantitative information describing the mechanisms.  
For directed acyclic graphs, there are two standard approaches: supply conditional
probability distributions (cpds) to get a BN, or supply equations to get a causal model.
Correspondingly, there are two approaches to probabilistic modeling based on hypergraphs.
The analogue of the first approach---supplying a probability $P(T|S)$ for each mechanism---%
    leads to the notion of a \emph{probabilistic dependency graph (PDG)} \citep{pdg-aaai,one-true-loss,pdg-infer}.
The analogue of the second approach---supplying an equation describing $T$ as a function of $S$ and independent random noise---leads to a novel generalization of a causal model
    (\cref{defn:GRPSEM}).
Models of either kind are of interest to us only insofar as they explain 
    how hypergraphs encode qualitative aspects of probability. 
Qualitative information in a PDG was characterized
by \citet{pdg-aaai}
using a scoring function
that, despite having some attractive  properties, 
lacks justification and has not been fully understood.
In particular, the PDG formalism does not appear to answer 
    a basic question:
\textit{%
    what does it mean for a distribution to be compatible
    with a directed hypergraph structure?
}

We develop precisely such a notion (\cref{defn:scompat})
of compatibility between a distribution $\mu$ and 
a directed hypergraph qualitatively representing a collection of independent mechanisms---%
or, for short, simply {(QIM-)}compatibility.
This 
definition
    allows us to use directed hypergraphs as
    a language for specifying structure in 
    probability distributions,
    of which the semantics of qualitative BNs are a special case (\cref{theorem:bns}).
Yet \scibility\ can do far more than represent conditional independencies in acyclic networks. 
For one thing, it can encode arbitrary 
    functional dependencies (\cref{theorem:func});
for another, it gives meaningful semantics to cyclic models.
\commentout{\Scibility\ satisfies a \emph{monotonicity} property (\cref{sec:monotone}) that helps us begin to understand these cyclic models.}%
Indeed, compatibility lets us go well beyond capturing dependence
and independence. The fact that \citet{pearl:2k} views
causal models as representing independent mechanisms suggests
that there might be a 
connection to causality. 
In fact, there is.
A \emph{witness} that a distribution $\mu$ is \cible\ with a 
\hgraph\ $\Ar$ is an extended distribution $\bar\mu$ that 
is nearly equivalent to (and guarantees the existence of)
a causal model
that explains $\mu$ with dependency structure $\Ar$
(\cref{prop:sc-graph-arise,prop:gen-sim-compat-means-arise,prop:witness-model-properties}). 
As we shall see, thinking in terms of 
witnesses and compatibility
allows us to tie
together causality, dependence, and independence.

Perhaps
surprisingly,
\cibility\ 
also has deep connections with information theory
(\cref{sec:info}).
\commentout{%
    These connections turn out to be quite useful: they yield a generic
    information-theoretic test for complex (e.g., cyclic) dependency
    structures, and enable causality to clarify some important
    misunderstandings in information theory
    (\cref{example:nonneg-ii,example:ditrichotomy}).
}%
The
conditional independencies of a BN can be viewed as a certain kind of information-theoretic constraint.
Our notion of compatibility with a hypergraph $\Ar$ 
    turns out to imply a generalization of this
    constraint
    (closely related to the qualitative PDG scoring function)
    that is meaningful for all hypergraphs
    (\cref{theorem:sdef-le0}).
Applied to cyclic models, it yields a causally inspired notion of pairwise interaction that clarifies some important misunderstandings in information theory (\cref{example:nonneg-ii,example:ditrichotomy}).
\vfull{%
    It also gracefully handles incomplete fragments of a causal picture, as well as ``over-determined'' ones.
}%
\commentout{
Last but far from least, 
    our definition of \scibility\ 
        also has deep ties to the causality literature.
    A witness that a distribution $\mu$ is \scible\ with $\Ar$ 
    is an extended distribution $\hat\mu$ that 
    can be naturally converted to a \emph{structural equations model (SEM)}
        \citep{pearl:2k}.
    \commentout{
        (possibly with constraints \citep{beckers2023causal}).
    }
    We explore this in \cref{sec:causal}.
}

\commentout{%
Clearly the applicaitons of \scibility\ extend far beyond providing a compatibility semantics for the qualitative information in a PDG. 
Yet, on that original front, the situation is only partially resolved.
\commentout{
\citet{pdg-aaai}, for instance, define what it means for a distribution
    to be compatible with the probabilistic information in a PDG,
    but do not provide an analogous definition for the qualitative information.
They do, however, provide a scoring function quantifying the discrepancy between a probability distribution and the qualitative information in a PDG.  
While that scoring function   has some nice properties,
it is not well-justified, and can also behave unintuitively.
(For example, it can be negative, making it difficult to interpret 
    \unskip.) 
We derive an arguably 
    more intuitive scoring function here, 
    based on \scibility.
    Our new scoring function turns out to have deep
connections to the original
one:
the former is both a pointwise upper bound on the latter, 
and also a special case of it (when applied to a transformed PDG).
We explore these connections in \cref{sec:pdgs}. 
}%
Although the two are very closely related, we have found that \scibility\ is not quite the same as the qualitative aspect of a PDG (at least, not with its current scoring function semantics). 
In fact, we will see (\cref{sec:pdgs}) that PDGs are actually more expressive, in the sense that our concept of mechanism independence can be naturally encoded within them. 
So, as far as PDGs are concerned, this paper can be interpreted in two ways: as an alternative, more restrictive but causally grounded semantics for all PDGs, 
or as a case study on just those PDGs in the image of the encoding (under their original semantics).
Either way, the concept has implications far beyond PDGs.
}

Saying that one approach to
qualitative graphical modeling has
connections to so many different notions is a rather bold claim.  We spend the rest of the paper justifying it.

\commentout{
To summarize, our semantics for qualitative PDGs gives us a powerful modeling
tool to capture dependencies and independencies; when combined with the
quantitative information in a full PDG, we can go well beyond standard
graphical models.  
An obvious question is how difficult inference is for such PDGs. 
Although our results show that 
    \scibility\ 
    be viewed as a special case of 
    PDG semantics,
    the inference algorithm provided by
    \citet{pdg-infer}, unfortunately, does not apply.
\commentout{
We describe an alternative approach in \cref{sec:null}
    that works well for small PDGs, 
    but we do not have a proof that it succeeds in all cases, 
    and it scales very poorly.
}
}

\section{Qualitative Independent-Mechanism (QIM) Compatibility}
    \label{sec:scompat}

In this section, we present the central definition of our paper: a way of making precise
Pearl's notion of ``independent mechanisms'', used to motivate Bayesian Networks from a causal perspective. 
\commentout{%
When specialized to hypergraphs $\Ar_G$ that come from directed acyclic graphs (\emph{dag}s), 
we get an intuitive characterization of independencies in BNs that is equivalent to, yet quite different from, the standard characterization.  
}%
\Citet[p.22]{pearl2009causality}
states that
    \textit{%
    ``each parent-child relationship in a causal Bayesian network represents a
    stable and autonomous physical mechanism.''
    }%
But, technically speaking,
a parent-child relationship only partially describes
the mechanism.  Instead, the autonomous mechanism that determines the child is really represented by that child's joint relationship with all its parents.
So, the qualitative aspect of a mechanism is best represented as a directed \emph{\hyperarc}
\citep{gallo-dirhypergraphs1993}
\unskip, that can have multiple sources.
\begin{defn}
    A \emph{directed hypergraph}
    (or simply a {\hgraph}, since all our \hgraph s will be directed)
    consists of a set $\N$ of nodes
    and a set $\Ar$ of 
    directed hyperedges, or \emph{\hyperarc s};
    each \hyperarc\  $a \in \Ar$ is associated with
        a set $\Src a \subseteq \N$ of
    source nodes and a set $\Tgt a \subseteq \N$ of target nodes.    %
    We write $\ed {\scriptstyle a}{S}{T} \in \Ar$ to specify a
    \arc\ $a \in \Ar$ together with its sources $S = \Src a$ and targets $T = \Tgt a$.
Nodes that are neither a source nor a target of any \hyperarc\ will
seldom have any effect on our constructions; the other nodes can
be recovered from the \hyperarc s (by selecting $\N := \bigcup_{a \in
    \Ar} \Src a \cup \Tgt a$). Thus, we often leave $\N$ implicit,
referring to the \hgraph\ simply as $\Ar$.
\commentout{
    We can assume without loss of generality that $\Ar$ contains an identity \arc\ 
    $\ed {\mathrm{id}}NN$
    for each $N \in \N$, 
    in which case $\Ar$ alone contains all the information in the hypergraph.
    For this reason, 
    
    we refer to the \dhygraph\ simply as $\Ar$.
}%
\commentout{
    We sometimes abuse notation and ignore the set of variables and their
    values, referrting to the the \dhygraph\ simply as $\Ar$.
}%
    \expandafter\commentout\vjoe{
    A \emph{directed hypergraph}
    consists of a set $\N$ of nodes
    and a set $\Ar$ of \emph{directed hyperedges}.
    Each node in $\N$ is associated with a variable, just as in other graphical models. 
    A directed hyperedge is a generalization of a directed edge in a
    directed graph; each directed hyperedge  $a \in \Ar$ is associated with
    a set $\Src a \subseteq \N$ of
    source nodes and a set $\Tgt a \subseteq \N$ of target
    nodes.
    (A directed edge in a directed graph can be thought of
    as having a source and targe that are singletons.)
                We write $\ed {\scriptstyle a}{S}{T} \in \Ar$ to specify
    a directed hyperedge
                $a \in \Ar$ together with its sources $S = \Src a$ and targets $T = \Tgt a$.
    We assume that every node in $\N$ is in the source or target of some
    directed hyperedge.  
    This means that we can recover $\N$ from $\Ar$, so
    we often surpress $\N$, referring to the
    \dhygraph\ simply as $\Ar$.
    From here on, for ease of exposition, we refer to  a directed
    hypergraph as just a 
    hypergraph, and a directed hyperedge as a \emph{\hyperarc}.  We
    also identify a node $X \in \N$ with the variable associated with
    it, so talk about a variable $X \in \N$.
    }
\end{defn}

\commentout{
\vjoe{
In  typical graphical models, a node is associated not just with a
variable $X$, but a set $V(X)$ of possible values of X.  We
deliberately do not do that here.  It makes perfect sense to talk
about a variable $X$ being independent of a variable $Y$ without
specifying the set of possible values of $X$ and $Y$.  However, when we talk
about a distribution $\mu$ on a set of variables, we assume that the
the variables are associated with values.  For example, we talk about
a distribution $\mu$ on $\X = (\N,V)$ for some choice of $V$ mapping each
variable $X \in \N$ to a set $V(X)$ of values.
}
\voli{
Because the
nodes that are relevant for our purposes  
be recovered from the \hyperarc s (by selecting $\N := \bigcup_{a \in \Ar} \Src a \cup \Tgt a$), we often surpress $\N$, referring to the \dhygraph\ simply as $\Ar$.
A \dhygraph---henceforth, simply a \hgraph, since it's the only kind we discuss---can be viewed as 
    an \emph{(unweighted) qualitative PDG} \citep{pdg-infer}
    by identifying its nodes with variable names.
To \emph{fully} interpret a node $X \in \N$ as a variable, it must
be
associated with a (finite) set $\V(X)$
of possible values.
Thus, a set of variables $\X$ is equivalent to 
    the pair $(\N, \V \colon \N \to \mathbf{Sets})$.
\commentout{
    We draw a distinction between $\N$ and $\X$ to emphasize that $\X$ can be decomposed into its qualitative part $\N$ the set of variable symbols, and the map $\V$ which gives a quantiative interpretation of each variable as a set of possible values.
    It is convenient (and common practice) to conflate a variable name $X \in \N$ with  its corresponding variable $(X, \V\!X ) \in \X$.
}%
\emph{Structural} properties of a joint distribution $\mu(\X)$, however,
can depend only on $\N$, not on $\V$.
With that in mind, we use the same symbol 
    for a variable and its name,
    and abuse terminology by calling both objects ``variables'' when no ambiguity can result.
For example, 
    ``$X$ determines the variable $Y$''
    is meaningful across distributions in 
        which $V(Y)$ can vary, 
    because ``$Y$'' here refers only to a variable name.
}
}%
\commentout{Following the graphical models literature, we are interested in \hgraph s whose nodes correspond to variables \unskip---or, more precisely, where $\N$ is a set of variable \emph{names}. It is standard to take $\N$ to be a set $\X$ of variables, each of which implicitly comes with a  (measurable) set $\V(X)$ of possible values; we deliberately do not do that here. 
It makes perfect sense to say that $X$ is independent of $Y$ without specifying possible values of $X$ and $Y$. Of course, when we talk concretely about a distribution $\mu$ on a set of variables $\X \cong (\N, \V)$, those variables must have possible values. The key point is that \emph{qualitative} properties of $\mu$, such as (in)dependence are expressible purely in terms of $\N$, not $\V$. }
Following the graphical models literature, we are interested in \hgraph s whose nodes represent variables,
so that each $X \in \N$ will ultimately be associated with a
(for simplicity, finite)
set $\V(X)$ of possible values.
However, one should not think of $\V$ as part of the information carried by the hypergraph. 
{%
It makes perfect sense to say that $X$ and $Y$ are independent without specifying the possible values of $X$ and $Y$.
}%
Of course, when we talk concretely about a distribution $\mu$ on a set of variables $\X \cong (\N, \V)$, those variables must have possible values
\unskip---but the \emph{qualitative} properties of $\mu$, such as independence, can be expressed purely in terms of $\N$, without reference to $\V$. 

Intuitively, we expect a joint distribution $\mu(\X)$ to be 
qualitatively
compatible with a {s}et of {i}ndependent {m}echanisms
(whose structure is given by a hypergraph $\Ar$)
if 
there is
a mechanistic explanation of how each target
arises as a function of the variable(s) on which it depends
and independent random noise.
This is made precise by the following definition.

\begin{defn}[\scibility]
        \label{defn:scompat}
    Let $\X$ and $\mathcal Y$ be (possibly identical) sets of variables, and
    $\Ar = \{ \ed{a}{\Src a}{\Tgt a} \}_{a \in \Ar}$ be a hypergraph with 
    nodes $\X$.    
    We say a distribution $\mu(\mathcal Y)$ is
    \emph{%
        qualitatively
        independent-mechanism compatible}, or (\SQIM-)compatible,
    with $\Ar$
    (symbolically: $\mu \models \Diamond \Ar$)
    iff
    there exists an extended distribution
    $\bar\mu(\mathcal Y \cup \X \cup \,\U_{\Ar})$
    of $\mu(\mathcal Y)$ to 
    $\X$ and
    to\, $\U_{\Ar} = \{ U_a \}_{a \in \Ar}$,
    an additional set of ``noise'' variables 
    (one variable per \arc)
    according to which:
\begin{enumerate}[label=(\alph*),itemsep=0pt,topsep=0.0ex,parsep=0.5ex]
\item 
    the variables $\mathcal Y$ are distributed according to $\mu$
    \hfill(i.e., $\bar\mu(\mathcal Y) = \mu(\mathcal Y)$),
\item
the variables\, $\U_{\Ar} $ are mutually independent
    \hfill (i.e., $\bar\mu(\U_{\Ar} ) = \prod_{a \in \Ar} \bar\mu(U_a)$ ),
    and
\item the target variable(s) $\Tgt a$ of each \arc\ $a \in \Ar$ are\\ determined
    by $U_a$ and the source variable(s) $\Src a$
    \hfill (i.e., $\forall a \in \Ar.~ \bar\mu \models (\Src a, U_a) \tto \Tgt a$)
    \unskip.%
\end{enumerate}

We call such a distribution
$\bar\mu(\X \cup \mathcal Y \cup\, \U_{\!\Ar})$
 a \emph{witness}
    that $\mu$ is \scible\ with $\Ar$.
\end{defn}

While \cref{defn:scompat} requires the noise variables $\{ U_a \}_{a \in \Ar}$ to be independent of one another, note that they need not be independent of any variables in $\X$.
In particular, $U_a$ may not be independent of $\Src a$,
and so 
the situation can diverge from what one would expect from a randomized algorithm, whose
randomness $U$ is assumed to be independent of its input 
$S$. 
Furthermore, the variables in $\U$ may not be independent of one another conditional on the value of some $X \in \X$.

\begin{example}
    $\mu(X,Y)$ is
    \cible\ with
    $\Ar = \{ \ed 1\emptyset{  \{X\}}, \ed2{\emptyset} {\{Y\}} \}$     
    (depicted in PDG notation as
    $\smash{
    \begin{tikzpicture}[center base]
        \node[dpadinline] (X) {$X$};
        \node[dpadinline,right=0.5em of X] (Y) {$Y$};
        \draw[arr1, <-] (X) -- node[above,inner sep=2pt]{} +(-.7,0);
        \draw[arr1, <-] (Y) -- node[above,inner sep=2pt]{} +(.7,0);
    \end{tikzpicture}
    }$
    ) iff
     $X$ and $Y$ are independent, i.e., 
    $\mu(X,Y) = \mu(X)\mu(Y)$.
    For if $U_1$ and $U_2$ are independent and respectively determine
    $X$ and $Y$, then $X$ and $Y$ must also be independent.%
\end{example}

This is a simple illustration of a more general
phenomenon: when $\Ar$ describes the structure of a
Bayesian Network (BN), then 
\scibility\ with $\Ar$ coincides with satisfying the independencies
    of that BN (which are given, equivalently, by the \emph{ordered Markov properties} \citep{lauritzen-dag-indeps}, 
    \emph{factoring} as a product of probability tables,
    or \emph{d-separation} \citep{geiger-pearl-d-separation}).
To state the general result (\cref{theorem:bns}),
    we must first 
    clarify how the graphs of standard graphical and causal models
    give rise to \dhygraph s.

    \label{sec:bns}
Suppose that
$G = (V,E)$ is a graph, whose edges may be directed or undirected.
Given a vertex $u \in V$,
write $\infl_G(u) := \{ v : (v,u) \in E\}$
    for the set of vertices that can ``influence'' $u$.
There is a natural way to interpret the graph $G$ as giving rise to a set of mechanisms:
    one for each variable $u$, 
    which determines the value of $u$ based
    the values of the variables on which $u$ can depend. 
Formally, let
$\Ar_G := 
    \smash{\big\{} 
    ~{\Pa}_G(u) \ed{u}{~}{~} \{ u \} ~
    \smash{\big\}_{u \in V}}$
be the \hgraph\ \emph{corresponding} to the graph $G$.
\commentout{
The \hgraph s generated by \eqref{defn:graph-to-hypergraph}
    form a special class, and share an important property.
    
\begin{defn}
\commentout{
    A \dhygraph\ $(\N,\Ar)$ is called \emph{\partl} if the targets of its \arc s $\Ar$ form a partition of its node set $\N$,
    and \emph{non-\partl} otherwise.
    A \hgraph\ is \emph{\subpartl} if it contains a node 
    that is not a target of any \arc\
    (i.e., $\cup_{a \in \Ar}\Tgt a \subsetneq \N$),
    \emph{\suppartl} if it contains 
        one   that is a target of multiple \arc s
    (i.e., $\Tgt{a} \cap \Tgt{a'} \ne \emptyset$ for $a \ne a' \in \Ar$),
    and \emph{\unipartl} if its \hyperarc\ targets form a singleton partition 
    (i.e., $\cup_{a \in \Ar}\Tgt a \subsetneq \N$).
}%
    A \dhygraph\ $(\N,\Ar)$ is \emph{\subpartl} if the targets of its \arc s $\Ar$ are disjoint sets, and \emph{\partl} if they form a partition of its node set $\N$.
\end{defn}
}

\commentout{
The targets $\{ \Tgt a \}_{a \in \Ar}$ of $\Ar$ form the singleton partition if and only if $\Ar = \Ar_G$ for some graph $G$.
}

\begin{linked}{theorem}{bns}
If $G$ is a directed acyclic graph and $\mathcal I(G)$ consists of the
independencies of its corresponding Bayesian network,
then $\mu \models \Diamond \Ar_G$ if and only if
    $\mu$ satisfies $\mathcal I(G)$. 
\end{linked}

\commentout{
    In \cref{sec:causal}, we will see how
        a witness of compatibility $\mu(\X, \U)$ can be seen
        as a (partially specified) causal model.
    From this perspective, \cref{theorem:bns} is a form of an equivalence
        well-known in the causlity community:
    every acyclic causal model with independent per-variable noise induces a distribution with the independencies of the appropriate Bayesian Network---%
    and, conversely,
    every distribution with those independencies
    arises from such a causal model.
}

\Cref{theorem:bns} shows, for hypergraphs that 
correspond to directed acyclic graphs (dags),
    our definition of
    \cibility\ 
    reduces exactly to the well-understood independencies of 
BNs.
This means that \scibility, a notion based on the independence of
causal mechanisms, 
gives us a very different way of 
characterizing these independencies
\unskip---\unskip
one that can be generalized to a much larger class of graphical models
that includes, for example, cyclic variants \citep{Baier_2022}.
Moreover, \scibility\ can capture properties other than
independence.  
As the next example shows, it can capture determinism.

\commentout{
But not all \hgraph s are of this form.
Among \partl\ \hgraph s, cyclic models \citep{Baier_2022} are
of particular interest. Our definition of \scibility\ 
gives rise to an interesting generalization of BN independencies 
for such \hgraph s, but we will need the tools developed in 
\cref{sec:monotone,sec:info} to better understand what \scibility\ means in these cases.
Non-partitional \hgraph s are also of interest since they allow us to express determinism, as the following example shows.
}

\begin{example}\label{example:two-edge-det} %
If $\Ar = \smash{
    \begin{tikzpicture}[center base]
        \node[dpadinline] (X) {$X$};
        \draw[arr1, <-] (X) -- node[above,inner sep=1.5pt,pos=.65]{$\scriptstyle 1$} +(-.7,0);
        \draw[arr1, <-] (X) -- node[above,inner sep=1.5pt,pos=.65]{$\scriptstyle 2$} +(.7,0);
    \end{tikzpicture}
    }$
 consists of just two \arc s pointing
to a single variable $X$, then a distribution $\mu(X)$ is \scible\ 
with $\Ar$ iff
$\mu$ places all mass on a single value $x \in \V(X)$.
\end{example}

Intuitively, if two independent coins always give the same answer (the value of $X$), then neither coin can be random. 
This simple example shows that we can capture determinism with multiple \arc s pointing to the same variable.
Such hypergraphs do not correspond to graphs; 
    recall that in a BN, two arrows pointing to $X$
    (e.g., $Y \to X$ and $Z \to X$)
    represent a single mechanism by which $X$ is jointly determined
        (by $Y$ and $Z$),
    rather than two distinct mechanisms.
\vfull{
A central thrust of \citeauthor{pdg-aaai}'s original argument for PDGs over BNs is their
    ability to describe two different probabilities describing a single variable,
    such as $\Pr(X|Y)$ and $\Pr(X|Z)$.
The qualitative analogue of that expressiveness is precisely what allows us to capture functional dependence.
}

\commentout{
\begin{linked}{theorem}{func}
    If $\Ar_0$ is a \dhygraph\ and $\Ar := \Ar_0 \sqcup \{ \ed 1XY, \ed2XY \}$ is the same \hgraph\ augmented with two additional \arc s from $X$ to $Y$, then 
    \begin{enumerate}[label={(\alph*)}]
    \item If 
    $\mu \models \Diamond \Ar_0$ and $\mu \models X \tto Y$, then $\mu \models \Diamond \Ar$.
    \item
    If $\Ar_0 = \Ar_G$ for a directed acyclic graph $G$, 
    then
    $\mu \models \Diamond \Ar$ iff $\mu \models \Diamond \Ar_0$ and $\mu \models X \tto Y$. 
        \item 
        $\mu \models \Diamond \Ar_0$ and $\mu \models X\tto Y$
        if and only if $\mu \models \Diamond \Ar_0 \sqcup \{X \to Y\}^n$ for all $n > 0$.
    \end{enumerate}
\end{linked}
}

Given a \hgraph\ $\Ar = (\N, \Ar)$, 
$X, Y \subseteq\N$,  and a natural number $n \ge 0$,
\gdef\ArXY#1{%
        \Ar{\,\sqcup}\genfrac{}{}{0pt}{}{(+#1)}{X{\to}Y}
    }%
    let $\ArXY n$ 
        denote the \hgraph\ that results from augmenting $\Ar$ with $n$
    additional (distinct) \arc s from $X$ to $Y$.
\begin{linked}{theorem}{func}
    \commentout{
    Let $\Ar = (\N, \Ar)$ be any \hgraph, and $\mu$ a distribution over its nodes. 
    Given $X, Y \subseteq \N$ and a natural number $n \ge 0$,
    \gdef\ArXY##1{%
        \Ar{\,\sqcup}\genfrac{}{}{0pt}{}{(+##1)}{X{\to}Y}
        }%
    let $\ArXY n$ denote the \hgraph\ resulting by augmenting $\Ar$ with $n$ additional (distinct) \arc s from $X$ to $Y$. For all $\mu$ and $\Ar$, we have that 
    }
    \begin{enumerate}[label={(\alph*)}, itemsep=0pt,topsep=0pt,parsep=0.4ex]
    \item 
    $\mu \models X {\tto} Y \land \Diamond \Ar$    
    ~~if and only if~~
    $\forall n \ge 0.~ \mu \models \Diamond \ArXY n$ . 
    \item
            if $\Ar = \Ar_G$ for a dag $G$, then
    $\mu \models X {\tto} Y \land \Diamond \Ar$ if and only if
    $\mu \models \Diamond \ArXY 1$.
    
\item 
if $\exists a \in \Ar$ such that $\Src a = \emptyset$ and $X \in \Tgt a$, 
then $\mu \models X {\tto} Y \land \Diamond \Ar$
iff
 $\mu \models \Diamond \ArXY 2$.
\commentout{
    \item When
    $\Ar$ describes an unconditional mechanism for $X$
    (formally, when $\Ar \rightsquigarrow \begin{tikzpicture}[center base]
        \node[dpadinline] (X) {$X$};
        \draw[arr1,<-] (X) to[] +(0.6,0);
    \end{tikzpicture}$, a condition we will define in \cref{sec:monotone}),
    $\mu \models X {\tto} Y \land \Diamond \Ar$
    if and only if
     $\mu \models \Diamond \ArXY 2$.
}
    \end{enumerate}
\end{linked}
\commentout{
Part (c) generalizes \cref{example:two-edge-det}. 
Based only on that example, it may seem unnecessary to ever add more than two arcs, as in part (a). 
Intuitively, if two independent randomized procedures (always) obtain the same value of $Y$ from input $X$, then neither one can be using its randomness, and so $Y$ is a (deterministic) function of $X$.
However, this intuition (based on randomized algorithms) implicitly assumes not only that the randomness $U_1$ and $U_2$ for the two mechanisms are not only independent, but also conditionally independent given $X$. 
Indeed the right hand side of (a) cannot be replaced by any fixed finite number of arcs without similarly strong assumptions on $\Ar$; see \cref{sec:func-counterexamples} for (non-trivial) counterexamples when $\Ar$ comes from a cyclic graph. 
}
Based on the intuition given after \cref{example:two-edge-det}, it may seem unnecessary to ever add more than two parallel \hyperarc s to ensure functional dependence in part (a). 
However, this intuition implicitly assumes that the randomness $U_1$ and $U_2$ of the two mechanisms is independent conditional on $X$,
which may not be the case.
See \cref{sec:func-counterexamples} for counterexamples
\unskip.

Finally, as mentioned above, \scibility\ gives meaning to cyclic structures,
    a topic that we will revisit often in \cref{sec:causal,sec:info}. 
We start with a simple example. 
\begin{example}
        \label{example:xy-cycle}
    Every $\mu(X,Y)$ is \cible\ with
    \begin{tikzpicture}[center base]
        \node[dpadinline] (X) {$X$};
        \node[dpadinline,right=1.2em of X] (Y) {$Y$};
        \draw[arr1] (X) to[bend left=13] (Y);
        \draw[arr1] (Y) to[bend left=13] (X);
    \end{tikzpicture},
    because every distribution is \cible\ with
    \begin{tikzpicture}[center base]
        \node[dpadinline] (X) {$X$};
        \node[dpadinline,right=1.0em of X] (Y) {$Y$};
        \draw[arr1] (X) to (Y);
        \draw[arr1,<-] (X) to +(-0.68,0);
    \end{tikzpicture},
    and a mechanism with no inputs is a special case of one that can depend on $Y$.
    \qedhere
\end{example}
The logic above is an instance of an important reasoning principle,
    which we develop in \cref{sec:monotone}. 
Although the 2-cycle in \cref{example:xy-cycle} is straightforward, generalizing it even slightly to a 3-cycle raises a not-so-straightforward question, 
    whose answer will turn out to have surprisingly broad implications.
     
\newlength{\cycleboxlen}
\settowidth{\cycleboxlen}{\usebox{\cyclebox}}
\begin{wrapfigure}[5]{o}{0.75\cycleboxlen}
    \vspace{-0.8em}
    \!\begin{tikzpicture}[center base, scale=0.76]
        \node[dpad0] (X) at (0:.8) {$X$};
        \node[dpad0] (Y) at (120:.8) {$Y$};
        \node[dpad0] (Z) at (-120:.8) {$Z$};
        \draw[arr2] (X) to 
            (Y);
        \draw[arr2] (Y) to
            (Z);
        \draw[arr2] (Z) to 
            (X);
    \end{tikzpicture}\!\!
\end{wrapfigure}
\refstepcounter{example}
\label{example:xyz-cycle-1}
\textbf{Example {\theexample}.~} 
    What $\mu(X,Y,Z)$ are \cible\ with the 3-cycle shown, on the right?
    By the reasoning above,
     among them must be all distributions consistent with a linear chain ${\to}X{\to}Y{\to}Z$. Thus,  
    any distribution in which two variables are conditionally independent given the third is compatible with the 3-cycle.
    Are there
    distributions that are \emph{not} compatible with 
    this hypergraph? It is not obvious. We return to this
     in \cref{sec:pdgs}. 
\hfill$\triangle$

Because \scibility\ applies to cyclic structures,  one might wonder if
    it also captures the independencies of undirected models 
    \unskip.
Our definition of $\Ar_G$, as is common, implicitly identifies a undirected edge $A {-} B$ with the pair $\{ A{\to}B, B{\to}A\}$ of directed edges;
in this way, it naturally converts even an \emph{undirected} graph $G$ to a (directed) \hgraph. 
Compatibility with $\Ar_G$, however, does not coincide with any of the standard Markov properties
corresponding to $G$ \citep{koller2009probabilistic}.
This may appear to be a flaw in \cref{defn:scompat},
but it is unavoidable (see \cref{sec:monotone}) if we wish to also capture causality, as we do in the next section.

\commentout{In this section, we have seen how \scibility\ with a single graph can represent many (in)dependencies of interest. But we have only scratched the surface here; in the remaining sections, we will come to understand how \scibility\ works in general: its connection to causality (\cref{sec:causal}), to information theory (\cref{sec:info}), and a simple principle for reasoning about it (\cref{sec:monotone}). 
}

\section{%
    \SCibility\ and Causality
}
    \label{sec:causal}

\commentout{
In this section, we spell out the fundemental connection between causality and \scibility. In brief: $\mu \models \Diamond \Ar$ if and only if it is possible for $\mu$ to arise from a causal model with dependency structure $\Ar$ (\cref{ssec:cm-arise}); furthermore, a witness of \scibility\ can be naturally converted to such a causal model. To make this precise, we must review some standard definitions from causality. }

Recall that in the definition of \Scibility, each hyperarc represents an independent mechanism. 
Equations in a causal model are also viewed as representing independent
mechanisms.  
This suggests a possible connection between the two formalisms, which we now explore.
We will show that \scibility\ with $\Ar$ means exactly that
a distribution can be generated by a causal model with the corresponding dependency structure (\cref{ssec:ssec:cm-arise}). Moreover, such causal models and \scibility\ witnesses are themselves closely related (\cref{sec:witness-to-causal-model}).
In this section, we establish a causal grounding for \scibility.
To do so, we must first review some standard definitions.

\begin{defn}[\citet{pearl2009causality}]
        \label{defn:SEM}
    A \emph{structural equations model} (SEM) is a tuple
    $M = (\U, \enV, \mathcal F)$,
    where
    \begin{itemize}[left=0pt,itemsep=0pt,topsep=0pt,parsep=0.3ex]
    \item $\U$ is a set of exogenous variables;
    \item $\enV$ is a set of endogenous variables (disjoint from $\U$);
    \item $\mathcal F = \{ f_Y \}_{Y \in \enV}$ associates to each endogenous variable $Y$ an \emph{equation}
    $f_Y : \V (\U \cup \enV - Y) \to \V(Y)$
     that determines its value as a function of the other variables. 
\commentout{
    \item $\mathcal I$ is a set of allowable interventions,
    typically taken to be the set of all strings
    of the form $X \gets x$
    for $X \subseteq \enV$
    and $x \in \V(X)$.
}
        \qedhere
    \end{itemize}
\end{defn}

\commentout{
Given a SEM $M = (\U, \enV, \mathcal F)$, and a joint setting $\mat u \in \V(\U)$ of exogenous variables, let
\[
    \mathcal F(\mat u)
    := \Big\{
        \mat x \in \V\! \enV ~\Big|~
            \forall Y \,{\in}\, \enV.~
            f_{Y}\big(\mat u,~ \mat x[ \enV {-} Y ]\big) = \mat x[Y]
    \Big\}
\]
be the set of joint endogenous variable settings that are consistent with $\mat u$, according to the equations $\mathcal F$.
}
In a SEM $M$, a variable $X \in \enV$ \emph{does not depend} on $Y \in \enV\cup\U$ if $f_X(\ldots, y, \ldots) =
f_X(\ldots, y', \ldots)$ for all $y, y' \in \V(Y)$.
Let the parents $\Pa_{M}(X)
$
of $X$ be the set of variables on which $X$ depends.
$M$ is \emph{acyclic}
iff $\Pa_M(X) \cap \enV = \Pa_G(X)$ for some dag $G$ with vertices $\enV$. 
\commentout{
    It is easy to see that,
    by determining the value of each exogenous variable according to this order,
    $\mathcal F(\mat u)$ is a singleton for acyclic models, in which case
    $\U$ determine $\enV$ via a function $f_{\enV} : \V (\U) \to \V(\enV)$.
}
In an acyclic SEM, it is easy to see that 
a setting of the exogneous variables determines the
values of the endogenous variables
(symbolically: $M \models \U \tto \enV$).
A \emph{probabilistic SEM} (PSEM)
$\mathcal M = (M, P)$ 
is a SEM,
together with a probability 
$P$ over the exogenous variables.
\commentout{
If $\cal M$ is acyclic, then it
    determines a 
    distribution over the endogenous variables $\enV$,
    by the pushforward of $P(\U)$ through
    $f_{\enV}$.
}%
\commentout{
    In this case, $P$ extends uniquely to a joint distribution
    $\nu(\U, \enV) = P(\U) \delta\! f_{\enV}(\enV \mid \U)$, where the notation $\delta g(B|A)$ describes the (deterministic) cpd
    that corresponds to the function $g : \V (A) \to \V (B)$.
}%
{%
    When $\mathcal M \models \U \tto \enV$ 
    \vfull{%
    (such as when $M$ is acyclic)%
    },
    the distribution
    $P(\U)$
    extends uniquely to a distribution over $\V(\enV \cup \U)$.    
}%
A cylic PSEM, however, may induce more than one such distribution, or none at all.
In general, a PSEM $\cal M$
induces a (possiby empty) convex set of distributions over $\V(\U \cup \enV)$.
This set is defined by two (linear) constraints:
    the equations $\mathcal F$ must hold with probability 1, 
    and
    \vfull{, in the case of a PSEM,}
    the marginal probability over $\U$ must
    equal $P$. 
\vfull{
Formally, for a PSEM $\mathcal M = (M, P)$, define
$\SD{\mathcal M} := $
\vspace{-1.5ex}
\[
    \bigg\{
       \nu {\,\in\,} \Delta\!\V(\enV \cup \, \U)
            \,\bigg|\!
        \begin{array}{l}
            \forall Y \!\in\! \enV.~~
            \nu\Big( f_Y(\mathcal U,
                \enV {-} Y
            ) \,{=}\, Y\Big) \,{=}\, 1
            ,
            \\[0.2ex]
            \nu(\U) = P(\U)
        \end{array}\!\!
        \bigg\}
\vspace{-1ex}
\]
and define $\SD{M}$ for an ``ordinary'' SEM $M$ in the same way,
    except without the constraint involving $P$.
To unpack the other constraint,
$f_{Y}(\U, \enV - Y)$
is a random variable
on the outcome space $\V(\enV,\U)$,
and that it
has the same value as $Y$ is an event
which, according to the equation $f_Y$,
must always occur.
}%
Given a 
PSEM
 $\mathcal M$, let $\SD{\mathcal M}$ consist of all 
joint distributions $\nu(\U, \enV)$ that satisfy the two constraints above
\vfull{%
(or just the first of them, in the case of a non-probabilistic SEM).
$\SD{\cal M}$ can be thought of as the set of distributions compatible
wth $\cal M$. It
}%
\unskip; this set
captures the behavior of $\cal M$ in
the absence of interventions.  
A joint distribution $\mu(\mat X)$ over $\mat X \subseteq \enV \cup \U$ \emph{can arise from} a (P)SEM $\mathcal M$ iff there 
is
some $\nu \in \SD{\mathcal M}$ whose marginal on $\mat X$ is $\mu$.
\commentout{
We remark that a (P)SEMs $\cal M$ is naturally a special case of a PDG, and that $\SD{\mathcal M}$ is the set-of-distributions semantics of that PDG. 
}
\commentout{
    For a SEM $M$, we take $\SD{M}$ to consist of all 
    distributions $\nu \in \SD{M,P}$ for some distribution $P$ on the exogenous variables of $M$; 
    equivalently
    $\SD{M}$ consists of all
    distributions that satisfy
    the first of the two constraints above.
}

\commentout{
    The connection is particularly important in the other direction: witnesses of \scibility\ can be converted to randomized causal models,
    as we now show.
}

We now review the syntax of a language for describing causality. 
A \emph{basic causal formula} is one of the form $[\mat Y {\gets}
    \mat y]\varphi$, where $\varphi$ is a Boolean expression over the
endogenous variables $\enV$, $\mat Y \subseteq \enV$ is a subset of them, and $\mat y \in \V(\mat Y)$.
The language
then consists of all Boolean combinations 
 of basic formulas.
In a causal model $M$ and context $\mat u \in \V(\U)$, 
a Boolean expression $\varphi$ over $\enV$ is true iff it holds
for all $(\mat u, \mat x) \in \V(\U,\enV)$ consistent with the equations of $M$.  
Basic causal formulas are then given semantics by
$
    (M, \mat u) \models [\mat Y {\gets} \mat y]\varphi
$ iff $
    (M_{
    \mat Y \gets \mat y}, \mat u) \models \varphi,
$
where $M_{\mat Y \gets \mat y}$ is the result of changing each $f_Y$, for $Y \in \mat Y$, to
the constant function $\mat s \mapsto \mat y[Y]$, 
which returns (on all inputs $\mat s$) the value of $Y$ in the joint setting $\mat y$.
\vfull{%
From here, the truth relation can be extended to arbitrary causal formualas by structural induction in the usual way.
    \unskip\footnote{
    $M\! \models\! \varphi_1 \land \varphi_2$ iff $M \!\models\! \varphi_1$ and $M \!\models\! \varphi_2$;
    $M \!\models \lnot \varphi$ iff $M \!\not\models \varphi$. 
    }
}%
The dual formula
$\langle \mat Y {\gets} \mat y\rangle \varphi := \lnot [\mat Y{\gets }\mat y]\lnot \varphi$
is equivalent 
to $[\mat Y {\gets} \mat y]\varphi$
in SEMs where each context $\mat u$ induces a unique setting of the endogenous variables \cite{halpern-2000}.
A PSEM $\mathcal M = (M, P)$ assigns probabilities to causal formulas according to $\Pr_{\mathcal M}(\varphi) := P(\{ \mat u \in \V(\U) : (M, \mat u) \models \varphi \})$.

\commentout{
\subsection{Capturing Causal Model Structures with \DHygraph s}
One consequence of \cref{theorem:func} is that \dhygraph s
can capture the structural aspects of causal models
(which are functional dependencies). 
\commentout{
More precisely, given a SEM $M = (\enV, \U, \mathcal F)$,  
define the parents
$\Pa_{M}(X) $
 of endogenous variable $X \in \enV$ to be
the set of variables 
on which $X$ depends 
 (recall that the notion of $X$ not depending on $Y$ was defined above).
}%
\commentout{
can affect $X$ directly, i.e.,
$\Pa_{M}(X) := $
\vspace{-1.6ex}
\[
    \bigg\{
        Y \in \U \cup \enV
     ~\bigg|~
        \begin{array}{@{}l@{}}
            \exists \mat z \in \V(\U {\cup }\enV - \{X,Y\}), 
            \exists y,y' \in \V(Y).\\ 
            f_X(\mat z, y) \ne f_X(\mat z, y')
        \end{array}
    \bigg\}.
\]
}%
More precisely, given a SEM $M$,
we need two parallel \arc s with the dependency structure each equation, to get a \hgraph\ 
$
    \Ar_{M} := 
    \big\{
        \Pa_{M}(X) 
        \mathrel{\raisebox{-2pt}{$\xrightarrow{(\iota,X)}$}}
        X
    \big\}_{(\iota,X) \in \{0,1\} \times \enV}
$
representing precisely the structural information in $M$.  
What this notation is saying is that for each node $X$ in the causal
model $M$, there is a pair of hyperarcs in $\Ar_M$, each one going
from the parents of $X$ to $X$.

\begin{linked}{prop}{causal-model-structure-capture}
    If $M$ is a SEM, then
    $\mu(\enV,\U) \models \Ar_{M}$ 
        iff
        there exists $M'$ such that $\Pa_M = \Pa_{M'}$ and
        $\mu \in \SD{M'}$.
\end{linked}

\vfull{
$\Ar_M$ is \subpartl, because the exogenous variables $\U$ are not targets of any mechanism. 
A PSEM $\mathcal M = (M, P)$, however,
    augments $M$ with a distribution $P$,
which can be viewed as a randomized mechanism 
    explaining where the variables $\U$ come from;
    thus $\Ar_{\cal M} := \Ar_M \cup \{ \emptyset \to \U \}$
    is \partl.
Perhaps surprisingly, \Scibility\ with $\Ar_M$ and $\Ar_{\cal M}$
    pick out the same set of distributions
    (i.e., $\mu \models \Ar_M$ iff $\mu \models \Ar_{\mathcal M}$).
But that does not make the two structures interchangable,
    as we will see in \cref{sec:equivalence}.
}    

}

Some authors assume that 
for each variable $X$, there is a special
``independent noise'' exogenous variable 
$U_X$
\vfull{%
(often written $\epsilon_X$ in the literature)
}%
on which only the equation $f_X$ can depend;
we call a PSEM $(M, P)$
\emph{randomized} if it contains 
such exogenous variables that are mutually independent according to $P$
{%
    \unskip, and \emph{fully randomized} if {all} its exogenous variables are of this form
}%
\unskip.
\commentout{
but one that turns out to be important in
presenting our results. Given a PSEM $\cal M$, there is a randomized
PSEM $\cal M'$ such that $\cal M$ and $\cal M'$ agree on the
probability of all forulas in the causal language.  [[Perhaps add one
    sentence about the construction]]  Thus, in a sense, we do not
lose expressive power if we use randomized PSEMs.  
}%
\commentout{
Randomized PSEMs are a special class of PSEMs%
---yet, at the same time, may also be viewed as a generalization of a PSEM in which each equation is not deterministic but a randomized algorithm, taking as input an independent random seed. 
}%
Randomized PSEMs are clearly a special class of PSEMs,
\commentout{\unskip---yet, at the same time, correspond precisely to a
``generalization'' of PSEMs whose equations need not be deterministic.}%
but note also that every
PSEM can be
converted to an equivalent randomized PSEM by extending it 
with additional dummy variables $\{U_X\}_{X \in \enV}$ that can
take only a single value. 
Thus, we do not lose expressive power by using randomized PSEMs
\unskip. In fact, \emph{qualitatively}, 
randomized PSEMs are more expressive: they can encode independence.
\vfull{%
It should come as no surprise that randomized PSEMs and \scibility\ are related.
}

\subsection{The Equivalence Between \SCibility\ and Randomized PSEMs}
    \label{ssec:ssec:cm-arise}
\commentout{In this section, we give the top-level contour of the relationship between \scibility\ and causality: $\mu \models \Ar$ iff there is a causal model whose dependency structure is $\Ar$ that gives rise to $\mu$. To do this, we must be more precise about what it means to ``have dependency structure $\Ar$''. We do this by considering increasingly general classes of \hgraph s, which turn out to coincide with special classes of causal models. }

We are now equipped to formally describe the connection between \scibility\ and causality.
At a high level, this connection should be unsurprising:
witnesses and causal models both relate dependency structures to distributions, 
but in ``opposite directions''.
\Scibility\ starts with distributions and asks what dependency structures they are compatible with. Causal models, on the other hand, are explicit (quantitative) representations of dependency structures that give rise to sets of distributions.
We now show that the existence of a causal model coincides with the existence of a witness. 
We start by showing this for the hypergraphs generated by graphs (like Bayesian
networks, except possibly cyclic), which we show
correspond to
fully randomized causal models (\cref{prop:sc-graph-arise}).  
\commentout{ We then expand the set of hypergraphs we consider by allowing variables not to be the target of any mechanism (i.e., with additional exogenous variables); these are the dependency structures of PSEMs (\cref{prop:sc-psem-corresp}).}
We then give a natural generalization of a causal model that exactly captures \scibility\ with an arbitrary hypergraph (\cref{prop:gen-sim-compat-means-arise}). 
In both cases,
the high-level result is the same: $\mu \models \Ar$ iff there is a 
causal model that ``has dependency structure $\Ar$''
that gives rise to $\mu$. 

More precisely, a randomized causal model $\cal M$
\emph{has dependency structure $\Ar$}
iff there is a 1-1 correspondence between $a \in \Ar$ and the equations of $\cal M$, such that the equation $f_a$ produces a value of $\Tgt a$ and depends only on $\Src a$ and $U_a$. 
The definition above emphasizes the hypergraph; 
for readers interested in causality,
here is an equivalent one that emphasizes the causal model:
$\cal M$ is of dependency structure $\Ar$ iff 
the targets of $\Ar$ are disjoint singletons corresponding to the elements of $\enV$ (so $\Ar = \{ \Src Y \to \{ Y \} \}_{Y \in \enV}$),
and $\Pa_{\cal M}(Y) \subseteq \Src Y \cup \{ U_Y \}$ for all $Y \in \enV$. 
We start by presenting the result in 
the case where $\Ar$ corresponds to a directed graph. 

\commentout{
We start with an important special case:
\scibility\ with a \hgraph\ $\Ar_G$ coming from a graph $G$
means precisely that a distribution can arise from a fully randomized causal model in which each variable depends only on its parents and random noise.
}

\commentout{
\begin{linked}{prop}{sc-graph-arise}
    If $G$ is a graph whose nodes correspond to variables $\X$,
    and $\mu$ is a joint distribution over $\X$, then
    $\mu \models \Diamond \Ar_G$ iff there exists a randomized PSEM
    $(M, P)$
    from which $\mu$ can arise
    \unskip,
    satisfying $\Pa_M(Y) \subseteq \Pa_G(Y) \cup \{ U_Y \}$ for all $Y \in \X$.
\end{linked}}
\begin{linked}
        {prop}{sc-graph-arise}
    \commentout{
        If $G$ is a graph whose nodes correspond to variables $\enV$,
        then $\mu \models \Diamond \Ar_G$ iff there exists a fully randomized PSEM $\mathcal M = (\{U_Y\}_{Y \in \enV}, \enV, \mathcal F, P)$ from which $\mu$ can arise, satisfying $\Pa_M(Y) \subseteq \Pa_G(Y) \cup \{ U_Y \}$ for all $Y \in \enV$.
    }
    Given a graph $G$ and a distribution $\mu$, 
    $\mu \models \Diamond \Ar_G$ iff there exists a fully randomized PSEM
    of
    dependency structure $\Ar_G$ from which $\mu$ can arise.
\end{linked}

In other words, compatibility with a hypergraph corresponding to a graph
means arising from a fully randomized PSEM of the appropriate dependency structure.
In light of this, \cref{theorem:bns} can be viewed as formalizing
a phenomenon that seems to be
almost universally implicitly understood:
every acyclic fully randomized SEM
induces a distribution with the independencies of the corresponding
Bayesian Network.  Conversely,
every distribution with those independencies
arises from such a causal model.
    Both halves have been recognized before.
    \Citet[Theorem 1]{Druzdzel1993CausalityIB} arguably establish one direction of the correspondence (turning a BN into a causal model), but their statement of the result obscures the possibility of a converse.%
        \footnote{Indeed, \citet{Druzdzel1993CausalityIB} state that ``a causal structure does not necessarily imply independences'', suggesting that they did not realize that their result could be used to characterize BN independencies.}
    Pearl's \emph{causal Markov condition} \cite[Theorem 1]{pearl2009causaloverview}, on the other hand, is closely related to that converse 
        (as will be made explicit by our \cref{prop:sc-graph-arise}).
    Yet, to the best of our knowledge, the two results have not before 
    been combined and recognized as an equivalent characterization  
    of a BN's conditional independencies.

\commentout{
What about causal models that are not fully randomized, and \hgraph s that are not induced by graphs?

\begin{prop} \label{prop:sc-psem-corresp}
     If $(\N, \Ar)$ is a \hgraph\ whose targets are distinct singletons, then $\mu \models \Diamond \Ar$ iff $\mu$ can arise from a randomized PSEM of dependency structure $\Ar$. 
 \end{prop}
 \commentout{%
     \begin{prop}
         If $\Ar = \Ar_{\cal M}$ for some PSEM $\cal M$, then $\mu \models \Ar$ iff $\mu$ can arise from a (possibly different) randomized PSEM of the same dependency structure $\Ar$. 
     \end{prop}
    }%

This is just \cref{prop:sc-graph-arise} 
for the graph of the PSEM's ``\emph{depends on}'' relation,
except for a small difference: since the exogenous variables $\U$ have no equations, they are not associated with \arc s.
}

Like before, \scibility\ allows us to go much futher.
It is easy to extend 
\cref{prop:sc-graph-arise}
to the dependency structures of all randomized PSEMs.  But what happens if $\Ar$ contains
\arc s with overlapping targets?
Here the correspondence starts to break down for a simple reason:
by definition, there is at most one equation per variable in a
(P)SEM;
thus, no PSEM can have dependency structure $\Ar$. 
Nevertheless, 
    the correspondence between witnesses and causal models persists if we simply drop the
(traditional)
requirement that $\mathcal F$ is indexed by $\enV$.
This leads us to consider
a natural generalization of a 
(randomized)
PSEM that has an arbitrary set of equations---not just one per variable. 

\gdef\GRPSEM{generalized randomized PSEM}
\begin{defn}\label{defn:GRPSEM}
\commentout{
  A \emph{generalized PSEM} is a tuple $M = ((\U,\enV, \mathcal F), P)$,
  where $U$, $\enV$, and $P$ are just as in a PSEM, and $\mathcal F$
  is [[PLEASE FILL IN]].   $\SD{\dg M}$ and \emph{can arise} are
  defined for generalized PSEMs just as for PSEMs.
  A generalized PSEM $M = ((\U,\enV,{\mathcal F}),P)$
  \emph{corresponds} to a \hgraph   $(\N,\Ar)$ if (a) $M$ is
  randomized,
  (b) $\U = \{ U_a \}_{a \in \Ar}$, (c) $\X = \N$, and (d)
      $ \mathcal F = \{
  f_a : \V(\Src a) \times \V(U_a) \to \V(\Tgt a)    \}$.
}
    Let $(\N, \Ar)$ be a \hgraph. 
    A \emph{generalized randomized PSEM}
     $\mathcal M = (\X, \U, \mathcal F, P)$
    \emph{with
    structure
    $\Ar$} consists of
    sets of variables $\X$ 
    and $\U = \{ U_a \}_{a \in \Ar}$, 
    together with a set of functions
    $ \mathcal F \!=\! \{
            f_a : \V(\Src a) \times \V(U_a) \to \V(\Tgt a)
    \}_{a \in \Ar},
    $
    and a probability $P_a$ over each independent noise variable $U_a$.
    \commentout{
        Let $\SD{\mathcal M}$ consist of all joint distributions $\nu(\X,\U)$ satisfying the equations with probability 1, and whose marginal on $\U$ is $\nu(\U) = \prod_{a \in \Ar} P_a(U_a)$. 
        Just like in an (ordinary) SEM, 
        a joint distribution $\mu(\mat X)$ over $\mat X \subseteq \X \cup \U$ \emph{can arise from} $\mathcal M$ iff there exists some $\nu \in \SD{\mathcal M}$ whose marginal on $\mat X$ is $\mu$.
    }%
    The meanings of $\SD{\cal M}$ and \emph{can arise} are the same as for a PSEM.
\end{defn}
\commentout{
If $\Ar = \Ar_G$ is a directed graph, a
\GRPSEM\ 
of shape $\Ar_G$ is 
a
fully randomized PSEM with endogenous variables represented by the nodes of $G$;
\commentout{In particular, if $G$ is a dag, then a GRPSEM of shape $\Ar_G$ is a causal Bayesian Network with underlying graph $G$. }
if the targets of $\Ar$ are disjoint singletons, 
then a \GRPSEM\ of shape $\Ar$ is a randomized PSEM 
\commentout{
    with the same sets of endogenous and exogenous variables,
    in which the equation $f_Y$ can depend only on the variables on which it depends in $M$ (and also independent noise $U_a$).
    }%
of dependency structure $\Ar$.
\GRPSEM s have a simple and direct relationship with \scibility\ with an arbitrary \hgraph.
}

\begin{linked}{prop}{gen-sim-compat-means-arise}
    $\mu \models \Diamond \Ar$ iff
    there exists 
    a \GRPSEM\
    with
     structure $\Ar$\,
    from which $\mu$ can arise.
    \vspace{-1ex}
\end{linked}

\commentout{
We now describe
    the same connection from another perspective:
    like causal models, witnesses of \scibility\ can be used to ascribe probabilities to causal formalas.
    A basic causal formula $[\mat Y {\gets} \mat y]\varphi$ is true in 
    $\nu$ for a context $\mat u \in \V(\U)$, written
    $(\nu, \mat u) \models [\mat Y {\gets} \mat y]\varphi$, iff
    \[
        \nu(\varphi \mid \U{=}\mat u[\hat U_{\mat Y} \gets 
            \Lambda Y. \lambda\textunderscore.\mat y[Y]]) = 1,
    \]
    where $\Lambda Y.\lambda \textunderscore.\mat y[Y]$ is just system F syntax for the indexed set of functions (one for each $Y \in \mat Y$) that ignore their inputs and return components of the constant $\mat y$. 

    \begin{linked}{prop}{all-models-formula-equiv}
        Let $\varphi \in \mathcal L(\U,\X)$ be causal formula, and $\nu(\X,\U)$ a witness to $\mu \models \Ar$.
            Then
        \[
        (\nu, \mat u) \models \varphi
        \quad\iff\quad
        \forall \mathcal M \in \PSEMs(\nu).~~
        (\mathcal M, \mat u) \models \varphi.
        \]
    \end{linked}
}

Generalized randomized PSEMs
can capture functional dependencies, and constraints. 
For instance, an equality (say $X = Y$) 
can be encoded in a
\GRPSEM\
with a second equation for $X$.
Indeed, we believe that 
\GRPSEM s
can capture a wide class of constraints, 
and are closely related to
\emph{causal models with constraints} \cite{beckers2023causal},
a discussion we defer to future work.
\commentout{
GRPSEMs, and hence \scibility\ in its most general form, 
are closely related to several existing generalizations of SEMs. 
They can be viewed as causal models with constraints 
\citep{beckers2023causal} by splitting each variable that is a target of multiple \hyperarc s into multiple copies (so that each copy is only a target of a single \hyperarc), and then add a constraint requiring that all copoes take the same value.
They can be viewed as GSEMs \citep{peters-halpern-GSEMs} by selecting the defining function from $\V(\U)$ to subsets of $\V(\X)$ to be the set of states that satisfy the equations in a given context.
}

\subsection{Interventions and the Correspondence Between Witnesses and Causal Models}
    \label{sec:witness-to-causal-model}

\commentout{
The relationship shown in the previous section, culminating in \cref{prop:gen-sim-compat-means-arise}, is simple and exact, but it does not fully capture the depth of the connection between \scibility\ and causality.
One unsatisfying aspect has been the focus the behavior of these causal models in the absence of interventions. 
This is by necessity: 
in the discussion so far, both the witness of \scibility\ $\bar\mu$
(implicit in the assertion $\mu \models \Ar$) and the causal model $\mathcal M$ (implicit in the statement that $\mu$ can arise from some causal model) have been obfuscated behind existential quantifiers. 
Yet $\bar\mu$ and $\mathcal M$, previously left implicit, are themselves closely related. 
A closer look at the details will reveal 
a clean relationship between \scibility\ witnesses and causal models
(\cref{prop:witness-model-properties}),
and, moreover,
extend the correspondence to say something interesting about
    causal formulas with interventions
 (\cref{theorem:condition-intervention-alignment}).
}
We have seen that \scibility\ with $\Ar$ (i.e., the existence of a witness $\bar\mu$) coincides exactly with the existence of a causal model $\cal M$ from which a distribution can arise. But which witnesses correspond to which causal models? The answer to this 
question will be critical to extend the correspondence we have given
so that it can deal with interventions.  
Different causal models may give rise to the
same distribution, yet handle interventions differently.

There are two directions of the correspondence. 
Given a randomized PSEM $\cal M$, distributions arising from it are
\cible\
with its dependency structure, and the corresponding
witnesses are exactly 
the distributions in $\SD{\mathcal M}$ (see \cref{appendix:sem2witness}).
In particular, if $\cal M$ is
acyclic, there is a unique witness.  
The converse is more interesting: how can we turn a witness into a causal model?

\commentout{
\begin{itemize}[wide]
    \item 
    When $\Ar = \Ar_G$ for some is a directed acyclic graph $G$ (i.e., a qualitative BN), then a witness to $\mu \models \Ar$ is
         essentially an acyclic causal model
        in which the equation for $X$ depends on $\Pa_G(X)$ and
        independent noise, and which ultimately gives rise to $\mu$ in the absence of intervention.
        
        From this perspective, \cref{theorem:bns} can be viewed as result 
            that implicitly undergirds much of the work on causality:
        every acyclic causal model with independent per-variable noise induces a distribution with the independencies of the appropriate Bayesian Network---%
        and, conversely,
        every distribution with those independencies
        arises from such a causal model.
    \item
    More generally, if $\Ar$ is a \partl\ hypergraph,
        then a witness to 
        $\mu \models \Ar$ 
        contains a (possibly cyclic) causal model 
        with respect to which $\mu$ is a possible 
        (distribution over) fixed-point solution(s).    
    
    \item 
    Finally, a witness of \scibility\ with an arbitrary \dhygraph\ $\Ar$
    corresponds to a causal model with logical constraints \citep{beckers2023causal}.
\end{itemize}
}

\commentout{
    In the section we show that if $\nu(\U,\X)$ is a witness to a distribution
    $\mu(\X)$ be \scibile\ with a \dhygraph $(\X,\Ar)$, then there is a
    unique causal model $M$ with endogenous variables $\X$ and exogenous
    variables $\U$ such that \ldots.
}%
\commentout{\color{gray!80}
\paragraph{Causal models from \scibility\ witnesses.}
Suppose that $\mu(\X)$ is \cible\ with a \hgraph\ $\Ar$, 
and for the sake of more closely matching the standard causality literature, assume that $\Ar$ is \subpartl. 
A witness $\bar\mu(\X,\U_{\Ar})$ 
to this \scibility\ 
can be converted to a 
randomized SEM
with endogenous variables $\enV := \cup_{a \in \Ar} \Tgt a$ and exogenous variables $\U := \U_{\Ar} \cup (\X{-}\enV)$
\unskip, as follows.
For the distriubution $P(\U)$ over the exogenous variables, 
    select the marginal $\bar\mu(\U)$. 
\commentout{
The equations, for the most part, fall out of \cref{defn:scompat} (c). 
For each variable $X \in \enV$, use the following procedure to define $f_X$.}%
For all $X \in \enV$, define the equation $f_X$ as follows.
Because of the definition of $\enV$ and the fact that $\Ar$ is \subpartl, 
    there is a unique
    \arc\ $a_{\!X} \in \Ar$ whose targets $\Tgt {a_{\!X}}$ contain $X$.
Since
$\bar\mu \models ( U_{\!a_{\!X}}, \Src{a_{\!X}}) \tto \Tgt {a_{\!X}}$
(by \cref{defn:scompat}(c)),
$X \in \Tgt {a_{\!X}}$ must also be 
a function of $\Src{a_{\!X}}$ and $U_{\!a_{\!X}}$.
Roughly speaking, we take $f_X$ to be this function. 
More precisely,
for each $u \in \V (U_{\!a_{\!X}})$ and $\mat s \in \V (\Src {a_{\!X}})$
for which $\bar\mu(U_{\!a_{\!X}} {=}\, u, \Src {a_{\!X}} {=}\, \mat s) > 0$,
there is a unique
 $t \in \V (\Tgt {a_{\!X}})$ such that $\bar\mu(u,\mat s, t) > 0$.
In this case, choose $f_X(u, \mat s, \ldots) := t[X]$.
On the other hand, if $\bar\mu(U_{\!a_{\!X}} {=}\, u, \Src {a_{\!X}} {=}\, \mat s) = 0$,
then any choice of $f_X(u, \mat s, \ldots)$ that depends only on $\Src a \cup \{U_a\}$ will do.
Let 
$\PSEMsA(\bar\mu)$
be the set of PSEMs that can can be obtained from $\bar\mu$
    in this way.
Even if $\Ar$ is not \subpartl, our construction of $\PSEMsA(\bar\mu)$ still makes sense, except the models will have more than one equation for some variable and thus be GRPSEMs, rather than causal models in the traditional sense. 
}

\begin{constr}
        \label{constr:PSEMs}
    Given a witness $\bar\mu(\X)$ to \cibility\ with
    a \hgraph\ $\Ar$ with disjoint targets,
    construct a PSEM according to the 
    following (non-deterministic) procedure.  
    Take $\enV := \cup_{a \in \Ar} \Tgt a$,  $\U := \U_{\Ar} \cup (\X{-}\enV)$, and $P(\U) := \bar\mu(\U)$. 
For each $X \in \enV$, there is a unique $a_{\!X} \in \Ar$ whose targets $\Tgt {a_{\!X}}$ contain $X$. 
Since
$\bar\mu \models ( U_{\!a_{\!X}}, \Src{a_{\!X}}) \tto \Tgt {a_{\!X}}$ 
(this is just property (c) in
\cref{defn:scompat}%
 ),
$X \in \Tgt {a_{\!X}}$ must also be 
a function of $\Src{a_{\!X}}$ and $U_{\!a_{\!X}}$;
take $f_X$ to be such a function. 
More precisely,
for each $u \in \V (U_{\!a_{\!X}})$ and $\mat s \in \V (\Src {a_{\!X}})$
for which $\bar\mu(U_{\!a_{\!X}} {=}\, u, \Src {a_{\!X}} {=}\, \mat s) > 0$,
there is a unique
$t \in \V (\Tgt {a_{\!X}})$ such that $\bar\mu(u,\mat s, t) > 0$.
In this case, set $f_X(u, \mat s, \ldots) := t[X]$.
If $\bar\mu(U_{\!a_{\!X}} {=}\, u, \Src {a_{\!X}} {=}\, \mat s) = 0$, 
$f_X(u, \mat s, \ldots)$ can be an arbitrary function of $u$ and $\mat s$.
Let $\PSEMsA(\bar\mu)$ denote the set of PSEMs that can result.
\end{constr}

It's clear from \cref{constr:PSEMs} that
$\PSEMsA(\bar\mu)$ 
is always nonempty, and is a singleton
iff $\bar\mu(u,s) > 0$ for all $(a,u,s) \in \sqcup_{a \in \Ar} \V(U_a, \Src a)$
\unskip.
A witness with this property exists when
$\mu$ is positive (i.e.,
$\mu(\X{=}\mat x) > 0$ for all $\mat x \in \V(\X)$
\unskip), in which case the construction gives a unique causal model.
Conversely, we have seen that an acylic model $\mathcal M$ gives rise to a unique witness.
So, in the simplest cases, models
$\cal M$ with structure $\Ar$ and
 witnesses $\bar \mu$ to \cibility\ with $\Ar$
 are equivalent.
But there are two important caveats.
\begin{enumerate}[topsep=0pt,itemsep=0pt]
    \item A causal model $\cal M$ can contain more information than a witness $\bar\mu$ if some events have probability zero. For instance, $\bar\mu$ could be a point mass on a single joint outcome $\omega$ of all variables that satisfies the equations of $\cal M$. But $\cal M$ cannot be reconstructed uniquely from $\bar\mu$ because there may be many causal models for which $\omega$ is a solution. 
    \item A witness $\bar\mu$ can contain more information than a causal model 
    $\cal M$ if $\cal M$ is cyclic. For example, 
    suppose that
    $\cal M$ consists
    of two variables, $X$ and $X'$, and equations $f_X(X') = X'$ and
      $f_{X'}(X) = X$
    \unskip. In this case, $\bar\mu$ cannot be reconstructed from $\cal M$, because $\cal M$ does not contain information about the distribution of $X$. 
\end{enumerate}
These two caveats appear to be very different, but they fit
together in a surprisingly elegant way.
\commentout{
But even when the correspondence is not unique, the causal models ${\cal M} \in \PSEMsA(\bar\mu)$ still contain much of the same information as the witness $\bar\mu$ from which they are derived. 
}

\begin{linked}{prop}{witness-model-properties}
    \commentout{
    Suppose $\Ar$ is \subpartl\ and $\bar\mu(\X,\U)$ is a witness
    to $\mu \models \Diamond \Ar$. Then:
    \begin{enumerate}[label=(\alph*), nosep]
    \item 
    $\bar\mu$ can arise from every causal model in $\PSEMsA(\bar\mu)$.
        \hfill
        $\Big(\text{i.e., }\displaystyle
        \bar\mu \in~
        \bigcap_{\mathrlap{\!\!\!\!\mathcal M \in \PSEMs(\bar\mu)} } ~~\SD[\big]{\cal M}
        .
        \Big)
        $
    \item Moreover,
    $\PSEMsA(\bar\mu)$
        is the set of all PSEMs 
        from which $\bar\mu$ can arise, with equations consistent with a derandomization of the dependency structure $\Ar$.
    That is, 
    \[
        \PSEMsA(\bar\mu) = \{ 
        \mathcal M 
        ~|~ \bar\mu \in \SD{\mathcal M}
        ~~\land~~ \forall X {\in }\X.\, {\Pa}_{\mathcal M}(X) \subseteq \Src {a_{\!X}} \cup \{U_{a_{\!X}}\} \}.
    \]
    \end{enumerate}
    }
    If $\bar\mu(\X,\U_{\Ar})$ is a witness for \scibility\ with $\Ar$ and $\mathcal M$ is a PSEM with dependency structure $\Ar$,
    then $\bar\mu \in \SD{\cal M}$ if and only if $\mathcal M \in \PSEMsA(\bar\mu)$. 
\end{linked}

\commentout{
In other words, if you convert a distribution $\bar\mu$ witnessing $\mu \models \Ar$ to a causal model $\mathcal M \in \PSEMs(\nu)$, then 
$\bar\mu$ can arise from $\mathcal M$; moreover,
$\PSEMs(\bar\mu)$ is the set of all SEMs from which $\bar\mu$ can arise,
with equations consistent with the dependency structure of the \hgraph\ $\Ar$. 
}

Equivalently, this means that $\PSEMsA(\bar\mu)$, the possible outputs of \cref{constr:PSEMs}, are precisely the randomized PSEMs 
of dependency structure $\Ar$ that can give rise to $\bar\mu$. 
This is already substantial evidence that causal models $\mathcal M \in \PSEMsA(\bar\mu)$ are closely related to the \scibility\ witness 
$\bar\mu$. But everything we have seen so far describes only the
 correspondence in the absence of intervention, a setting in which many causal models are indistinguishable.  
Yet the correspondence goes deeper; it extends to interventions. 
This claim may seem dubious, as the obvous distinction between observing and doing is a fundemental principle of causality.
What could a witness, which is purely probabilistic, have to say about intervention? 
In 
a randomized PSEM $\mathcal M$, we can define an event 
\begin{equation}
    \mathrm{do}_{\mathcal M}(\mat X{=}\mat x) := 
        ~~
        \bigcap_{X \in \mat X} \bigcap_{\mathrlap{\mat s \in \V(\Pa(X))}~~}~~~
            \big(
            f_X(U_{\!{X}}, \mat s) = \mat x[X]
            \big)
            ,
    \qquad\parbox{3cm}{\centering
    where $\mat x[X]$ is the value of $X$ in $\mat x$.
    }
\end{equation}
This is intuitively the event in which the randomness is such that $\mat X = \mat x$ regardless of the values of the parent variables.
\vfull{%
\unskip\footnote{\vfull{This is essentially the event in which, for each $X \in \mat X$, the \emph{response variable} 
$\hat U_X := \lambda \mat s. f_{X}(\mat s, U_X)$,
whose possible values $\smash{\V(\hat U_X)}$ are functions from $\V(\Pa_{M}(X))$ to $\V(X)$ \citep{Rubin74,BP94}, takes on the constant function
$\lambda \mat p. ~x$.%
}}%
}%
As we now show,
conditioning
$\bar \mu$
on 
$\mathrm{do}_{\cal M}(\mat X{=}\mat x)$ has the effect of intervention
\unskip---at least
as long as the noise variables $\U_{\Ar} = \{\mathcal U_a\}_{a \in 
\Ar}$ are independent of the other exogenous variables $\U \setminus \U_{\Ar}$ in addition to one another (e.g., when $\mathcal M$ is fully randomized).

\begin{linked}{theorem}{condition-intervention-alignment}
Suppose that
$\bar\mu$ is a QIM witness for $\Ar$, that
$\mathcal M = (\U, \enV, \mathcal F, P) \in \PSEMsA(\bar\mu)$
is a corresponding PSEM,
and that the noise variables $\U_{\Ar} = \{ U_X \}_{X \in \enV}$ are independent of the other exogenous variables $\U \setminus \U_{\Ar}$.
For all $\mat X \subseteq \enV$ and $\mat x \in \V(\mat X)$,
if $\bar\mu(\mathrm{do}_{\cal M}(\mat X{=}\mat x)) > 0$, then
\begin{enumerate}[label={(\alph*)}, topsep=0pt,itemsep=0pt]
\item
    $\bar\mu(\enV \,|\, \mathrm{do}_{\cal M}(\mat X{=}\mat x))$ can arise from $\mathcal M_{\mat X {\gets} \mat x}$;
    
    \item
    for all events $\varphi \subseteq \V(\enV)$, 
    $ \displaystyle
        \Pr\nolimits_{\mathcal M}
        \big([\mat X {\gets} \mat x]\varphi\big)
    \le
        \bar\mu\big(\varphi \,\big|\, \mathrm{do}_{\cal M}(\mat X{=}\mat x)\big)
    \le
        \Pr\nolimits_{\mathcal M}
        \big(\langle\mat X {\gets} \mat x\rangle\varphi\big)
    $ \\
    and all three are equal when $\mathcal M \models \U \tto \enV$
    (such as when $\mathcal M$ is acyclic). 
\end{enumerate}
\end{linked}

\cref{theorem:condition-intervention-alignment} shows that the relationship between witnesses and causal models extends to interventions. 
\commentout{%
To those familiar with the causality literature, it may be surprising that $\mathrm{do}_M(\mat X{=}\mat x)$ is actually an event \citep{sep-causal-models,halpern-2000}, and even moreso that 
ordinary conditioning on this event has the effect of intervention, since conditioning and intervention 
are usually thought of as fundmentally different. 
(Indeed, intervening on $\mat X{=}\mat x$ remains fundementally different from conditioning on $\mat X{=}\mat x$.)
Nevertheless, 
intervention on $\mat X{=}\mat x$ can be viewed as conditioning on $\mathrm{do}_M(\mat X{=}\mat x)$, a state of affairs that has allowed us to demonstrate that the relationship between causality and witnesses extends to interventions.
}%
\commentout{%
To those familiar with the causality literature, it may be surprising that $\mathrm{do}_M(\mat X{=}\mat x)$ is actually an event \citep{sep-causal-models,halpern-2000}, and that conditioning on 
    it has the effect of intervention, since conditioning and intervention are 
    conceptually very different. 
}%
\commentout{%
(We remark that although $\mathrm{do}_M(X{=}x)$ is often an event of
probability 0, \cref{theorem:condition-intervention-alignment} can be
extended to events ``close'' to 
$\mathrm{do}_{\mathcal M}(X{=}x)$ that have positive probability; conditioning on
them has an effect close to that of intervention.  We defer formal
datails to the full paper.)
}%
Even when $\mathrm{do}_{\mathcal M}(\mat X{=}\mat x)$ has probability zero, 
    it is always possible to find a nearly equivalent
    setting where the bounds of the theorem apply.
    \unskip\footnote{More precisely, 
    for all $\epsilon > 0$, there exists some $\mathcal M'$ that
    differs from $\mathcal M$ on the probabilities of all causal formulas
    by at most $\epsilon$, and a distribution $\bar\mu'$ that is $\epsilon$-close to $\bar\mu$, such that $\bar\mu'(\mathrm{do}_{\mathcal M'}(\mat X{=}\mat x))>0$.
    \vfull{%
     As a result, \cref{theorem:condition-intervention-alignment} places bounds on the conditional probabilities that are possible limits of sequences of distributions $(\nu_k)_{k\ge 0}$ where $\nu_k(\mathrm{do}_M(\mat X{=}\mat x)) > 0$, i.e., the 
     possible outcomes of conditioning a \emph{non-standard} probability measure
        \cite{halpern-RAU}
     on this probability-zero event.}%
    } %
\commentout{
The connection between conditioning and intervension here may seem
surprising, because they are viewed as being conceptually very different.
}%
Intervention and conditioning are conceptually very different, so
it may seem surprising that
conditioning can have the effect of intervention
    (and also that the Pearl's $\mathrm{do}(\,\cdot\,)$ notation actually corresponds to an event \citep{sep-causal-models}). 
We emphasize that the conditioning (on $\mathrm{do}_{\cal M}(\mat X{=}\mat x)$) is on the randomness 
$\{ U_X \}_{X \in \mat X}$ and not $\mat X$ itself; 
intervening on $\mat X{=}\mat x$ is indeed fundamentally different
from conditioning on $\mat X{=}\mat x$.
\commentout{%
    Yet, as \cref{theorem:condition-intervention-alignment}
    shows, even the former can be encoded as 
    probabilistic conditioning in an extended distribution (a witness), 
    deepening the connection between witnesses and causal models.
}%

\commentout{
\subsubsection{Acyclic Hypergraphs}
When $\Ar = \Ar_G$ for some is a directed acyclic graph $G$ (i.e., a qualitative BN), then a witness to $\mu \models \Ar$ is
essentially an acyclic causal model
in which the equation for $X$ depends on $\Pa_G(X)$ and
independent noise, 
that uniquely determines the distribution $\mu(\X)$ in the absence of intervention.
    In light of this,
    \cref{theorem:bns} can be viewed as result 
        that implicitly undergirds much of the work on causality:
    every acyclic randomized SEM
    induces a distribution with the independencies of the appropriate Bayesian Network---%
    and, conversely,
    every distribution with those independencies
    arises from such a causal model.

    \commentout{
    \begin{prop}
        If $\Ar = \Ar_G$ is acyclic,
        $\mu \models \Ar_G$, and $G$ is minimal 
            (i.e.,$\mu \not\models \Ar_{G'}$ for all $G' \subsetneq G$),
        then for all $\nu \in \Wits(\mu, \Ar)$ and $\mathcal M \in \PSEMs(\nu)$, $\Ar_{G} = \Ar_{\mathcal M}$.
    \end{prop}
    }
}

\commentout{
\subsubsection{Non-partitional Hypergraphs}
In the case where $\Ar$ is \subpartl\ (but not \suppartl), 
the situation is very similar. 
The only difference is that the variables that are not targets of any \hyperarc\ do not naturally correspond to an equation, and hence should be treated as exogenous.
Thus, a witness $\bar\mu$ to \scibility\ with a \subpartl\ \hgraph\  $(\X, \Ar)$ corresponds to causal model(s) $\PSEMs(\bar\mu)$ with 
endogenous variables $\bigcup_{a \in \Ar} \Tgt a$ and exogenous variables $\{ U_a \}_{a \in \Ar} \cup (\X \setminus \bigcup_{a \in \Ar}\Tgt a)$.
}

\commentout{
\begin{conj}
Every PSEM $\mathcal M$  with structure $\Ar$ determining a distribution
    $\mu(\X)$ in the absence of interventions, satisfies
\[
    \Wits(\mu, \Ar\mid D_x)_*( B )
    \le
    \mathcal M(B \mid \mathrm{do}(X{=}x))
    \le
    \Wits(\mu, \Ar \mid D_x)^*( B )
\]
for every event $B \subseteq \V(\X)$.
Here $\Wits(\mu, \Ar \mid D_x)$ denotes the set of witnesses
that $\mu$ is compatible with $\Ar$ conditioned on the event $D_x$,
while $\Wits(\mu, \Ar)_* := \sup_{\nu \in \Wits(\mu,\Ar)} \nu(B \mid D_x)$,
$\Wits(\mu, \Ar)^* := \inf_{\nu \in \Wits(\mu,\Ar)} \nu(B \mid D_x)$ 
are the corresponding lower and upper probabilities, respectively.
\end{conj}
}

\section{\SCibility\ and Information Theory}
    \label{sec:pdgs}
    \label{sec:info}

\commentout{
In this section, we investigate the rich relationship between \scibility\ and information theory, which is mediated by the theory of probabilistic dependency graphs. 
One aspect, which we will develop in \cref{sec:scoring-funs}, is an entropy-based scoring function that measures a distribution's degree of \SQIM-incompatibility. That function is quite different from (although can be viewed as essentially a special case of), the original qualitative PDG scoring function \citep{one-true-loss,pdg-infer}. 
Before that (\cref{sec:sc-infobound}), we present a more useful and fundmental
 aspect of the relationship, based on the same scoring function: a generic and highly non-trivial information-theoretic consequence of \scibility.
}
The fact that the dependency structure of a (causal) Bayesian network describes the independencies of the distribution it induces is fundamental to both causality and probability. 
It makes explicit the distributional consequences of BN structure.
Yet, despite
substantial interest \cite{Baier_2022},
generalizing the BN case to more
complex (e.g., cyclic) dependency structures remains largely an open
problem.  
In \cref{sec:sc-infobound}, we 
generalize the BN case by
providing an information-theoretic constraint,
capable of capturing conditional independence, functional dependence, and more,
on the distributions that can arise from an \emph{arbitrary} 
dependency structure.
This connection between causality and information theory has implications for both fields. It grounds the cyclic dependency structures found in causality in concrete constraints on the distributions they represent. 
At the same time, it allows us to resolve longstanding confusion about structure in information theory, clarifying the meaning of the so-called ``interaction information'', and
recasting a standard counterexample to substantiate the claim it was intended to oppose.
In \cref{sec:scoring-funs}, we strengthen this connection. Using entropy to measure distance to (in)dependence, we develop a scoring function to measure how far a distribution is from being QIM-compatible with a given dependency 
structure.  
This function turns out to have an intimate relationship with the qualitative PDG scoring fucntion $\IDef$, which we use to 
show that our information-theoretic constraints degrade gracefully on ``near-compatible'' distributions.

\commentout{
\begin{figure}
    \centering
    \def\circsize{1.6}
    \def\radsize{0.8}
    \edef\outerrad{1.6}
    \def\innerrad{1.00}
    \scalebox{0.7}{
    \begin{tikzpicture}[center base]
        \node at ({-\circsize-\radsize+0.1},{\circsize+\radsize-0.2}) {\underline{\textbf{(a)}}};

        \path[fill=white!90!black] (210:\radsize) circle (\circsize) ++(225:\circsize) node[label={[label distance=2pt,inner sep=0pt]left:\Large $X$}]{};
        \path[fill=white!90!black] (90:\radsize) circle (\circsize)
            ++(65:\circsize) node[label={[label distance=1ex]right:\Large $Y$}]{};
        \path[fill=white!90!black] (-30:\radsize) circle (\circsize) ++(-45:\circsize) node[label={[label distance=1pt,inner sep=0pt]right:\Large$Z$\!}]{};

        \begin{scope}[gray]
            \draw[] (-30:\radsize) circle (\circsize);
            \draw[] (210:\radsize) circle (\circsize);
            \draw[] (90:\radsize) circle (\circsize);
        \end{scope}
        
        \begin{pgfinterruptboundingbox} %
            \begin{scope}[even odd rule]
                \clip (210:\radsize) circle ({\circsize-0.03}) (-3,-3) rectangle (3,3);
                \clip (-30:\radsize) circle ({\circsize+0.03});
                \draw[ultra thick,blue,fill=blue, fill opacity=0.3] (-30:\radsize) circle (\circsize);
                \draw[ultra thick,blue] (210:\radsize) circle (\circsize);
            \end{scope}
        \end{pgfinterruptboundingbox}
        \begin{pgfinterruptboundingbox} %
            \begin{scope}[even odd rule]
                \clip (90:\radsize) circle ({\circsize-0.03}) (-3,-3) rectangle (3,3);
                \clip (210:\radsize) circle ({\circsize+0.03});
                \draw[ultra thick,purple,fill=purple, fill opacity=0.3] (210:\radsize) circle (\circsize);
                \draw[ultra thick,purple] (90:\radsize) circle (\circsize);
            \end{scope}
        \end{pgfinterruptboundingbox}
        \begin{pgfinterruptboundingbox} %
            \begin{scope}[even odd rule]
                \clip (-30:\radsize) circle ({\circsize-0.03}) (-3,-3) rectangle (3,3);
                \clip (90:\radsize) circle ({\circsize+0.03});
                \draw[ultra thick,violet!50!black,fill=violet, fill opacity=0.2] (90:\radsize) circle (\circsize);
                \draw[ultra thick,violet!50!black] (-30:\radsize) circle (\circsize);
            \end{scope}
        \end{pgfinterruptboundingbox}
        \node at (-90:\innerrad) {\small$\I(X;Z|Y)$};
        \node[rotate=-60] at (30:\innerrad) {\small$\I(Y;Z|X)$};
        \node[rotate=60] at (150:\innerrad) {\small$\I(X;Y|Z)$};
        \node at (0:0){\small$\I(X;Y;Z)$};
        \node at (90:\outerrad) {\small$\H(Y|X,Z)$};
        \node[rotate=60] at (-30:\outerrad) {\small$\H(Z|Y,X)$};
        \node[rotate=-60] at (-150:\outerrad) {\small$\H(X|Y,Z)$};
        
        \node[purple!50!black,text opacity=1,fill=purple!15!white,fill opacity=1,
                rotate=0] (HAB) at (-110:2.0) 
            {\small$\H(X|Y)$};
        \node[violet!20!black,text opacity=1,fill=violet!15!white,fill opacity=1,
                rotate=0] (HBC) at (145:2.10) 
            {\small$\H(Y|Z)$};
        \node[blue!50!black,text opacity=1,fill=blue!15!white,fill opacity=1,
            rotate=0] (IBC) at (15:2.1)
            {\small$\H(Z|X)$};
    \end{tikzpicture}
    }
    \!   
    \def\circsize{0.7}
    \def\radsize{0.42}
    \def\twosize{0.5}
    \begin{tabular}{@{}c@{}c@{}}
    \begin{tikzpicture}[center base,scale=0.9]
        \node at ({-\circsize-\radsize+0.1},{\circsize+\radsize-0.2}) {\textbf{\underline{(b)}}};
        \begin{scope}
            \path[fill=green!60!black,fill opacity=0.25]
                (210:\radsize) circle (\circsize) ++(210:\circsize)
                (210:\radsize) circle (\circsize) ++(210:\circsize)
                (-30:\radsize) circle (\circsize) ++(-30:\circsize);
        \end{scope}
        \begin{scope}[even odd rule]
            \clip (-30:\radsize) circle (\circsize)
                    (210:\radsize) circle (\circsize)
                    (90:\radsize) circle (\circsize);
            \fill[fill=white!90!black] 
                    (-30:\radsize) circle (\circsize)
                    (210:\radsize) circle (\circsize)
                    (90:\radsize) circle (\circsize);
         \end{scope}
         \begin{scope}
            \clip (-30:\radsize) circle (\circsize);
            \clip (210:\radsize) circle (\circsize);
            \clip (90:\radsize) circle (\circsize);
            \fill[fill=orange!80!black!40!white] 
                (90:\radsize) circle (\circsize);
          \end{scope}
        \begin{scope}[gray]
            \draw[] (-30:\radsize) circle (\circsize);
            \draw[] (210:\radsize) circle (\circsize);
            \draw[] (90:\radsize) circle (\circsize);
        \end{scope}
        \node at (-90:\twosize) {\small +1};
        \node at (30:\twosize) {\small +1};
        \node at (150:\twosize) {\small +1};
        \node at (0,0) {\small -1};
    \end{tikzpicture}
    &
    \begin{tikzpicture}[center base,scale=0.9]
        \node at ({-\circsize-\radsize+0.1},{\circsize+\radsize-0.2}) {\underline{\textbf{(c)}}};
        \begin{scope}
            \path[fill=green!60!black,fill opacity=0.25]
                (210:\radsize) circle (\circsize) ++(210:\circsize)
                (210:\radsize) circle (\circsize) ++(210:\circsize)
                (-30:\radsize) circle (\circsize) ++(-30:\circsize);
        \end{scope}
        \begin{scope}[even odd rule]
             \clip (-30:\radsize) circle (\circsize)
                    (210:\radsize) circle (\circsize)
                    (90:\radsize) circle (\circsize);
            \fill[fill=white!90!black] 
                    (-30:\radsize) circle (\circsize)
                    (210:\radsize) circle (\circsize)
                    (90:\radsize) circle (\circsize);
         \end{scope}
        \begin{scope}[gray]
            \draw[] (-30:\radsize) circle (\circsize);
            \draw[] (210:\radsize) circle (\circsize);
            \draw[] (90:\radsize) circle (\circsize);
        \end{scope}
        \node at (-90:\twosize) {\small +1};
        \node at (30:\twosize) {\small +1};
        \node at (150:\twosize) {\small +1};
    \end{tikzpicture}
    \end{tabular}
    \caption{
        \textbf{(a)} 
            The components of the information profile over three variables;
        \textbf{(b)} the information profile of $\muxor$;
        \textbf{(c)} the information profile of $P$ and $Q$ from \cref{example:ditrichotomy}.
    }
        \label{fig:info-diag}
\end{figure}
}

\begin{wrapfigure}[11]{o}{3.55cm}
    \def\circsize{1.6}
    \def\radsize{0.8}
    \edef\outerrad{1.6}
    \def\innerrad{1.00}
    \scalebox{0.7}{\!\!\!
    \begin{tikzpicture}[center base]

        \path[fill=white!90!black] (210:\radsize) circle (\circsize) ++(225:\circsize) node[label={[label distance=2pt,inner sep=0pt]left:\Large $X$}]{};
        \path[fill=white!90!black] (90:\radsize) circle (\circsize)
            ++(65:\circsize) node[label={[label distance=1ex]right:\Large $Y$}]{};
        \path[fill=white!90!black] (-30:\radsize) circle (\circsize) ++(-45:\circsize) node[label={[label distance=1pt,inner sep=0pt]right:\Large$Z$\!}]{};

        \begin{scope}[gray]
            \draw[] (-30:\radsize) circle (\circsize);
            \draw[] (210:\radsize) circle (\circsize);
            \draw[] (90:\radsize) circle (\circsize);
        \end{scope}
        
        \begin{pgfinterruptboundingbox} %
            \begin{scope}[even odd rule]
                \clip (210:\radsize) circle ({\circsize-0.03}) (-3,-3) rectangle (3,3);
                \clip (-30:\radsize) circle ({\circsize+0.03});
                \draw[ultra thick,blue,fill=blue, fill opacity=0.3] (-30:\radsize) circle (\circsize);
                \draw[ultra thick,blue] (210:\radsize) circle (\circsize);
            \end{scope}
        \end{pgfinterruptboundingbox}
        \begin{pgfinterruptboundingbox} %
            \begin{scope}[even odd rule]
                \clip (90:\radsize) circle ({\circsize-0.03}) (-3,-3) rectangle (3,3);
                \clip (210:\radsize) circle ({\circsize+0.03});
                \draw[ultra thick,purple,fill=purple, fill opacity=0.3] (210:\radsize) circle (\circsize);
                \draw[ultra thick,purple] (90:\radsize) circle (\circsize);
            \end{scope}
        \end{pgfinterruptboundingbox}
        \begin{pgfinterruptboundingbox} %
            \begin{scope}[even odd rule]
                \clip (-30:\radsize) circle ({\circsize-0.03}) (-3,-3) rectangle (3,3);
                \clip (90:\radsize) circle ({\circsize+0.03});
                \draw[ultra thick,violet!50!black,fill=violet, fill opacity=0.2] (90:\radsize) circle (\circsize);
                \draw[ultra thick,violet!50!black] (-30:\radsize) circle (\circsize);
            \end{scope}
        \end{pgfinterruptboundingbox}
        \node at (-90:\innerrad) {\small$\I(X;Z|Y)$};
        \node[rotate=-60] at (30:\innerrad) {\small$\I(Y;Z|X)$};
        \node[rotate=60] at (150:\innerrad) {\small$\I(X;Y|Z)$};
        \node at (0:0){\small$\I(X;Y;Z)$};
        \node at (90:\outerrad) {\small$\H(Y|X,Z)$};
        \node[rotate=60] at (-30:\outerrad) {\small$\H(Z|Y,X)$};
        \node[rotate=-60] at (-150:\outerrad) {\small$\H(X|Y,Z)$};
        
        \node[purple!50!black,text opacity=1,fill=purple!15!white,fill opacity=1,
                rotate=0] (HAB) at (-110:2.0) 
            {\small$\H(X|Y)$};
        \node[violet!20!black,text opacity=1,fill=violet!15!white,fill opacity=1,
                rotate=0] (HBC) at (145:2.10) 
            {\small$\H(Y|Z)$};
        \node[blue!50!black,text opacity=1,fill=blue!15!white,fill opacity=1,
            rotate=0] (IBC) at (15:2.1)
            {\small$\H(Z|X)$};
    \end{tikzpicture}\!\!\!%
    }
    \caption{$\mat I_\mu$.
     }
        \label{fig:info-diag-a}
\end{wrapfigure}
We now review the critical information theoretic concepts and their relationships to
    (in)dependence
    (see \cref{appendix:info-theory-primer} for a full primer).
Conditional entropy $\H_\mu(Y|X)$ measures how far $\mu$ is from satisfying the functional dependency $X \tto Y$
\unskip. Conditional
mutual information $\I_\mu(Y;Z| X) $ measures how far $\mu$ is from satisfying the conditional independence $Y \CI Z \mid X$. 
\commentout{
The two kinds of quantities and their unconditional variants are all interdefinable
(indeed, $\H(Y|X) = \I(Y;Y|X)$, a strengthening of \eqref{eq:conditional-self-independence-det}) \unskip, and generated by linear combinations of one another; moreover, they are related by an inclusion-exclusion rule (illustrated for three variables in 
\cref{fig:info-diag-a}
). Thus, to specify all standard information-theoretic quantities for a joint distribution $\mu(\X)$,
it suffices to give a vector of $2^{|\X|}-1$ numbers, which we will call the \emph{information proile} of $\mu$, and denote $\mat I_\mu$. 
}
Linear combinations of these quantities 
(for $X,Y,Z \subseteq \X$)
can be viewed as the inner product between a coefficient vector $\mat v$
and a $2^{|\X|}-1$ dimensional vector $\mat I_\mu$ that we will call the \emph{information profile} of $\mu$. For three variables, the components of this vector are illustrated in 
\cref{fig:info-diag-a} (right).
It is not hard to see that an arbitrary conjunction of (conditional) (in)dependencies 
can be expressed as
a constraint $\mat I_\mu \cdot \mat v \ge 0$, for 
some appropriate choice of vector $\mat v$
\unskip.

We now formally introduce the qualitative PDG scoring function $\IDef$, which interprets a hypergraph structure $\Ar$ as a function of the form $\mat I_\mu \cdot \mat v_{\!\Ar}$.
This \emph{information deficiency}, given by
\begin{equation}
    \IDef_{\!\Ar}(\mu) 
    = \mat I_\mu \cdot \mat v_{\!\Ar} 
    := - \H_\mu(\X)  + \sum_{a \in \Ar} \H_\mu(\Tgt a \mid \Src a)
    ,
        \label{eq:idef}
\end{equation}
is the difference between the number of bits needed to (independently) specify the randomness in $\mu$ along the \arc s of $\Ar$, and
the number of bits needed to specify a sample of $\mu$ according to its own structure 
($\emptyset \to \X$).
While $\IDef$ has some nice properties\footnote{%
\commentout{
These properties include:
\begin{itemize}[nosep]
    \item capturing BN independencies,
    \item capturing the dependencies in of \cref{theorem:func},
    \item nice interactions with the qualitative information in a PDG given by a sum of relative entropies,
    \item capturing factor graphs and their associated exponential families (as a result of the previous point),
    \item reducing to maximum entropy in the absence of structural information, 
    \item being easy to calculate from $\mu$ in closed form, and
    \item being the opposite of 
    the corresponding quantitative term in $\Inc$,
     in the sense of having the opposite vector field.
\end{itemize}}%
It 
captures BN independencies and the dependencies of \cref{theorem:func}, 
reduces to maximum entropy for the empty \hgraph, and
combines with the quantitative PDG scoring function \cite{pdg-aaai} to capture factor graphs.
},
it can also behave unintuitively in some cases; for instance, it can be negative.
Clearly, it does not measure how close $\mu$ is to being structurally compatible with $\Ar$, in general. 
\commentout{%
Nevertheless, if $\mu\models \Ar$, we prove that $- \IDef_{\!\Ar}(\mu) = \mat v_{\!\Ar} \cdot \mat I_\mu \ge 0$. 
This is quite useful; by inspecting $\mat v_{\!\Ar}$, one can often directly interpret the constraints imposed by \scibility\  in terms of well-studied information-theoretic primitives. The same connection also allows \scibility\ to elucidate some of the less inuitive information-theoretic quantities. 
}%
Nevertheless, there is still a fundamental relationship 
between $\IDef$ and \scibility, as we now show.
\commentout{
Still, the quantity $\IDef_{\!\Ar}(\mu)$ does not determine whether or not $\mu \models \Ar$. 
In \cref{sec:scoring-funs}, we derive an entropy-based scoring function that does, by directly translating conditions (a-c) to  scoring functions on witnesses, and then adding an infemum to take care of the existential quantifier (i.e. the choice of witness) implicit in \cref{defn:scompat}.
Perhaps surprisingly, without the infemum, this function is
an instance of $\IDef$---yet not for the original \hgraph, but a transformed one.
Thus the new scoring function is essentially an instance of the old one, tailored for \scibility.  When combined with a modest extension of \cref{theorem:sdef-le0}, this observation also gives us natural easy-to-compute bounds 
(\cref{theorem:siminc-idef-bounds})
for
this new \scibility-based
scoring function.
}

\commentout{
\clearpage
\subsection*{Previous \S5 Intro}
\commentout{    
\citet{pdg-aaai} provided a scoring function that 
    was intended to measure 
    how far a probability distribution was from being structurally  compatible with a PDG, viewed as a hypergraph.
    We discuss this
    scoring function below; for now, we just mention that the scoring
    function was rather difficult to interpret as a degree of
    incompatibility since, for example, it could be negative.
    Fortunately, the definition of structural incompatibility we have
    given here leads to a natural scoring function, which we now describe.
}
\commentout{
    \citet{pdg-aaai} defined an entropy-based scoring function ($\IDef$) that 
        was intended to measure discrepancy between a probability distribution
        and a PDG's underlying (weighted) \dhygraph\ structure.
    Although that scoring function has many nice properties,
        it is rather difficult to interpret directly as a ``degree of
        incompatibility'' since, for example, it could be negative.
    In \cref{sec:sc-infobound}, we discuss an important
     connection between \scibility\ and the information theory behind $\IDef$, which works in two directions.
    On one hand, it gives us an entropy-based criterion to conclude that a distribution cannot be \scible\ with a \dhygraph; on the other, our definition of \scibility\ can be used to shine light on a (seeming) counterexample in the information theory literature.
    Nevertheless, $\IDef$ still does not exactly coincide with \scibility.
        In \cref{sec:scoring-funs}, we derive a new scoring function that
        directly measures how far a distribution is being \scible\ with a 
        given \dhygraph---%
        which, perhaps surprisnigly, turns out to be an instance of the original scoring function, applied to a transformed \dhygraph.
    To understand these information theoretic issues, we begin by reviewing some standard definitions.
}
Just as \scibility\ has 
an important relationship with causal models
 (\cref{sec:causal}), so too does it have 
an important relationpship with information theory.
We have seen how \scibility\ can capture 
certain conditional indepependencies (\cref{theorem:bns})
and dependencies (\cref{theorem:func}),
by restricting to special classes of \hgraph s. 
But, for \hgraph s outside these classes (such as those generated by cyclic graphs), \scibility\ can 
impose much subtler constraints on a
    distribution
\unskip, which require information theory to describe
     \unskip.
The connection also useful in the other direction:
 \scibility\ can clarify important misunderstandings in information theory
    \unskip.
It is also possible to directly capture \scibility\ with 
an entropy-based scoring function and an infemum.
\commentout{
    We detail both phenomena in \cref{sec:sc-infobound}.
    In \cref{sec:scoring-funs}, we derive an entropy-based scoring function that measures how far a distribution is from being \scible\ with a given \hgraph. 
    But, perhaps surprisingly,
    this new scoring function coincides with the original PDG scoring function
    if the underlying hypergraph is first transformed appropriately
        \unskip.
    Thus PDGs as originally defined, difficult as they may be to interperet in general, capture \scibility\ in a natural way.
}

In some sense, information theory is a continuous relaxation of probabilitic (in)dependence:
the conditional forms of mutual information $\I_\mu(Y;Z|X)$
and entropy $\H_\mu(Y|X)$%
\footnote{
Either quantity can be derived from the other. 
In particular, $\H_\mu(Y|X) = \I_{\mu}(Y;Y|X)$, deepening 
the observation made in \cref{eq:conditional-self-independence-det}.
See \cref{appendix:info-theory-primer} for definitions and a short review.}---the most important quantities in the field---measure
how far a distribution $\mu$ is from satisfying
a conditional independence $Y \CI Z \mid X$ and
a functional dependence $X \tto Y$,
respectively. 
Multivariate information theory promises
to distill \emph{all} {structural} relationships between a 
set $\X$ of variables 
(distributed according to $\mu$)
into a vector of $2^{|\X|}-1$ 
such
numbers 
that fit together with an inclusion-exclusion rule,
called
    the
    \emph{information profile}
    of $\mu$.
This is depicted for $\X = \{X,Y,Z\}$ in \cref{fig:info-diag}(a).
\commentout{ 
This representation is made possible by the fact that
    these quantities are related (defined, in fact)
    by an inclusion-exclusion rule, 
which allows us to read off
identities like
\commentout{
    $\I(X;Y;Z) = \H(X,Y,Z) -  {\color{purple!60!black}\H(X|Y)} + {\color{violet!65!black}\H(Y|Z)} + {\color{blue!50!black}\H(Z|X)}$ 
        (which will soon play an important role)
        off such a diagram.
}
$\H(X|Y) = \H(X|Y,Z)+ \I(X;Z|Y)$. 
}%
Does multivariate information theory deliver on its promises? 
Many have argued that the answer is no.
Yet one of the most intuitively powerful arguments rests on 
a poorly motivated notion
 of structure in distributions
\unskip. When
we replace that notion with \scibility, the
argument
is not only defused, but 
 supports the opposite of its intended purpose (\cref{example:ditrichotomy}).

\commentout{
\begin{defn}[Qualitative PDG]
    A Qualitative Dependency Graph (QDG) $\dg Q = (\N, \Ar, \balpha)$,
    is a \dhygraph\ $(\N, \Ar)$, together with a vector
    $\balpha = \{ \alpha_a \}_{a \in \Ar} \in \mathbb R^{\Ar}$ of
    weights for each \arc.
\end{defn}

For completeness, a (purely) quantitative probabilistic
dependency graph is a tuple $(\X, \mathbb P, \bbeta)$,
in which each node $N \in \N$ is interpreted a variable $X_N \in \X$
by associating it with a set $\V(N)$ of possible values,
and each \arc\ $\ed aST \in \Ar$ is 
    interpreted as a mechanism by associating it
with a conditional probability distribution
$\p_a(\Tgt a | \Src a)$ on its target given its source.
}

\commentout{
With multivariate information theory, we can describe the ``subtler constraints'' imposed by \scibility:
they are linear inequalities
of the form $\mat{i}_\mu \cdot \mat v \ge 0$, 
    where $\mat i_{\mu}$ is $\mu$'s information profile, and $\mat v$ is a vector of integers.
}
SIM-compatibility with a hypergraph imposes a fundemental constraint on a distribution's information profile,
\unskip\footnote{
Specifically, a constraint of the form $\mat{i}_\mu \cdot \mat v \ge 0$, where $\mat i_{\mu}$ is $\mu$'s information profile, and $\mat v$ is a vector of integers.
It is not hard to see that arbitrary collections of (conditional) (in)dependencies can be captured by specifying $\mat v$ appropriately. 
    Given this, the reader might wonder: why bother with \scibility\ at all? 
    Here are two reasons: the dimension of the information profile is exponential in $|\X|$, and its components become increasingly unintuitive as the number of variables grows. 
    Our goal is not only to capture (in)dependence 
    \unskip, but to do so in 
    a compact, graphical, and causally motivated manner. 
}
\commentout{
To this end, we would like to compile a \hgraph\ to a function of an information profile, which we do directly in \cref{sec:scoring-funs}. 
But there is a simpler and more fundemental scoring function in the literature that has a more intricate and useful relationship with \scibility, an even closer relationship with \scibility, and ultimately will be able to capture both notions indirectly. 
}

\citet{pdg-aaai} define
a 
scoring function 
that converts a \hgraph\ to 
such
a linear function of $\mu$'s information profile, which was intended to quantify discrepancy between $\mu$ and a PDG's underlying structure.
\commentout{
    \Citet{pdg-aaai}
    define the \emph{information deficiency} of
    a distribution $\mu$ with respect to a \dhygraph\ $\Ar$
    as:
}
This \emph{information deficiency}, given by
\begin{equation}
    \IDef_{\!\Ar}(\mu) := - \H_\mu(\X)  + \sum_{a \in \Ar} \H_\mu(\Tgt a \mid \Src a),
\end{equation}
\commentout{
    The first term is the number of bits needed to specify an outcome of
    $\mu$
    according to $\mu$'s own structure $(\to \X)$
    while the second
    is the number of bits needed to (independently) specify randomness along the \arc s of $\Ar$.
    The structural deficiency $\IDef_{\!\Ar}(\mu)$ is
    therefore the difference between the cost of a sample of
        $\mu$ along the structure $\Ar$, and describing it directly.
}
is the difference between the number of bits needed to (independently) specify the randomness in $\mu$ along the \arc s of $\Ar$, and the number of bits needed to specify a sample of $\mu$ according to its own structure $(\to \X)$.
$\IDef$ has many nice properties.%
\footnote{
When $G$ is a dag, then $\IDef_{\!\Ar_G}(\mu)$
is non-negative, and
is zero iff $\mu$ has the independencies of the corresponding BN, 
which is a direct
analogue of our \cref{theorem:bns}.
$\IDef$ also
\ifvfull\vfull{
also satisfies analogues of
\cref{prop:equiv-factorizations-cnd,prop:strong-mono}.%
    \footnote{
    Specifically, if $\Ar_1, \Ar_2, \Ar_3$ are the three \hgraph s in \cref{prop:equiv-factorizations-cnd}, then it is easy to prove that $\IDef_{\!\Ar_1} = \IDef_{\!\Ar_2} = \IDef_{\!\Ar_3}$.
    In addition,
    If $\Ar \rightsquigarrow \Ar'$, then it is easy to show that $\IDef_{\!\Ar} \ge  \IDef_{\!\Ar'}$.
    See appendix for details.}
}
\else
satisfies an analogue of \cref{prop:strong-mono}: if $\Ar \rightsquigarrow \Ar'$, then $\IDef_{\!\Ar}(\mu) \ge \IDef_{\!\Ar'}(\mu)$.
\fi
}
However $\IDef_{\!\Ar}(\mu)$ can be
negative in some cases (e.g., when $\Ar = \emptyset$)),
and so clearly it does not measure ``how far a $\mu$ is from being compatible with $\Ar$''
\unskip, in general.
Nevertheless,
\scibility\ and $\IDef$ are
deeply connected.
We will demonstrate a different aspect of that connection in each of the coming subsections.

\commentout{
In \cref{sec:sc-infobound}, we describe a deep and nontrivial connection between \scibility\ and
information deficiency ($\IDef$),
which works in two directions.  In one direction, it gives us an entropy-based criterion to conclude that a distribution cannot be \scible\ with a \hgraph; in the other, 
it enables
 our definition of \scibility\ to clarify some important
 misunderstandings
in information theory. 
We then (\cref{sec:scoring-funs}) derive a new entropy-based scoring function to directly measure how far a distribution is being \scible\ with a given \hgraph.
But, perhaps surpringly, the new function turns out to be an instance of 
$\IDef$ applied to a transformed \hgraph.
}

\commentout{
In the PDG formalism
the idea is to find a distribution $\mu$ that minimizes $\IDef_{\!\Ar}(\mu)$
plus an observational incompatibility $\Inc_{\dg M}(\mu)$ with the qualitative information.
}
}

\subsection{A Necessary Condition for \SCibility}
    \label{sec:sc-infobound}
What constraints does \scibility\ with $\Ar$ place on a distribution $\mu$?
When $G$ is a dag, we have seen that if $\mu \models \Diamond \Ar_G$, then $\mu$ must satisfy the independencies of the corresponding Bayesian network (\cref{theorem:bns}); we have also seen that additional \arc s impose functional dependencies (\cref{theorem:func}). But these
results apply only when $\Ar$ is of a very special form.  
\commentout{%
    In \cref{sec:monotone} we saw how we can weaken \scibility\ relationships, but information is lost in the process, and the resulting \scibility\ statement often no easier to interperet.
}%
More generally, $\mu \models \Diamond \Ar$ implies that $\mu$ can arise from some 
randomized causal model whose equations
have dependency structure $\Ar$
(\cref{prop:sc-graph-arise,prop:gen-sim-compat-means-arise}).
Still,
unless $\Ar$ has a particularly special form, 
it is not obvious 
whether or not
this says something about $\mu$.
The primary result of this section is an
information-theoretic bound (\cref{theorem:sdef-le0}) that
generalizes most of the concrete consequences of \scibility\ 
we have seen so far (\cref{theorem:bns,theorem:func}).
\commentout{%
The result is a bridge between information theory and causality---%
in one direction, giving us a useful information-theoretic consequence of \scibility,
and in the other, enabling \scibility\ to dispell misconceptions in information theory.
}%
The result is a connection between information theory and causality; it yields an information-theoretic test for complex causal dependency structures, and enables causal notions of structure to dispel misconceptions in information theory.

\commentout{
    We start by stating a useful result that relates \scibility\ and the original PDG scoring function. 
}
\begin{linked}{theorem}{sdef-le0}
    If $\mu \models \Diamond\Ar$, then
    $\IDef_{\!\Ar}(\mu) \le 0$.
\end{linked}

\commentout{
\Cref{theorem:sdef-le0} gives  
a direct information-theoretic test  
for \scibility
\ with an arbitrary hypergraph.
Moreover, it subsumses every
}
\Cref{theorem:sdef-le0} applies to all hypergraphs, and subsumes every general-purpose technique we know of for proving that $\mu \not\models \Diamond\Ar$. 
Indeed, the negative directions of \cref{theorem:bns,theorem:func} are
 immediate consequences of it.
To illustrate
some of its subtler implications,
we return to the 3-cycle in \cref{example:xyz-cycle-1}.

\begin{wrapfigure}[5]{o}{1.7cm}
    \vspace{-2ex}
    \def\circsize{0.7}
    \def\radsize{0.42}
    \def\twosize{0.5}
    
    \scalebox{0.85}{
    \begin{tikzpicture}[center base,scale=0.9]
        \begin{scope}
            \path[fill=green!60!black,fill opacity=0.25]
                (210:\radsize) circle (\circsize) ++(210:\circsize)
                (210:\radsize) circle (\circsize) ++(210:\circsize)
                (-30:\radsize) circle (\circsize) ++(-30:\circsize);
        \end{scope}
        \begin{scope}[even odd rule]
            \clip (-30:\radsize) circle (\circsize)
                    (210:\radsize) circle (\circsize)
                    (90:\radsize) circle (\circsize);
            \fill[fill=white!90!black] 
                    (-30:\radsize) circle (\circsize)
                    (210:\radsize) circle (\circsize)
                    (90:\radsize) circle (\circsize);
         \end{scope}
         \begin{scope}
            \clip (-30:\radsize) circle (\circsize);
            \clip (210:\radsize) circle (\circsize);
            \clip (90:\radsize) circle (\circsize);
            \fill[fill=orange!80!black!40!white] 
                (90:\radsize) circle (\circsize);
          \end{scope}
        \begin{scope}[gray]
            \draw[] (-30:\radsize) circle (\circsize);
            \draw[] (210:\radsize) circle (\circsize);
            \draw[] (90:\radsize) circle (\circsize);
        \end{scope}
        \node at (-90:\twosize) {\small +1};
        \node at (30:\twosize) {\small +1};
        \node at (150:\twosize) {\small +1};
        \node at (0,0) {\small -1};
    \end{tikzpicture}}
\end{wrapfigure}
\refstepcounter{example}
    \label{example:nonneg-ii}
\textbf{Example \theexample.~}
It is easy to see 
(e.g., by inspecting \cref{fig:info-diag-a}) 
that
$\IDef_{\text{3-cycle}}(\mu) = 
    \H_\mu(Y|X) + \H_\mu(Z|Y) + \H_\mu(X|Z) - \H_\mu(XYZ)
    =
    -\I_\mu(X;Y;Z)$.
\Cref{theorem:sdef-le0} therefore tells us that 
a
distribution
    $\mu$ that is \scible\ 
    with the 3-cycle 
    cannot have negative interaction information
$
    \I_\mu(X; Y; Z)
        . 
$
What does this mean?
\vfull{
    Overall, 
    conditioning on
    the value of one variable can only reduce the amount of 
    remaining information in other 
    variables (in expectation).
}
When $\I(X;Y;Z) < 0$, 
conditioning on
one variable causes the other two to share more information than they did before. 
The most extreme instance is $\muxor$, the distribution in which two variables are independent and the third is their parity 
(illustrated on the right).
It seems intuitively clear that $\muxor$ cannot arise from 
the 3-cycle, a causal model with only pairwise dependencies.
This is difficult to prove directly, but is an immediate consequence of \cref{theorem:sdef-le0}.
\hfill$\triangle$

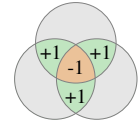
\commentout{
This may seem to require a fundementally ``3-way'' interaction between the variables, rather than pairwise relationships
---but, according to the traditional ways of making ``pairwise relationships'' precise (e.g.,  maximum entropy subject to pairwise marginal constraints, factoring over pairwise factors), this is incorrect. 
However, if we take it to mean there is a causal model without any joint dependencies, then the 3-cycle is the most expressive structure, and we have already seen that \scibility\ with it implies non-negative interaction information.
\commentout{
to mean \scibility\ with the cycle, then, by \cref{theorem:sdef-le0}, this intuition is well-founded. 
But historically, that intuition has been a sticking point for interpreting interaction information, because other ways of making ``2-way interactions'' precise, such as matching pairwise marginals, or factoring as a product of 2-variable factors, do not rule out negative interaction information $\I(X;Y;Z)$. 
}
\commentout{
    More concretely, \cref{theorem:sdef-le0} means that $\muxor$,
    the distribution in which two variables are independent and the third is their parity, cannot come from a causal model without a joint dependence
    (as one would expect).
    It is worth noting that $\muxor$ is
    the most extreme example of negative interaction information,
    because $\I(X;Y)$ is non-negative 
    (see \cref{fig:info-diag}(b)).
}
Our next example builds on this one.
}

For many, there is an
intuition that
$\I(X;Y;Z) < 0$
should require
a fundementally ``3-way'' interaction between the variables,
and should not arise through pairwise interactions alone \cite{dit-stumble}.
This has been a source of conflict 
\cite{williams2010nonnegative,mackay2003information,1219753,CoverThomas},
because traditional ways of making 
precise ``pairwise interactions'' 
(e.g., maximum entropy subject to pairwise marginal constraints 
and pairwise factorization)
 do not ensure that $\I(X;Y;Z) \ge 0$. 
But \scibility\ does. 
One can verify by enumeration that the 3-cycle is the most expressive causal structure with no joint dependencies, 
and we have already proven that \scibility\ with that \hgraph\ 
implies non-negative interaction information.
\Scibility\ has another even more noteworthy clarifying effect on information theory. 

There is a school of thought
that contends that \emph{all} structural information in $\mu(\X)$ is captured by its information profile $\mat I_\mu$.
This position has fallen out of favor
in some communities
due to standard counterexamples: distributions that have intuitively different structures yet share an information profile
\citep{multivar-beyondshannon17}. 
However, with ``structure''  explicated by compatibility,
the prototypical counterexample of this kind
suddenly
supports the very notion it was meant to challenge, suggesting in an unexpected way that the information profile may yet capture the essence of probabilistic structure.

\begin{wrapfigure}[5]{o}{1.8cm}
    \vspace{-2ex}
    \def\circsize{0.7}
    \def\radsize{0.42}
    \def\twosize{0.5}
    \scalebox{0.85}{%
    \begin{tikzpicture}[center base,scale=0.9]
        \begin{scope}
            \path[fill=green!60!black,fill opacity=0.25]
                (210:\radsize) circle (\circsize) ++(210:\circsize)
                (210:\radsize) circle (\circsize) ++(210:\circsize)
                (-30:\radsize) circle (\circsize) ++(-30:\circsize);
        \end{scope}
        \begin{scope}[even odd rule]
             \clip (-30:\radsize) circle (\circsize)
                    (210:\radsize) circle (\circsize)
                    (90:\radsize) circle (\circsize);
            \fill[fill=white!90!black] 
                    (-30:\radsize) circle (\circsize)
                    (210:\radsize) circle (\circsize)
                    (90:\radsize) circle (\circsize);
         \end{scope}
        \begin{scope}[gray]
            \draw[] (-30:\radsize) circle (\circsize);
            \draw[] (210:\radsize) circle (\circsize);
            \draw[] (90:\radsize) circle (\circsize);
        \end{scope}
        \node at (-90:\twosize) {\small +1};
        \node at (30:\twosize) {\small +1};
        \node at (150:\twosize) {\small +1};
    \end{tikzpicture}}
\end{wrapfigure}
\refstepcounter{example} \label{example:ditrichotomy}
\textbf{Example {\theexample}.~}         
    Let $A, B$, and $C$ be variables
    with $\V(A), \V(B), \V(C) = 
    \{ 0,1\}^2$. 
    Using independent fair coin flips $X_1$, $X_2$, and $X_3$, define
        two joint distributions, $P$ and $Q$, 
        over $A,B,C$ as follows. 
    Define $P$ by selecting $A := (X_1, X_2)$, $B := (X_2, X_3)$, and $C := (X_3, X_1)$.
    Define $Q$ by selecting
    $A := (X_1, X_2)$, $B := (X_1, X_3)$, and $C := (X_1, X_2 \oplus X_3)$.
    Structurally, $P$ and $Q$ appear to be very different. 
    According to $P$, the first components of the three variables ($A,B,C$) are independent, yet they are identical according to $Q$. 
    Moreover, $P$ has only simple pairwise interactions between
        the variables, while $Q$ has $\muxor$ 
        (a clear 3-way interaction) embedded within it. 
    Yet $P$ and $Q$ have identical information profiles 
    (see right):
    in both cases, 
    each of $\{A,B,C\}$
    is
    determined by the values of the other two,
    each pair share one bit of information given the third,
    and $\I(A;B;C) = 0$.
    
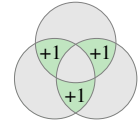
\commentout{
\begin{center}
    \def\circsize{0.7}
    \def\radsize{0.42}
    \def\twosize{0.5}
    \begin{tikzpicture}[center base]
        \path[fill=white!70!black] (210:\radsize) circle (\circsize) ++(210:\circsize) node[label={[label distance=0pt,inner sep=0pt]left:\small$A$}]{};
        \path[fill=white!70!black] (90:\radsize) circle (\circsize) ++(135:\circsize) node[label={left:\small$B$}]{};
        \path[fill=white!70!black] (-30:\radsize) circle (\circsize) ++(-30:\circsize) node[label={[label distance=0pt,inner sep=0pt]right:\small$C$}]{};
        \begin{scope}
            \draw[] (-30:\radsize) circle (\circsize);
            \draw[] (210:\radsize) circle (\circsize);
            \draw[] (90:\radsize) circle (\circsize);
        \end{scope}
        \node at (-90:\twosize) {\small +1};
        \node at (30:\twosize) {\small +1};
        \node at (150:\twosize) {\small +1};
    \end{tikzpicture}
\end{center}}

This example has been used to argue that multivariate Shannon
    information does not take into account important structural
    differences between distributions \citep{multivar-beyondshannon17}.
We are now in a position to give a novel and particularly persuasive response, by appealing to \scibility.%
\vfull{%
\footnote{%
    Note that $P$ and $Q$ no longer have the same profile if we split each variable into its two components.  
    Since the notion of ``component'' is based on
    the assignment $\V$ of variables to possible values,
    our view that $\V$ is {not} structural information
    diffuses this counterexample by assumption---%
    but the present argument is a deeper one.
}
}
    Unsurprisingly, $P$ is \cible\ with the 3-cycle;
    it is clearly consists of ``2-way'' interactions, as each pair of variables shares a bit.
But, 
counterintuitively, 
the distribution $Q$ is \emph{also} \cible\ with the 3-cycle!
(The reader is encouraged to verify that 
    $U_1 = X_3 \oplus X_1$, $U_2 = X_2$, and $U_3 = X_3$
    serves as a witness.) 
To emphasize: this is despite the fact that $Q$ 
is just 
    $\muxor$
    (which is certainly not compatible with the 3-cycle)
together with a seemingly irrelevant random bit $X_1$.
By the results of \cref{sec:causal},
this means there is a causal model without joint dependence giving rise to $Q$---%
    so,
    despite appearances,
    $Q$ does not 
    require a 3-way interaction.
    Indeed, $P$ and $Q$ are \scible\ with precisely the same \hgraph s over $\{A,B,C\}$, 
    suggesting that they don't have a structural difference after all. 
\hfill$\triangle$

\commentout{%
The shared bit $X_1$ between the three variables, suggestively,
contains precisely the magnitude of negative interaction information in $\muxor$. One might therefore wonder if the converse of \cref{theorem:sdef-le0} holds. 
}%
In light of \cref{example:ditrichotomy}, one might reasonably conjecture that the converse of \cref{theorem:sdef-le0} holds.
Unfortunately, it does not (see \cref{appendix:converse-sdef-le0});
the quantity $\IDef_{\!\Ar}(\mu)$ does not 
completely determine whether or not 
$\mu \models \Diamond \Ar$. We now pursue a new (entropy-based) scoring function that does.
This will allow us to generalize \cref{theorem:sdef-le0} to distributions that are only ``near-compatible'' with $\Ar$. 

\subsection{A Scoring Function for \SCibility}
    \label{sec:scoring-funs}
Here is a function that measures how far a distribution $\mu$ is from being \scible\ with 
$\Ar$.
\begin{align*}
    \SIMInc_{\Ar}(\mu) :=
    \inf_{\substack{
            \nu(\mathcal U\!,\, \X)\\
            \mathclap{\nu(\X) = \mu(\X)}}}
            \;
    - \H_\nu (\mathcal U) + 
    \sum_{a \in \Ar} \H_\nu (U_a)
    +
        \sum_{a \in \Ar}
        \H_\nu(\Tgt a | \Src a , U_a)
    .  \numberthis\label{eq:siminc}
    \vspace{-1ex}
\end{align*}

\commentout{
This may look complicated, 
but each line is just a generic (entropy-based) measure of how far an extended joint distribution $\nu$ is from satisfying one of the conditions (a-c) of being a \scibility\ witness (\cref{defn:scompat}). 
Because each is non-negative and zero iff the appropriate condition is satisfied, 
$\SIMInc_{\!\Ar}$
overall is non-negative, and zero iff all three conditions are satisfied.
}
$\SIMInc$ is a direct translation of \cref{defn:scompat} (a-c); 
it measures the (optimal) quality of an extended distribution $\nu$ as a witness.
The infimum restricts the search to $\nu$ satisfying (a), the first two terms measure $\nu$'s discrepancy of with (b), and the last term measures $\nu$'s discrepancy with (c). 
Therefore:

\begin{linked}{prop}{sinc-nonneg-s2}
\commentout{
    For all $\boldsymbol\lambda > 0$, we have
    $\SIMInc_{(\Ar,\boldsymbol\lambda)}(\mu) \ge 0$, and this holds
    with equality if and only if
    $\mu \models \Ar$.
}
    $\SIMInc_{\!\Ar}(\mu) \ge 0$, 
    with equality iff
    $\mu \models \Diamond \Ar$.
\end{linked}

\commentout{
Although it is built from a sum of entropies, $\SIMInc_{\Ar}$ is not a linear function of its argument's information profile; rather, because of the infemum, it is an optimization problem over the information profiles of extended distributions.
This makes it difficult to compute, and obfuscates the connection to information theoretic primitives. 
}%
Although they seem to be very different, $\SIMInc$ and $\IDef$ turn out to be closely related.  
\commentout{
    We now want to compare the scoring function $\SIMInc$ to that introduced
    by \citet{pdg-aaai}.  We focus on qualitative PDGs (i.e., \dhygraphs)
    here.%
    \footnote{In \citet{pdg-aaai}, a qualitative PDG also includes a
    weight function $\alpha$ that associates with each \hyperarc\ $a$ a
    weight that intuitively represents the  modeler's degreee of
    confidence in the dependence represented by $a$.
    For simplicitly, we
    ignore the weights here (implicitly taking them all to be 1).
    }
}%
In fact, modulo the infimum, $\SIMInc_{\Ar}$ is a special case of
$\IDef$---not for the \hgraph\ $\Ar$, but rather for a transformed one $\Ar^\dagger$ that models the noise variables explcitly. 
To construct $\Ar^{\dagger}$ from $\Ar$,
add new nodes $\U = \{ U_a \}_{a \in \Ar}$, and
replace each \hyperarc\
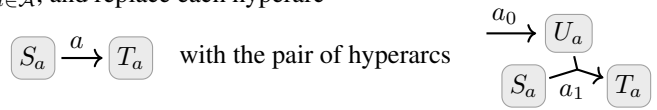
\vspace{-1em}
\[
\begin{tikzpicture}[center base]
    \node[dpad0] (S) at (0,0) {$\Src a$};
    \node[dpad0] (T) at (1.3,0) {$\Tgt a$};
    \draw[arr2] (S) to node[above,pos=0.4]{$a$}
        (T);
\end{tikzpicture}
\quad 
\text{with the pair of \hyperarc s}
\quad
\begin{tikzpicture}
        [center base]
    \node[dpad0] (S) at (0,0) {$\Src a$};
    \node[dpad0] (T) at (1.4,0) {$\Tgt a$};
    \node[dpad0] (U) at (0.6,0.7) {$U_a$};
    \draw[arr,<-] (U) to node[above,pos=0.6]
        {$a_0$}
         +(-1.2,0);
    \mergearr[arr1] SUT
    \node[below=2pt of center-SUT,xshift=-0.2em] 
        {$a_1$};
\end{tikzpicture}~.
\]
\commentout{
    In words, construct $\Ar^{\dagger}$ by starting 
    and two \hyperarc s
    one generating $U_a$ without input, and another
    from $\Src a \cup \{ U_a \}$ to $\Tgt a$.
}
Finally, add one additional
\hyperarc\ $\mathcal U \to \X$
\unskip. (Intuitively, this \hyperarc\ creates functional dependencies in the spirit of \cref{theorem:func}.)
\commentout{
In total, it has \arc s
\[
    \Ar^\dagger :=
        \Bigg \{~~
        \begin{tikzpicture}[center base]
            \node[dpad0] (S) at (0,0) {$\Src a$};
            \node[dpad0] (T) at (1.5,0) {$\Tgt a$};
            \node[dpad0] (U) at (0.8,1) {$U_a$};
            \draw[arr,<-] (U) -- +(-1,0);
            \mergearr SUT
    \end{tikzpicture}~~
    \Bigg \}_{a \in \Ar} ~~\cup~~ \Big\{ \ed {}{\mathcal U}{\X} \Big \}.
\]
}%
\commentout{
    In a moment, we will see how $\IDef$ for $\Ar^\dagger$ corresponds to
    $\SIMInc$ for $\Ar$. But first, notice a potential issue with these assertions
    at face value: the final \arc\ seems to be double-counting
    causal information that was already given in the previous set of \arc s.
    How is it that the randomness variables $\mathcal U$ determine $\X$ up to noise?
    Via the mechanisms that have already been described!
    }%
With these definitions in place,
we can state a theorem that bounds $\SIMInc$ above and below with information deficiencies
(\cref{theorem:siminc-idef-bounds})%
. 
The lower bound generalizes \cref{theorem:sdef-le0} by giving an upper limit on $\IDef_{\!\Ar}(\mu)$ even for distributions $\mu$ that are 
not \scible\ with $\Ar$.
The upper bound is tight in general, and shows that $\SIMInc_{\!\Ar}$ can be equivalently defined as a minimization over 
    $\IDef_{\!\Ar^\dagger}$.

\begin{linked}{theorem}{siminc-idef-bounds}
    \begin{enumerate}[label={\normalfont(\alph*)},wide,topsep=0pt,itemsep=0pt,parsep=0pt] \item
    \commentout{
    If $\Ar$ is a \hgraph\ whose nodes are the names of variables 
    $\X$, $\U$ is a collection of variables indexed by
    $\Ar$,
    $\mu$ is a distribution over $\X$, and $\nu(\X,\U)$ is an extension of $\mu$ to $\U$, then
    }
    If $(\X,\Ar)$ is a hypergraph,
     $\mu(\X)$ is a distribution, and $\nu(\X,\U)$ is an extension of $\nu$ to additional variables $\U = \{U_a \}_{a \in \Ar}$ indexed by $\Ar$, then:
    \[
        \IDef_{\!\Ar}(\mu) \le
        \SIMInc_{\Ar}(\mu)
        \le \IDef_{\!\Ar^\dagger}(\nu).
    \]
    \item 
    For all $\mu$ and $\Ar$, there is a choice of $\nu$ that achieves the upper bound.
    That is, 
    \[%
        \SIMInc_{\Ar}(\mu)
        = \min \Big\{\ \IDef_{\!\Ar^\dagger}(\nu) : 
            \begin{array}{c}
            \nu \in  \Delta\V(\X,
                \U
            ) \\ \nu(\X) = \mu(\X)
            \end{array}
         \Big\}
         .
    \]
\commentout{
    where $\hat \U_\Ar^\X$ is the set of response variables along the \arc s of $\Ar$ for variables $\X$,
        defined in \cref{sec:responsvars}.
}
\vfull{
where the minimization is over all possible ways of assigning values to 
the variables in
$\U$. The
minimum is achieved when $|\V(U_a)| \le |\V(\Tgt a)|^{|\V(\Src a)|}$.
}
    \end{enumerate}
\end{linked}

The semantics of PDGs are based on the idea of measuring (and resolving)
 \emph{inconsistency}, 
    which is defined as a minimization over $\IDef$ 
(plus a term that captures relevant concrete probabilistic information)
\unskip. Thus,     
\cref{theorem:siminc-idef-bounds} (b)
    tells us that QIM-compatibility (with $\Ar$) can be captured with a 
        qualitative PDG (namely, $\Ar^\dagger$). 
It follows that
 our notion of QIM-compatibility 
can be viewed as a special case of the semantics of PDGs---one that, 
as we have shown, has a causal interpretation.
\commentout{
The first result gives us bounds on $\SIMInc$ in terms of $\IDef$, making it easier to compute, and the second shows that, in the context of a minimization where it usually is \citep{pdg-aaai}, $\SIMInc$ is a special case of $\IDef$, for a graph that models the noise variables explicitly. 

The first inequality of \cref{theorem:siminc-idef-bounds}(a) is a slight strengthening of \cref{theorem:sdef-le0} that applies even when $\mu \not\models \Diamond \Ar$. 
But \cref{theorem:sdef-le0} is not based on some choice of scoring function---it describes a direct relationship between the information profile and \scibility.
}

\commentout{
\begin{conj}
    Suppose $\mu \models \Ar$, $\mu' \models \Ar'$, and $\alpha \in [0,1]$.
    Then $\IDef_{\!\Ar{:}\alpha{:}\Ar'}((1-\alpha) \mu + (\alpha) \mu) \le 0$,
    where $\Ar{:}\alpha{:}\Ar'$ is the QDG that contains the \arc s $\Ar$
    with weight $\alpha$
\end{conj}
\begin{proof}
    $\IDef_{\!\Ar}(\mu) \le 0$, and $\IDef_{\!\Ar'}(\mu') \le 0$,
    so...
\end{proof}
}

\section{Discussion and Conclusions}

We have shown how \dhygraph s can be used to represent 
structural aspects of distributions.  
Moreover, they can do so in a way that 
generalizes conditional independencies and functional dependencies and
has deep connections to causality and information theory. 
This notion of QIM-compatibility can be captured with PDGs,
    and also partially explains the qualitative foundations of these models.
Still, many questions remain open.
\commentout{
    A simple sounding one is the following: are there any distributions that are compatible with the 3-cycle in one direction, but not in the opposite direction? 
    We suspect of the symmetry, and because $\IDef$ is the same for both cycles%
        ---but we have not been able to resolve the question one way or another (see \cref{sec:equivalence} for a partial answer).
    Another obvious question is the decision procedure for deciding whether or not $\mu \models \Ar$.
}%

Perhaps the most important open problem is that of computing whether or not 
a given distribution $\mu$ is QIM compatible with a directed hypergraph $\Ar$.
We have implemented a rudimentary approach (based on solving problem \eqref{eq:siminc} to calculate $\SIMInc$) that works in practice for 
small examples, but that approach scales poorly, and its correctness
has not yet been proved.
Even representing a distribution $\mu$ over $n$ variables requires $\Omega(2^n)$ space in general,
and a candidate witness $\bar\mu$ is even bigger:
    if all variables are binary, 
    $|\Ar| = m$,
    and $|\Src a|, |\Tgt a| \le k$ for all $a \in \Ar$, then 
    a direct implementation of \eqref{eq:siminc} is a non-convex optimization problem with 
at most
$2^{n+m k (2^k)}$ variables.
Even accepting the (substantial) cost of representing
extended distributions, we do not have a bound on the time needed to solve the optimization problem. 
There are more compact ways of representing the joint distributions $\mu$ used in practice (by assuming (in)dependencies), but we do not 
know if such independence assumptions make it easier to determine whether
$\mu \models \Diamond \Ar$
    for arbitrary $\Ar$.
But computing $\IDef_{\!\Ar}(\mu)$ can be much easier.%
    \footnote{The complexity of calculating $\IDef_{\!\Ar}(\mu)$ is 
        typically dominated by the difficulty of calculating the joint entropy $\H(\mu)$. 
    It can be difficult to compute $\H(\mu)$ in some cases (e.g., for undirected models),
    but in others (e.g., for Bayesian Networks or clique trees) 
    the same assumptions that enable a compact representation of $\mu$ also 
        make it easy to calculate $\H(\mu)$.
    }
We suspect that \cref{theorem:sdef-le0}, 
    a nontrivial condition for QIM-compatibility that requires only computing $\IDef_{\!\Ar}(\mu)$,
could play a critical role in designing such an inference procedure.

Another major 
    open problem is that of
    more precisely
    understanding
    the implications of \scibility\ in cyclic models.  
We do not yet know, for example, whether the same set of distributions are \scible\ with the
clockwise and counter-clockwise 3-cycles.
\commentout{ 
Inference for PDGs was discussed in \citep{pdg-infer}, but the focus
was on the quantative aspects.
Getting analogous algorithms
that dealt with the qualitative aspects and \scibility\ would be of
great interest.  Finally, it would be useful to get more axiomatic
approaches for reasoning about \scibility.  
We believe that ideas of
monotonicity, together with axioms that relate independence and
depence, in the spirit of
(\ref{eq:conditional-self-independence-det}), might take us a long way.
}%

As mentioned in \cref{sec:causal}, our notion of QIM-compatibility
has led us to a generalization of a standard causal model
(\cref{defn:GRPSEM}).
A proper investigation of this novel modeling tool
(which we have not attempted in this paper)
    would include concrete motivating examples, a careful account of interventions and counterfactuals in this general setting, 
    and results situating these causal models among other generalizations of causal models in the literature. 

We hope to address these questions in future work.

\newpage
\begin{ack}
    We would like to thank the reviewers for 
        useful discussion and helpful feedback, 
            such as the pointer to \citet{Druzdzel1993CausalityIB},
        and for asking us to expand on the complexity of inference. 
    Thank you to 
    Matt MacDermott for identifying a bug in a prior version of \cref{theorem:condition-intervention-alignment},
    and
    to Matthias Georg Mayer for catching several low-level issues with the presentation. 
    The work of Halpern and Richardson was supported in part by
    AFOSR grant FA23862114029, MURI grant W911NF-19-1-0217, ARO grant
    W911NF-22-1-0061, and NSF grant FMitF-2319186.
    S.P. is supported in part by the NSF under Grants Nos.~CCF-2122230 and CCF-2312296, a Packard Foundation Fellowship, and a generous gift from Google.
\end{ack}

\bibliography{qdg-refs,joe}

\clearpage
\appendix

\ifchecklist
\newpage
\section*{NeurIPS Paper Checklist}

\begin{enumerate}

\item {\bf Claims}
    \item[] Question: Do the main claims made in the abstract and introduction accurately reflect the paper's contributions and scope?
    \item[] Answer: \answerYes{}  %
    \item[] Justification: All claims are substantiated with precise theorem statements in Sections 2-4, and their implications are carefully discussed.

\item {\bf Limitations}
    \item[] Question: Does the paper discuss the limitations of the work performed by the authors?
    \item[] Answer: \answerYes{}
    \item[] Justification: For example, we discuss the inherent limitations of our central notion with respect to undirected graphical models, and we pose several open problems regarding the distributional implications of cyclic dependency structures that we were not able to solve in this work. We also explain that we only partially characterize the distributions compatible with general dependency structures.

\item {\bf Theory Assumptions and Proofs}
    \item[] Question: For each theoretical result, does the paper provide the full set of assumptions and a complete (and correct) proof?
    \item[] Answer: \answerYes{}
    \item[] Justification: This is a theoretical work and as such we take mathematical precision very seriously.

    \item {\bf Experimental Result Reproducibility}
    \item[] Question: Does the paper fully disclose all the information needed to reproduce the main experimental results of the paper to the extent that it affects the main claims and/or conclusions of the paper (regardless of whether the code and data are provided or not)?
    \item[] Answer: \answerNA{}
    \item[] Justification: This paper is theoretical in nature and does not include experimental results.

\item {\bf Open access to data and code}
    \item[] Question: Does the paper provide open access to the data and code, with sufficient instructions to faithfully reproduce the main experimental results, as described in supplemental material?
    \item[] Answer: \answerNA{}
    \item[] Justification: This paper is theoretical in nature and there is no associated data or code.

\item {\bf Experimental Setting/Details}
    \item[] Question: Does the paper specify all the training and test details (e.g., data splits, hyperparameters, how they were chosen, type of optimizer, etc.) necessary to understand the results?
    \item[] Answer: \answerNA{}
    \item[] Justification: This paper is theoretical and does not involve any training or testing of models.

\item {\bf Experiment Statistical Significance}
    \item[] Question: Does the paper report error bars suitably and correctly defined or other appropriate information about the statistical significance of the experiments?
    \item[] Answer: \answerNA{}
    \item[] Justification: This paper is theoretical and does not contain experiments.

\item {\bf Experiments Compute Resources}
    \item[] Question: For each experiment, does the paper provide sufficient information on the computer resources (type of compute workers, memory, time of execution) needed to reproduce the experiments?
    \item[] Answer: \answerNA{}
    \item[] Justification: This paper does not have experiments, computational or otherwise.
    
\item {\bf Code Of Ethics}
    \item[] Question: Does the research conducted in the paper conform, in every respect, with the NeurIPS Code of Ethics \url{https://neurips.cc/public/EthicsGuidelines}?
    \item[] Answer: \answerYes{}
    \item[] Justification: To the best of our understanding, since this paper focuses on purely theoretical questions regarding graphical languages and the structure of probability distributions, there are no substantive ethical concerns to address.

\item {\bf Broader Impacts}
    \item[] Question: Does the paper discuss both potential positive societal impacts and negative societal impacts of the work performed?
    \item[] Answer: \answerNA{}
    \item[] Justification: This paper is best viewed as basic theoretical research and is therefore far from direct societal impacts.
    
\item {\bf Safeguards}
    \item[] Question: Does the paper describe safeguards that have been put in place for responsible release of data or models that have a high risk for misuse (e.g., pretrained language models, image generators, or scraped datasets)?
    \item[] Answer: \answerNA{}
    \item[] Justification: We do not have and have not released any data or models.

\item {\bf Licenses for existing assets}
    \item[] Question: Are the creators or original owners of assets (e.g., code, data, models), used in the paper, properly credited and are the license and terms of use explicitly mentioned and properly respected?
    \item[] Answer: \answerNA{}
    \item[] Justification: We use no external assets.

\item {\bf New Assets}
    \item[] Question: Are new assets introduced in the paper well documented and is the documentation provided alongside the assets?
    \item[] Answer: \answerNA{}
    \item[] Justification: We do not release new assets.

\item {\bf Crowdsourcing and Research with Human Subjects}
    \item[] Question: For crowdsourcing experiments and research with human subjects, does the paper include the full text of instructions given to participants and screenshots, if applicable, as well as details about compensation (if any)? 
    \item[] Answer: \answerNA{}
    \item[] Justification: The paper does not involve crowdsourcing nor research with human subjects.

\item {\bf Institutional Review Board (IRB) Approvals or Equivalent for Research with Human Subjects}
    \item[] Question: Does the paper describe potential risks incurred by study participants, whether such risks were disclosed to the subjects, and whether Institutional Review Board (IRB) approvals (or an equivalent approval/review based on the requirements of your country or institution) were obtained?
    \item[] Answer: \answerNA{}
    \item[] Justification: The paper does not involve crowdsourcing nor research with human subjects.

\end{enumerate}
\clearpage
\fi

\section{Proofs}
    \label{appendix:proofs}

We begin with a de-randomization construction, that will be useful for the proofs. 

\subsection{From CPDs to Distributions over Functions}
    \label{sec:cpd-derandomize}
Compare two objects:
\begin{itemize}[nosep]
    \item a cpd $p(Y|X)$, and
    \item a distribution $q(Y^X)$ over functions $g : \V X \to \V Y$. 
\end{itemize}

The latter is significantly larger --- if both $|\V X| = | \V Y | = N$, then
$q$ is a $N^N$ dimensional object, while $p$ is only dimension $N^2$. 
A choice of distribution $q(Y^X)$ corresponds to a unique choice cpd $p(Y|X)$, according to 
\[
    p(Y{=}y\mid X{=}x) :=
        q(Y^X(x) = y)
        .
\]
\begin{claim}
\begin{enumerate}[nosep]
\item The definition above in fact yields a cpd, i.e., $\sum_y p(Y{=}y|X{=}x) = 1$
for all $x \in \V X$. 
\item 
This definition of $p(Y|X)$ is the conditional marginal of any
    joint distribution  $\mu(X,Y,Y^X)$ satisfying $\mu(Y^X) = q$ and
    $\mu(Y = Y^X(X)) = 1$. 
\end{enumerate}
\end{claim}
Both $p$ and $q$ give probabilistic information about $Y$ conditioned
on $X$. But $q(Y^X)$ contains strictly more information.
Not only does it specify the distribution over $Y$ given $X{=}x$,
but it also contains counter-factual information about the distribution of $Y$
    if $X$ were equal to $x'$, conditioned on the fact that, in reality, $X{=}x$.

Is there a natural construction that goes in the opposite direction, intuitively
making as many independence assumptions as possible? It turns out there is:
\[
    q(Y^X{=} g) = \prod_{x \in \V X} p(Y{=}g(x) \mid X{=}x).
\]
Think of $Y^X$ as a collection of variables $\{Y^x : x \in \V X\}$ describing the
value of the function for each input, so that $q$ is a joint distribution over them.
This construction simply asks that these variables be independent. 
Specifying a distribution with these independences amounts to a choice
    of ``marginal'' distribution $q(Y^x)$ for each $x \in V X$, and hence
    is essentially a funciton of type $\V\! X \to \Delta \V Y$, the same as $p$. 
In addition, if we apply the previous construction, we recover $p$, since:
\begin{align*}
    q(Y^X(x) = y)
        &= \sum_{g: \V\! X \to \V Y} \mathbbm1[g(x)=y]\prod_{x' \in \V\! X} p(Y{=}g(x') \mid X{=}x') \\
        &= \sum_{g: \V\! X \to \V Y} \mathbbm1[g(x)=y] p(Y{=}g(x)\mid X{=}x)\prod_{x' \ne x} p(Y{=}g(x') \mid X{=}x') \\
        &=  p(Y{=}y\mid X{=}x) \sum_{g: \V\! X \to \V Y} \mathbbm1[g(x)=y] \prod_{x' \ne x} p(Y{=}g(x') \mid X{=}x') \\
        &=  p(Y{=}y\mid X{=}x) \sum_{g: \V\! X \setminus \{x\} \to \V Y} \prod_{x' \in \V \! X \setminus \{x\} } p(Y{=}g(x') \mid X{=}x') \\
        &=  p(Y{=}y\mid X{=}x).
\end{align*}
The final equality holds because the remainder of the terms can be viewed as the probability of selecting any function from $X \setminus \{x\}$ to $Y$, under an analogous measure; thus, it equals 1. 
This will be a useful construction for us in general.
\commentout{
given a cpd $p(Y|X)$, let's
write $\tilde p(Y^X)$ for this distribution.  
}%

\subsection{Results on (In)dependence}

\begin{lemma} \label{lem:indep-fun}
    Suppose $X_1, \ldots, X_n$ are variables, $Y_1, \ldots, Y_n$ are sets,
        and for each $i \in \{1, \ldots n\}$, we have a function
        $f_i : \V (X_i) \to Y_i$.
    Then if $X_1, \ldots, X_n$ are mutually independent (according to a
        joint distribution $\mu$), then
    so are $f_1(X_1), \ldots, f_n(X_n)$.
\end{lemma}
\begin{proof}
    This is an intuitive fact, but we provide a
    proof for completeness.
    Explicitly, mutual independence of $X_1, \ldots, X_n$ means that,
        for all joint settings $\mat x = (x_1, \ldots x_n)$,
    we have $\mu(X_1{=}x_1,\ldots,X_n{=}x_n) = \prod_{i=1}^n \mu(X_i{=}x_i)$.
    So, for any joint setting $\mat y = (y_1, \ldots, y_n) \in Y_1 \times \cdots \times Y_n$, we have
    \begin{align*}
        \mu\Big(f_1(X_1){=}y_1, \ldots, f_n(X_n){=}y_n\Big)
            &= \mu(\{ \mat x : \mat f(\mat x) = \mat y\}) \\
            &= \sum_{\substack{(x_1,\ldots, x_n) \in \V(X_1,\ldots,X_n) \\ f_1(x_1) = y_1,\; \ldots,\; f_n(x_n) = y_n}}
                \mu(X_1{=}x_1, \,\ldots,\, X_n{=}x_n) \\
            &= \sum_{\substack{x_1 \in \V\! X_1 \\ f_1(x_1) = y_1}}
                    \cdots
                \sum_{\substack{x_n \in \V\! X_n \\ f_n(x_n) = y_n}}
                \mu(X_1{=}x_1, \,\ldots,\, X_n{=}x_n) \\
            &= \sum_{\substack{x_1 \in \V\! X_1 \\ f_1(x_1) = y_1}}
                    \cdots
                \sum_{\substack{x_n \in \V\! X_n \\ f_n(x_n) = y_n}}
                \prod_{i=1}^n \mu(X_i{=}x_i) \\
            &= \bigg(\sum_{\substack{x_1 \in \V\! X_1 \\ f_1(x_1) = y_1}}
                    \mu(X_1{=}x_1) \bigg) \cdots
                \bigg(\sum_{\substack{x_n \in \V\! X_n \\ f_n(x_n) = y_n}}
                    \mu(Y_1=y_1)\bigg) \\
            &= \prod_{i=1}^n \mu( f_i(X_i) = y_i ).
             \qedhere
    \end{align*}
\end{proof}

\begin{lemma}[properties of determination]
        \label{lem:detprop}
    ~
    \begin{enumerate}
        \item 
        If $\nu \models A \tto B$ and $\nu \models A \tto C$, then $\nu \models A \tto (B,C)$. 

        \item 
        If $\nu \models A \tto B$ and $\nu \models B \tto C$, then $\nu \models A \tto C$. 
    \end{enumerate}
\end{lemma}
\begin{lproof}
$\nu \models X \tto Y$, means there exists a function $f : V(A) \to V(B)$ such that $\nu( f(Y) = X ) = 1$, i.e., the event $f(A) = B$  occurs with probability 1.
\begin{enumerate}
    \item Let $f : \V(A) \to \V(B)$ and $g : \V(A) \to \V(C)$ be such that $\nu(f(A) = B) = 1 = \nu(g(A) = C)$. Since both events happen with probability 1, so must the event $f(A) = B \cap g(A) = C$. Thus the event $(f(A), g(A)) = (B,C)$ occurs with probability 1.  Therefore, $\nu \models A \tto (B,C)$. 
    \item The same ideas, but faster: we have $f : \V(A) \to \V(B)$ as before, and $g : \V(B) \to \V(C)$, such that the events $f(A) = B$ and $g(B) = C$ occur with proability 1. By the same logic, it follows that their conjunction holds with probability 1, and hence $C = f(g(A))$ occurs with probability 1. So $\nu \models A \tto C$. 
\end{enumerate}
\end{lproof}

\recall{theorem:bns}
\begin{lproof}
    \label{proof:bns}
    Label the vertics of $G = (\N, E)$ by natural numbers so that they are a topological sort of $G$---that is, without loss of generality, suppose $\N = [n] := \{1, 2, \ldots, n\}$, and $i < j$ whenever  $i \to j \in E$. 
    By the definition of $\Ar_G$, 
    the arcs $\Ar_G = \{ \Src i \xrightarrow{i} i \}_{i=1}^n$ are also indexed by integers.
    Finally, write $\X = (X_1, \ldots, X_n)$ for the variables $\X$ corresponding to $\N$ over which $\mu$ is defined. 
    
    \textbf{($\implies$).}
    Suppose $\mu \models \Diamond \Ar_G$. 
    This means there
    is an extension of $\bar\mu(\X, \U)$ of $\mu(\X)$ to additional
        independent variables $\U = (U_1, \ldots, U_n)$, such that $\bar\mu \models (\Src i, U_i) \tto i$ for all $i \in [n]$. 
    
    First, we claim that if $\bar\mu$ is such a witness, then 
    $\bar\mu \models (U_1, \ldots, U_k) \tto (X_1, \ldots, X_k)$
    for all $k \in [n]$, and so in particular, $\bar\mu \models \U \tto \X$. 
    This follows from \scibility's condition (c) and the fact that $G$ is acyclic, by induction.
    In more detail:
    The base case of $k=0$ holds vacuously.
    Suppose that $\bar\mu \models (X_1, \ldots, X_k)$ for some $k < n$. 
    Now, conditon (c) of \cref{defn:scompat} says 
    $\bar\mu  \models (\Src {k+1}, U_{k+1}) \tto X_{k+1}$. 
    Because the varaibles are sorted in topological order, the parent variables $\Src {k+1}$ are a subset of $\{X_1, \ldots, X_n\}$, which are determined by $\U$ by the induction hypothesis;
    at the same time clearly $\bar\mu \models (U_1, \ldots, U_{k+1}) \tto U_{k+1}$  as well. 
    So, by two instances of \cref{lem:detprop}, 
        $\bar\mu \models (U_1, \ldots U_{k+1}) \tto X_{k+1}$.
    Combining with our inductive hypothesis, we find that
        $\bar\mu \models (U_1, \ldots U_{k+1}) \tto (X_1, \ldots, X_{k+1})$.
    So, by induction, 
    $\bar\mu \models (U_1, \ldots, U_k) \tto (X_1, \ldots, X_k)$
    for $k \in [n]$, and in particular, $\bar\mu \models \U \tto \X$.

    With this in mind, we now return to proving that $\mu$ has the required independencies.
    It suffices to show that $\mu(\X) = \prod_{i=1}^n \mu(X_i \mid \Src i)$.
    We do so by showing that, for all $k \in [n]$, 
    $\mu(X_1, \ldots, X_{k}) = \mu(X_1, \ldots, X_{k-1}) \mu(X_{k}\mid \Src {k})$.
    By \scibility\ witness condition (c), we know that $\bar\mu\models (\Src k, U_k) \tto X_k$, and so there exists a function $f_k : \V(\Src k) \times \V(U_k) \to \V(X_k)$ for which the event $f_k(\Src k, U_k) = X_k$ occurs with probability 1. 
    Since $\bar\mu \models (U_1, \ldots, U_{k-1}) \tto (X_1, \ldots, X_{k-1})$, and $U_k$ is independent of $(U_1, \ldots, U_{k-1})$,
    it follows from \cref{lem:indep-fun} that 
    $\bar\mu \models (X_1, \ldots, X_{k-1}) \CI U_k$.
    Thus
    \begin{align*}
        \mu(X_1, \ldots, X_{k-1}, X_k)
             &= \sum_{u \in \V(U_k)}
                \mu(X_1, \ldots, X_{k-1}) \bar\mu(U_k = u) 
                    \cdot \mathbbm 1[X_k = f_k(\Src k, u)]
                    \\
            &= \mu(X_1, \ldots, X_{k-1}) \sum_{u \in \V(U_k)}\bar\mu(U_k = u) 
                \cdot \mathbbm 1[X_k = f_k(\Src k, u)].
    \end{align*}
    Observe that the quantity on the right, including the sum, is a function of $X_k$ and $S_k$, but no other variables; let  $\varphi(X_k, \Src k)$ denote this quantity. 
    Because $\mu$ is a probability distribution, know that
    $\varphi(X_k, \Src k)$ must be the conditional probability of $X_k$ given $X_1,\ldots, X_{k-1}$, and it depends only on the variables $\Src k$. Thus
    $\mu(X_1, \ldots, X_{k}) = \mu(X_1, \ldots, X_{k-1}) \mu(X_{k}\mid \Src {k})$.
    
    Therefore $\nu(\X) = \mu(\X)$ factors as required by the BN $G$, meaning that $\mu$ has the independencies specified by $G$. (See Koller \& Friedman Thm 3.2, for instance.)
    \bigskip
    
    \textbf{($\impliedby$).}
    Suppose $\mu$ satiesfies the independencies of $G$, meaning that each node is conditionally independent of its non-descendents given its parents.
    We now repeatedly apply the construction \cref{sec:cpd-derandomize} to construct a \scibility\ witness. 
    Specifically, for $k \in \{1, \ldots, n\}$, 
    let $U_k$ be a variable whose values $\V(U_k) := \V(X_k)^{\V(\Src k)}$ are functions from values of $X_k$'s parents, to values of $X_k$.
    Let $\mathcal U$ denote the joint variable $(U_1, \ldots, U_n)$,
    and observe that a setting $\mat g = (g_1, \ldots, g_n)$ of $\mathcal U$ uniquely picks out a value of $\X$, by evaluating each function in order. 
    Let's call this function $f : \V(\U) \to \V(\X)$. 
        
    To be more precise, we now construct $f(\mat g)$ inductively. 
    The first component we must produce is $X_1$, but since $X_1$ has no parents, $g_1$ effectively describes a single value of $X_1$, so we define the first component $f(\mat g)[X_1]$ to be that value.
    More generally, assuming that we have already defined the components $X_1, \ldots, X_{i-1}$, among which are the variables $\Src k$ on which $X_i$ depends, we can determine the value of $X_i$;
    formally, this means defining
    \[
        f(\mat g)[X_i] :=
            g_i( f(\mat g)[\Src i] ),
    \]
    which, by our inductive assumption, is well-defined.
    Note that, for all $\mat g \in \V(\U)$ and $\mat x \in \V(\X)$, the function $f$ is characterized by the property 
    \begin{equation}
            \label{eq:bn-func-char}
        f(\mat g) = \mat x
        \quad\iff\quad
        \bigwedge_{i=1}^n g_i(\mat x[\Src i]) = \mat x[X_i]. 
    \end{equation}
    To quickly verify this: if $f(\mat g) = \mat x$, then in particular, for $i \in [n]$, then $\mat x[X_i] = f(\mat g)[X_i] = g_i( \mat x[\Src i])$ by the definition above.
    Conversely, if the right hand side of \eqref{eq:bn-func-char} holds, then we can prove $f(\mat g) = \mat x$ by induction over our construction of $f$: if $f(\mat g)[X_{j}] = \mat x[X_j]$ for all $j < i$, then $f(\mat g)[X_i] = g_i( f(\mat g) [\Src i]) = g_i(\mat x[\Src i]) = \mat x[X_i]$.
    
    Next, we define an unconditional probability over each $U_k$ according to
    \[
        \bar\mu_i(U_i = g) := \prod_{\mat s \in \V(\Src k)} \mu(X_i = g(s) \mid \Src i \,{=}\, \mat s),
    \]
    which, as verified in \cref{sec:cpd-derandomize}, is indeed a conditional probability, and has the property that
    $\bar\mu_i(U_i(\mat s) = x) = \mu(X_i \,{=}\, x \mid \Src i \,{=}\, \mat s)$ for all $x \in \V(X_i)$ and $\mat s \in \V(\Src i)$.
    By taking an independent combination (tensor product) of each of these unconditional distributions, we obtain a joint distribution 
    $\bar\mu(\mathcal U) = \prod_{i=1}^n \bar\mu_i(U_i)$.
    Finally, we extend this distribution to a full joint distribution
        $\bar\mu(\U, \X)$ via the pushforward of $\bar \mu(\U)$ through the function $f$ defined by induction above.
    In this distribution, each $X_i$ is determined by $U_i$ and $\Src i$. 
    
    By construction, the variables $\U$ are mutually independent (for \cref{defn:scompat}(b)), and satisfy $(\Src k, U_k) \tto X_k$ for all $k \in [n]$ (\cref{defn:scompat}(c)). 
    It remains only to verify that the marginal of $\bar \mu$ on the variables $\X$ is the original distribution $\mu$ (\cref{defn:scompat}(a)). 
    Here is where we rely on the fact that $\mu$ satisfies the independencies of $G$, which means that we can factor $\mu(\X)$ as    $
        \mu(\X) = \prod_{i=1}^n \mu(X_i \mid \Src i).
    $
    \begin{align*}
        \bar\mu(\X{=}\mat x) 
        &= \sum_{\mat g \in \V(\U)} \bar\mu(\U{=}\mat g) 
            \cdot\delta\!f(\mat x \mid \mat g)
        \\
        &= \!\! \sum_{{(g_1, \ldots, g_n) \in \V(\U)}}\!\!
            \mathbbm1
                \big[
                \mat x = f(\mat g)
                \big] ~~
            \prod_{i=1}^n~ \bar\mu(U_i{=}g_i) \\
        &= \!\! \sum_{{(g_1, \ldots, g_n) \in \V(\U)}}\!\!
            \mathbbm1
                \Big[
                 \bigwedge_{i=1}^n g_i( \mat x[\Src i] ) = \mat x[X_i]
                \Big] ~~
            \prod_{i=1}^n~ \bar\mu(U_i{=}g_i) 
            & [\text{ by \eqref{eq:bn-func-char} }]
            \\
        &= \prod_{i=1}^n~~
            \sum_{{g \in \V(U_i)}} ~
            \mathbbm 1\big[ g(\mat x[\Src i]) = \mat x[X_i] \big] 
            \cdot
            \bar\mu(U_i = g) \\
        &= \prod_{i=1}^n~~
            \bar\mu \Big(\Big\{
                    g \in \V(U_i) ~\Big|~
                    g( \mat s_i ) = x_i
                \Big\}\Big)
                \quad\text{ where }
                \begin{array}{l}
                    x_i := \mat x[X_i],\\
                    \mat s_i := \mat x[\Src i]
                \end{array} \\
        &= \prod_{i=1}^n \,
            \bar\mu \big( U_i(\mat s_i) \,{=}\, x_i \big) \\
        &= \prod_{i=1}^n
            \mu(X_i = x_i \mid \Src i = \mat s_i) \\
        &= \mu(\X = \mat x).
    \end{align*}
    Therefore, when $\mu$ satisfies the independencies of a BN $G$, it is \scible\ with $\Ar_G$.
\end{lproof}

Before we move on to proving the other results in the paper, we first illustrate how this relatively substantial first half of the proof of \cref{theorem:bns} can be dramatically simplified by relying on two information theoretic arguments.

\begin{lproof}[Alternate, information-based proof]
    \textbf{($\implies$).}
    Let $G$ be a dag.     
    If $\mu \models \Diamond \Ar_{G}$, then by \cref{theorem:sdef-le0}, 
    $\IDef_{\!\Ar_{G}}(\mu) \le 0$. In the appendix of \cite{pdg-aaai}, it is shown that $\IDef_{\!\Ar_G}(\mu) \ge 0$ with equality iff $\mu$ satisfies the BN's independencies. 
    Thus $\mu$ must satisfy the appropriate independencies.
\end{lproof}

\recall{theorem:func}~%
\begin{lproof}\label{proof:func}
\textbf{(a).} 
The forward direction is straightforward. 
Suppose that $\mu \models \Diamond \Ar$ and $\mu \models X \tto Y$.
The former condition gives us a witness
$\nu(\X,\U)$ in which $\U = \{ U_a \}_{a \in \Ar}$ are mutually independent variables indexed by $\Ar$, that determine their respective edges. 
``Extend'' $\nu$ in the unique way to $n$ additional constant variables $U_1, \ldots, U_n$, each of which can only take on one value. 
We claim that this ``extended'' distribution $\nu'$, which we conflate with $\nu$ because it is not meaningfully different, is a witness to $\mu \models  \Diamond \ArXY n$.
Since $\mu \models X \tto Y$ it must also be that $\nu \models X \tto Y$, and it follows that $\nu \models (X,U_i)\tto Y$ for all $i \in \{1,\ldots,n\}$, demonstrating that
the new requirements of $\nu'$ imposed by \cref{defn:scompat}(c) hold. 
(The remainder of the requirements for condition (c), namely that $\nu' \models (\Src a, U_a) \tto \Tgt a$ for $a \in \Ar$, still hold because $\nu'$ is an extension of $\nu$, which we know has this property.)
Finally, since $\U$ are mutually independent and each $U_i$ is a constant (and hence independent of everything), the variables $\U' := \U \sqcup \{ U_i \}_{i=1}^n$ are also mutually independent. 
Thus $\nu$ (or, more precisely, an isomorphic ``extension'' of it to additional trivial variables) is a witness of
 $\mu \models  \Diamond \ArXY n.$
 
The reverse direction is difficult to prove directly, yet it is a straightforward application of \cref{theorem:sdef-le0}.
Suppose that $\mu \models \Diamond \ArXY n$ for all $n \ge 0$. 
By \cref{theorem:sdef-le0}, we know that
\[
    0 \ge  \IDef_{\!\ArXY n}(\mu) = \IDef_{\!\Ar}(\mu) + n \H_\mu(Y|X).
\]
Because $\IDef_{\!\Ar}(\mu)$ is bounded below (by $-\log |\V(\X)|$), it cannot be the case that $\H_\mu(Y|X) > 0$; otherwise, the inequality above would not hold for large $n$ (specifically, for $n > \log |\V(\X)| / \H_\mu(Y|X)$).  
By Gibbs inequailty, $\H_\mu(Y|X)$ is non-negative, and thus it must be the case that $\H_\mu(Y|X) = 0$. Thus $\mu \models X \tto Y$. 
It is also true that $\mu \models \Diamond \Ar$ by monotonicity (\cref{prop:strong-mono}), which is itself a direct application of \cref{theorem:sdef-le0}

\medskip

\textbf{(b).}
Now $\Ar = \Ar_G$ for some graph $G$.
The forward direction of the equivalence is strictly weaker than the one we already proved in part (a); we have shown $\mu \models \Diamond \ArXY n$ for all $n \ge 0$, and needed only to show it for $n=1$. 
The reverse direction is what's interesting.
As before, we will take a significant shortcut by using \cref{theorem:sdef-le0}. 
Suppose $\mu \models \Diamond \ArXY 1$. 
In this case where $\Ar = \Ar_G$, it was shown by \citet{pdg-aaai} that  $\IDef_{\!\Ar}(\mu) \ge 0$.
It follows that
\[
    0 ~\overset{\text{(\cref{theorem:sdef-le0})}}\ge~  \IDef_{\!\ArXY n}(\mu) = \IDef_{\!\Ar}(\mu) + \H_\mu(Y|X) 
    ~
    \ge~ 0,
\]
and thus $\H_\mu(Y|X) = 0$, meaning that $\mu \models X \tto Y$ as promised.  As before, we also have $\mu \models \Diamond \Ar$ by monotonicity.

\textbf{(c).}
As in part (b), the forward direction is a special case of the forward direction of part (a), and it remains only to prove the reverse direction.  
Equipped with the additional information that $\Ar \rightsquigarrow \{ \to \{ X \} \}$, 
suppose that $\mu \models \Diamond \ArXY 2$. 
By monotonicity, this means $\mu \models \Diamond \Ar$ and also that $\mu \models\begin{tikzpicture}[center base]
    \node[dpad0] (X) {$X$};
    \node[dpad0,right of=X] (Y) {$Y$};
    \draw[arr1] (X) to[bend left=15] (Y);
    \draw[arr1] (X) to[bend right=15] (Y);
    \draw[arr1,<-] (X) to[] +(-0.8,0);
\end{tikzpicture}$. 
Let $\Ar'$ denote this \hgraph. 
Once again by appeal to \cref{theorem:sdef-le0}, 
we have that 
\[
    0 \ge \IDef_{\!\Ar'} = - \H_\mu(X,Y) + \H(X) + 2 \H_\mu(Y|X)
         = \H_\mu(Y|X) \ge 0. 
\]
It follows that $\H_\mu(Y|X) = 0$, and thus $\mu \models X\tto Y$. 
As mentioned above, we also know that $\mu \models \Diamond \Ar$, 
and thus $\mu \models \Diamond \Ar ~\land~ X {\tto} Y$ as promised.
\end{lproof}

\newpage
\subsection{Causality Results of 
    \texorpdfstring{\cref{sec:causal}}{Section \ref*{sec:causal}}}
\recall{prop:sc-graph-arise}
\begin{lproof}\label{proof:sc-graph-arise}
    \textbf{($\implies$).~} Suppose $\mu \models \Diamond \Ar_G$. 
    Thus there exists some witness $\bar\mu(\X,\U)$ to this fact, satisfying conditions (a-c) of \cref{defn:scompat}. 
    Because $\Ar_G$ is partitional, the elements
    of $\PSEMs_{\Ar_G}(\bar\mu)$ are ordinary (i.e., not generalized) randomized PSEMs.
    We claim that every $\mathcal M = (M,P) \in \PSEMs_{\Ar_G}(\bar\mu)$ that is a randomized PSEM from which $\mu$ can arise,
    and also has the property that $\Pa_M(Y) \subseteq \Pa_G(Y) \cup \{U_Y\}$ for all $Y \in \X$. 
    \begin{itemize}[left=1em]
    \item 
    The \hyperarc s of $\Ar_G$ correspond to the vertices of $G$, which in turn correspond to the variables in $\X$; thus $\U = \{ U_X \}_{X \in \X}$. 
    By property (b) of \scibility\ witnesses (\cref{defn:scompat}), these variables $\{ U_X \}_{X \in \X}$ are mutually independent according to $\bar\mu$. 
    Furthermore, because $\mathcal M = (M, P) \in \PSEMs_{\Ar_{G}\!}(\bar\mu)$, we know that $\bar\mu(\U) = P$, 
    and thus the variables in $\U$ must be mutually independent according to $P$.
    By construction, in causal models $\mathcal M \in \PSEMs_{\Ar_G}(\bar\mu)$ the equation $f_Y$ can depend only on $\Src Y = \Pa_G(Y) \subseteq \X$ and $U_Y$. So, in particular, $f_Y$ does not depend on $U_X$ for $X \ne Y$.
    
    Altogether, we have shown that ${\cal M}$ contains exogenous variables $\{ U_X \}_{X \in \X}$ that are mutually independent according to $P$, and that $f_Y$ does not depend on $U_X$ when $X \ne Y$. 
    Thus, $\mathcal M$ is a randomized PSEM.

    \item 
    By condition (a) on \scibility\ witnesses (\cref{defn:scompat}), we know that $\bar\mu(\X) = \mu$. By \cref{prop:witness-model-properties}(a), we know that $\mu \in \SD{\cal M}$. 
    Together, the previous two sentences mean that $\mu$ can arise from $\mathcal M$.
    
    \item 
    Finally, as mentioned in the first bullet item, the equation $f_Y$ in $M$ can depend only on $\Src Y = \Pa_G(Y)$ and on $U_Y$. Thus
     $\Pa_M(Y) \subseteq \Pa_G(Y) \cup \{U_Y\}$ for all $Y \in \X$. 
    \end{itemize}
    Under the assumption that $\mu \models \Diamond \Ar_G$, we have now shown that there exists a randomized causal model $\cal M$ from which $\mu$ can arise, with the property that 
    $\Pa_{\cal M}(Y) \subseteq \Pa_G(Y) \cup \{ U_Y \}$ for all $Y \in \X$. 
\commentout{
    We have now shown that $\mathcal M$ is a randomized causal model from which $\mu$ can arise, with the property that 
    $\Pa_M(Y) \cap \X \subseteq \Pa_G(Y)$ for all $Y \in \X$. 
    To get an analogous statement with equality, we now introduce one additional variable $U_\bullet$, with $\V(\U_\bullet) :=  \{\bullet\} \cup 
        \prod_{Y \in \X} \V(Y)^{\V(\Pa_G(Y))}
        $.
    Take $P'(\U, U_\bullet) := P(\U) \delta(U_\bullet = \bullet)$ to be the unique distribution that extends $P$ with a dirac distribution stating that $U_\bullet$ always takes its special value $\bullet$. 
    Define new equations $\mathcal F' := \{ f'_Y \}_{Y \in \X}$
    according to 
    \[
    f'_Y(u_\bullet, \mat z) := 
    \begin{cases}
        & f_Y(\mat z) \text{ if } u_\bullet = \bullet \\
        & u_\bullet(Y)( \mat z[\Pa_G(Y)] ) \text{ if } u_\bullet \in \V(\X)
    \end{cases}.
    \]
    Let's break this down this notation. $u_\bullet$ is a value of $U_\bullet$, and hence a dependent function from variables $Y \in \X$ to functions of type $\V(\Pa_G(Y)) \to \V(Y)$. 
    Thus $u_\bullet(Y) : \V(\Pa_G(Y)) \to \V(Y)$.
    Next, $\mat z \in \V(\X \cup \U - Y)$, and so $\mat z[\Pa_M(Y)]$ consists of those components of $\mat z$ corresponding to the parents of $Y$ in the graph $G$. 
    So $u_\bullet(Y)( \mat z[\Pa_G(Y)] ) \in \V(Y)$ is a value of $Y$.
    Let $\mathcal M' = (M' = (\X, \U \cup \{ U_\bullet\}, \mathcal F'), P')$ denote this modified PSEM. 
    $M'$ still has the property that $\Pa_{M'}(Y) \cap \X \subseteq \Pa_{G}(Y)$, since $f'_Y$ depends on $U_\bullet$ and either the variables it previously depended on (when $u_\bullet = \bullet$) or on $\Pa_G(Y)$.
    But in $M'$, the endogenous variables on which $Y$ depends are precisely $\Pa_G(Y)$, not a subset.
    This is because, for all $X \in \Pa_G(Y)$, there is a tripple $(\mat u_\bullet, x, x')$ with $u_\bullet \in \V(U_\bullet)$ and $x,x' \in \V(X)$ such that $u_\bullet(Y)(\ldots, x, \ldots) \ne u_\bullet(Y)(\ldots, x', \ldots)$. 
    Therefore $\Pa_{M'}(Y) = \{ U_\bullet, U_Y \} \cup \Pa_G(Y)$,
        and hence $\Pa_{M'}(Y) \cap \X = \Pa_G(Y)$, as promised. 
}
    
    \medskip
    
    \textbf{($\impliedby$).~} Conversely, suppose there is a randomized PSEM 
    $\mathcal M = (M = (\mathcal Y,\U, \mathcal F), P)$
     with the property that
    $\Pa_M(Y) \subseteq \Pa_G(Y) \cup \{ U_Y \}$ for all $Y$,
    from which $\mu$ can arise. 
    The last clause means there exists some $\nu \in \SD{\cal M}$ such that $\nu(\X) = \mu$.
    We claim that this $\nu$ is a witness to $\mu \models \Diamond \Ar_G$. 
    We already know that condition (a) of being a \scibility\ witness is satisfied, since $\nu(\X) = \mu$. 
    Condition (b) holds because of the assumption that $\{ U_X \}_{X \in \X}$ are mutually independent in the distribution $P$ for a randomized PSEM (and the fact that $\nu(\U) = P$, since $\nu \in \SD{\mathcal M}$).
    Finally, we must show that (c) for each  $Y \in \X$, 
    $\nu \models \Pa_G(Y) \cup \{ U_Y \} \tto Y$. 
    Since $\nu \in \SD{\cal M}$, we know that $M$'s equation holds with probability 1 in $\nu$, and so it must be the case that $\nu \models \Pa_M(Y) \tto Y$. 
    Note that, in general, if $\mat A \subseteq \mat B$ and $\mat A \tto \mat C$, then $\mat B \tto \mat C$.
    By assumption, $\Pa_M(Y) \subseteq \Pa_G(Y) \cup \{U_Y\}$, and thus $\nu \models \Pa_G(Y) \cup \{U_Y\} \tto Y$.

    Thus $\nu$ satisfies all conditions (a-c) for a \scibility\ witness, and hence $\mu \models \Diamond \Ar_G$. 
\end{lproof}

\recall{prop:gen-sim-compat-means-arise}
\begin{lproof}\label{proof:gen-sim-compat-means-arise}
    \textbf{($\implies$).~} Suppose $\mu \models \Diamond \Ar$, meaning there exists a witness
    $\nu(\X,\U)$ with property \cref{defn:scompat}(c), meaning
    that, for all $a \in \Ar$, 
    there is a functional dependence $(\Src a, U_a) \tto \Tgt a$.
    Thus, there is some set of functions $\mathcal F$ with these types that holds with probability 1 according to $\nu$. 
    Meanwhile, by \cref{defn:scompat}(b), $\nu(\U)$ are mutually independent, so defining $P_a(U_a) := \nu(U_a)$, we have $\nu(\U) = \prod_{a \in \Ar}P_a(U_a)$. 
    Together, the previous two conditions (non-deterministically) define a 
    generalized randomized PSEM
    $\cal M$ of shape $\Ar$ for which $\nu \in \SD{\mathcal M}$. 
    Finally, by \cref{defn:scompat}(a), we know that $\mu$ can arise from $\mathcal M$.
    
    \textbf{($\impliedby$).~} Conversely, suppose there is a generalized randomized SEM $\mathcal M$ of shape $\Ar$ from which $\mu(\X)$ can arise. 
    Thus, there is some $\nu \in \SD{\mathcal M}$ whose marginal on $\X$ is $\mu$. 
    We claim that this $\nu$ is also a witness that $\mu \models \Diamond \Ar$.
    The marginal constraint from \cref{defn:scompat}(a) is clearly satisfied. 
    Condition (b) is immediate as well, because $\nu(\U) = \prod_{a} P_a(U_a)$. 
    Finally, condition (c) is satisfied, because the equations of $\mathcal M$ hold with probability 1, ensuring the appropriate functional dependencies. 
\end{lproof}

\recall{prop:witness-model-properties}
\begin{lproof}\label{proof:witness-model-properties}
    (a) is straightforward. 
    Suppose $\mathcal M \in \PSEMs(\nu)$. 
    By construction, the equations of $\mathcal M$ reflect functional dependencies in $\nu$, and hence hold with probability 1. 
    \unskip\footnote{When the probability of some combination of source variables is zero, there is typically more than one choice of functions that holds with probability 1; the choice of functions is essentially the choice of $\mathcal M \in \PSEMs(\nu)$.}
    Furthermore, the distribution $P(\U)$ in all $\mathcal M \in \PSEMs(\nu)$ is equal to $\nu(\U)$. 
    These two facts, demonstrate that $\nu$ satisfies the two constraints required for membership in $\SD{\mathcal M}$.

    (b). 
    We do the two directions separately. 
    First, suppose $\mathcal M \in \PSEMs(\nu)$.
    We have already shown (in part (a)) that $\nu \in \SD{\mathcal M}$. 
    The construction of $\PSEMs(\nu)$ depends on the hypergraph $\Ar$ (even if the dependence is not explicitly clear from our notation)
    in such a way that $f_X$ does not depend on any variables beyond $U_a$ and $\Src {a_{\!X}}$. Thus, $\Pa_{\mathcal M}(X) \subseteq \Src {a_{\!X}} \cup \{U_{a_{\!X}}\}$.
    
    Conversely, suppose $\mathcal M = (\X,\U, \mathcal F)$ is a PSEM satisfying $\nu \in \SD{\cal M}$ and $\Pa_{\cal M}(X) \subseteq \Src{a_{\!X}} \cup \{U_{a_{\!X}}\}$. 
    We would like to show that $\mathcal M \in \PSEMs(\nu)$. 
    Because $\nu \in \SD{\cal M}$, we know that the distribution $P(\U)$ over the exogenous variables in the PSEM $\cal M$ is equal to $\nu(\U)$, matching the first part of our construction.
    What remains is to show that the equations $\mathcal F$ are consistent with our transformation. 
    Choose any $X \in \X$. Because $\Ar$ is \subpartl, there is a unique $a_{\!X} \in \Ar$ such that $X \in \Tgt{ a_{\!X}}$. Now choose any values $\mat s \in \V(\Src {a_{\!X}})$ and $u \in \V(U_{a_{\!X}})$.
    If $\nu(\mat s, u) > 0$, then we know there is a unique value of $x \in \V(X)$ such that $\nu(\mat s, u, x) > 0$. Since $\mathcal M$'s equation for $X$, $f_{X}$, depends only on $\mat s$ and $u$, and holds with probability 1, we know that $f_X(\mat s, u) = t$, as required. 
    On the other hand, if $\nu(\mat s, u) = 0$, then any choice of $f_X(\mat s, u)$ is consistent with our procedure. 
    Since this is true for all $X$, and all possible inputs to the equation $f_X$, we conclude that the equations $\mathcal F$ can arise from the procedure described in the main text,
    and therefore $\mathcal M \in \PSEMs(\nu)$.     
\end{lproof}

\recall{theorem:condition-intervention-alignment}
\begin{lproof}\label{proof:condition-intervention-alignment}
    \textbf{(part a).}
    \def\doXx{\mathrm{do}_{\cal M}(\mat X{=}\mat x)}
    Because we have assumed $\bar\mu(\doXx) > 0$,
    the conditional distribution 
    \[ 
        \bar\mu\mid \doXx = \bar\mu(\U, \enV) \cdot 
        \prod_{X \in \mat X}
        \mathbbm 1\big[\forall \mat s. f_X(U_X, \mat s) = \mat x[X]\big]
        ~\Big/~ 
        \bar\mu(\doXx) 
    \]
    is defined. By assumption, $\mathcal M \in \PSEMs(\bar\mu)$ and
    $\bar\mu$ is a witness to the fact that $\mu \models \Diamond \Ar$.
    Thus, by \cref{prop:witness-model-properties}, we know that $\bar\mu \in \SD{\mathcal M}$. 
    So in particular, all equations of $\mathcal M$ hold for all joint settings $(\mat u, \mat v) \in \V(\U \cup \enV)$ in the support of $\bar\mu$. But the support of the conditional distribution $\bar\mu \mid \doXx$ is a subset of the support of $\bar\mu$, so all equations of $\mathcal M$ also hold in the conditioned distribution.  
    Furthermore, 
        the event $\doXx$ is the event in which, for all $X \in \mat X$, the variable $U_X$ takes on a value such that $f_X(\ldots, U_X,\ldots) = \mat x[X]$.  
    Thus the equations corresponding to $\mat X = \mat x$ also hold with probability 1 in $\bar\mu \mid \doXx$.

    This shows that all equations of $\mathcal M_{\mat X {\gets}\mat x}$ hold with probability 1 in $\bar\mu \mid \doXx$. However, the marginal distribution $\bar\mu(\U \mid \doXx)$ over $\U$ is typically not equal to $P(\U)$---after all, we have collapsed distribution of the variables $\U_{\mat X} := \{ U_X : X \in \mat X\}$. 
    Clearly, $\bar\mu \mid \doXx \notin \SD{\mathcal M_{\mat X {\gets}\mat x}}$.
    However, as we now show, there exists a \emph{different} distribution $\bar\mu' \in \SD{\mathcal M_{\mat X {\gets} \mat x}}$ such that $\bar\mu'(\enV) = \nu(\enV \mid \doXx)$. 

    Let $\U_0 := \U \setminus \U_{\mat X}$. 
    We can define $\bar\mu'$ according to
    \[
        \bar\mu'(\enV, \U_{\mat X}, \U_{0}) := 
            \bar\mu(\enV, \U_{0} \mid \doXx) P(\U_{\mat X}).
    \]
    The distribution $\bar\mu'$ satisfies 
    three critical properties:
    \begin{enumerate}
        \item 
        By construction, the marginal of $\bar\mu'$ on $\enV$ is $\bar\mu'(\enV) = \bar\mu(\enV \mid \doXx)$.
        \item At the same time, the marginal of $\bar\mu'$ on exogenous variables is
        \begin{align*}
            \bar\mu'(\U) &= 
            \bar\mu'(\U_{\mat X}, \U_{0})  \\
            &= \int_{\V(\enV)} \bar\mu'(\mat v, \U_{\mat X}, \U_0) \,\mathrm d \mat v \\
            &= \int_{\V(\enV)} \bar\mu(\mat v, \U_{0}\mid \doXx) P(\U_{\mat X}) \,\mathrm d\mat v
                &\text{[definition of $\bar\mu'$]} \\
            &= P(\U_{\mat X}) \bar\mu(\U_{0} \mid \doXx) 
                &\text{}\\
            &= P(\U_{\mat X}) P(\U_{0} \mid \doXx) 
                &\hspace{-2em}[\text{because $\bar\mu(\U) = P(\U)$, since $\bar\mu \in \SD{\mathcal M}$}]
                \\
            &= P(\U_{\mat X}) P(\U_{0}) 
                &\hspace{-2em}\left[~\parbox{5cm}{ since $\doXx$ depends only on $\U_{\mat X}$, while $\U_{\mat X}$ and $\U_{\mat Z}$ are independent in $\bar\mu$ (by the witness condition).}~\right]
                \\
            &= P(\U_{\mat X}, \U_{0})
                & [\text{  same reason as above }]
        \end{align*}
        \item Finally, $\bar\mu'$ satisfies all equations of $\mathcal M_{\mat X{\gets} \mat x}$. It satisfies the equations for the variables $\mat X$ because $\mat X = \mat x$ holds with probability 1. 
        At the same time, the equations of $\mathcal M_{\mat X{\gets} \mat x}$ corresponding to other variables, say $f_Z$ for $Z \in \enV \setminus \mat X$, also hold with probability one. 
        This is because the marginal $\bar\mu'(\U_{\mat Z}, \enV)$ is shared with the distribution $\bar\mu \mid \doXx$, and that distribution satisfies these equations. (It suffices to show that they share this particular marginal because the equations for $\mat Z$ do not depend on $\U_{\mat X}$.)
    \end{enumerate}
    Together, properties 2 and 3 show that $\bar\mu' \in \SD{\mathcal M_{\mat X{\gets} \mat x}}$, while property 1 shows that $\bar\mu(\enV \mid \doXx)$ can arise from $\mathcal M_{\mat X{\gets} \mat x}$. 
    This establishes part (a).

    \bigskip

    \textbf{(part b).}
    We will again make use of the distribution $\bar\mu'$ defined in part (a), and its three critical properties listed above. 
    Given a setting $\mat u \in \V(\U)$ of the exogenous variables, let
    \[
        \mathcal F_{\mat X \gets \mat x} (\mat u) 
        := \left\{ \omega \in \V(\enV) ~\middle|~
        \begin{array}{rl}
            \forall X \in \mat X.~& \omega[X] = \mat x[X]\\
            \forall Y \in \enV \setminus \mat X.~&\omega[Y] = f_X(\omega[\enV \setminus Y], \mat u) 
        \end{array}
        \right\} 
    \]
    denote the set of joint settings of endogenous variables that are consistent with the equations of $\mathcal M_{\mat X \gets \mat x}$. 

    If $\mat u \in \V(\U)$ is such that 
    \begin{align*}
        (M, \mat u) \models [\mat X {\gets} \mat x] \varphi
        \quad&\iff\quad
        (M_{\mat X {\gets} \mat x}, \mat u) \models \varphi  \\
        &\iff \forall \omega \in \mathcal F_{\mat X {\gets} \mat x}(\mat u).~~\omega \in \varphi \\
        &\iff 
            \mathcal F_{\mat X {\gets} \mat x}(\mat u)
                ~\subseteq~ \varphi, 
    \end{align*}
    then $\phi$ holds at all points that satisfy the equations of $M_{\mat X {\gets} \mat x}$. So, since $\bar\mu'$ is supported only on such points (property 3), it must be that $\bar\mu'(\varphi) = 1$. By property 1, 
    $\bar\mu'(\varphi) =  \bar\mu(\varphi \mid \doXx)$. 

    Furthermore, if $\bar\mu'(\varphi) > 0$, then there must exist some $\omega \in \mathcal F_{\mat X {\gets} \mat x}(\mat u)$ satisfying $\varphi$, and thus 
    $(M, \mat u) \models \langle \mat X {\gets}\mat x\rangle\varphi$. Putting both of these observations together, and with a bit more care to the symbolic manipulation, we find that:
    \begin{align*}
        \Pr_{\mathcal M}([\mat X {\gets} \mat x]\varphi) &= P(\{ 
            \mat u \in \V(\U) : (M, \mat u) \models [\mat X {\gets}\mat x]\varphi
        \})
        \\
        &= \sum_{\mat u \in \V(\U)} P(\mat u) \mathbbm 1
            \big[ \mathcal F_{\mat X \gets \mat x}(\mat u) \subseteq \varphi \big] \\
        &\le \sum_{\mat u \in \V(\U)} P(\mat u) \bar\mu'(\varphi \mid \mat u) 
        \qquad = \bar\mu'(\varphi) = \bar\mu(\varphi \mid \doXx) \\
        &\le \sum_{\mat u \in \V(\U)} P(\mat u) \mathbbm 1
        \big[ \mathcal F_{\mat X {\gets} \mat x}(\mat u) \cap \varphi \ne \emptyset \big] \\
        &= P(\{ 
            \mat u \in \V(\U) : (M, \mat u) \models \langle\mat X {\gets}\mat x\rangle\varphi
        \})
        \\&= \Pr_{\mathcal M}(\langle\mat X {\gets} \mat x\rangle\varphi),
        \qquad\text{    as desired. }
    \end{align*}

    Finally, if $\bar\mu \models \U \tto \enV$, then $\mathcal F_{\mat X {\gets} \mat x}(\mat u)$ is a singleton for all $\mat u$, and hence $\varphi$ holding for all $\omega \in \mathcal F_{\mat X {\gets} \mat x}$ and for some $\omega \in \mathcal F_{\mat X {\gets} \mat x}$ are equivalent. 
    So, in this case, 
    \[
        (M, \mat u) \models [\mat X{\gets}\mat x]\varphi
        \qquad \iff \qquad
        (M, \mat u) \models \langle\mat X{\gets}\mat x\rangle\varphi,
    \]
    and thus the probability of both formulas are the same---and it must also equal $\bar\mu(\varphi \mid \doXx)$ which we have shown lies between them. 
\end{lproof}

\commentout{
\recall{prop:all-models-formula-equiv}
\begin{lproof}\label{proof:all-models-formula-equiv}
    We prove the equivalence by induction on the structure of causal formulas $\varphi$. 
    \begin{itemize}
        \item \textbf{Case
        $\varphi = \varphi_1 \land \varphi_2$.} 
        Suppose inductively that $(\nu, \mat u) \models \varphi_1$ iff $(\mathcal M, \mat u) \models\varphi_1$ for all $\mathcal M \in \PSEMs(\nu)$, and the same for $\varphi_2$. 
        Then:
        \begin{align*}
            (\nu, \mat u) \models \varphi_1 \land \varphi_2
            \quad&\iff\quad
            (\nu, \mat u) \models \varphi_1  \quad\text{and}\quad
            (\nu, \mat u) \models \varphi_2 
            \\
            &\iff\quad
                \Big(\forall \mathcal M \in \PSEMs(\nu).~
                (\mathcal M, \mat u) \models \varphi_1 \Big)
                \text{ and }
                \Big(\forall \mathcal M \in \PSEMs(\nu).~
                (\mathcal M, \mat u) \models \varphi_2 \Big)
            \\
            &\iff\quad
            \forall \mathcal M \in \PSEMs(\nu).~
                \Big(
                (\mathcal M, \mat u) \models \varphi_1 \text{ and } 
                (\mathcal M, \mat u) \models \varphi_2 \Big)
            \\
            &\iff \quad
            \forall \mathcal M \in \PSEMs(\nu).~
            (\mathcal M, \mat u) \models \varphi_1 \land \varphi_2.
        \end{align*}
        \item \textbf{Case
        $\varphi = \lnot \varphi_1$.} 
        Suppose that the inductive hypthesis holds for $\varphi_1$. 
        \begin{align*}
            (\nu, \mat u) \models \lnot \varphi_1
                \quad& \iff\quad
                (\nu, \mat u) \not\models \varphi_1
                \\
                &\iff\quad
                \exists \mathcal M \in \PSEMs(\nu).~
                (\mathcal M, \mat u) \not\models \varphi_1.
                \\
                &\iff \quad 
                \exists \mathcal M \in \PSEMs(\nu).~
                (\mathcal M, \mat u) \not\models \varphi_1.
        \end{align*}
        {\color{red} oops, as expected, this doesn't work out.}
        
        \item \textbf{Case $\varphi = [Y {\gets}y]\varphi_1$}.
        \begin{align*}
            (\nu, \mat u) \models [Y {\gets}y]\varphi_1
            \quad
            &\iff \quad
            \nu(\varphi_1 \mid \U{=}\mat u[\hat U_{\mat Y} \gets 
                \Lambda Y. \lambda\textunderscore.\mat y[Y]]) = 1
        \end{align*}
        This means $\varphi_1$ holds in all 
        
    \end{itemize}
    
    $(\nu, \mat u) \models \varphi$. 
    By definition, this means   
\end{lproof}
}

\vfull{
\recall{prop:causal-mrf}
\begin{lproof}
    \label{proof:causal-mrf}
    1.
    Suppose $\mu \models \Diamond \Ar_G$, $\bar\mu(\U,\X)$ is a witness to this, and $\mathcal M \in \PSEMs_{\Ar_G \!}(\bar\mu)$. 
    
    \textbf{($\implies$).}
    Suppose $\mat X \CI_G \mat Y \mid \mat Z$, meaning every path from $\mat X$ to $\mat Y$ in $G$ goes through $\mat Z$.
    We now show that

    \TODO
    
\end{lproof}
}

\subsection{Information Theoretic Results of \texorpdfstring{\Cref{sec:info}}{Section \ref*{sec:info}}}
To prove \cref{theorem:sdef-le0} and \cref{theorem:siminc-idef-bounds}(a), we will need the following Lemma. 

\begin{lemma}
        \label{lem:Y-filter}
    Consider a set of variables $\mat Y = \{Y_1, \ldots, Y_n\}$, and another (set of) variable(s) $X$.
    Every joint distribution $\mu(X, \mat Y)$ over the values of $X$ and $\mat Y$ satisfies
    \[
        \sum_{i=1}^n \I_\mu(X \,;\, Y_i) 
        ~\le~
         \I_{\mu}(X \,;\, \mat Y) 
         + \sum_{i=1}^n \H_\mu(Y_i) - \H_\mu(\mat Y).
    \]
\end{lemma}
\begin{lproof}
    Since there is only one joint distribution in scope, we omit the subscript $\mu$, writing $\I(-)$ instead of $\I_\mu(-)$ and $\H(-)$ instead of $\H_\mu(-)$, in the body of this proof.
    The following fact will also be very useful:
    \begin{align}
        \I(A; B, C) &= \I(A; C) + \I(A; B \mid C)
            &\text{(the chain rule for mutual information).}
            \label{eq:chain-rule-mi}
    \end{align}

    We prove this by induction on $n$. In the base case ($n=1$), 
    we must show that $\I(X;Y)  \le \I(X;Y) + \H(Y) - \H(Y)$, which is an obvious tautology. 
    Now, suppose inductively that 
    \[
        \sum_{i=1}^k \I(X\,;\,Y_i) 
        ~\le~
         \I(X \,;\, \mat Y_{1:k}) + \sum_{i=1}^k \H(Y_i) - \H(\mat Y_{1:k})
         \tag{IH$_k$}
    \]
    for some $k < n$, where $\mat Y_{1:k} = (Y_1, \ldots, Y_k)$. 
    We now prove that the analogue for $k+1$ also holds. 
    Some calculation reveals:
    \begin{align*}
        \I(X &; Y_{k+1})
        \\
        &= \I(X ; \mat Y_{1:k+1}) - \I(X ; \mat Y_{1:k} \mid Y_{k+1}) 
            & \big[\text{ by MI chain rule \eqref{eq:chain-rule-mi} } \big]
        \\
            &\le \I(X ; \mat Y_{1:k+1})
            & \Big[\text{ since } \I(X ; \mat Y_{1:k} \mid Y_{k+1}) \ge 0\, \Big] \\
        &= \I(X;  Y_{k+1} \mid \mat Y_{1:k}) + \I(\mat Y_{1:k} ;  Y_{k+1})
            & \big[\text{ by MI chain rule \eqref{eq:chain-rule-mi} } \big] \\
        &= \left(
        \begin{array}{c}
        \I(X ; \mat Y_{1:k+1}) + \H(Y_{k+1}) - \H( \mat Y_{1:k+1}) \\
         - \I(X ; \mat Y_{1:k})  \qquad\quad\qquad + \H( \mat Y_{1:k} )           
     \end{array}\right)
        & \Big[
        \begin{array}{c}
        \text{ left: one more MI chain rule \eqref{eq:chain-rule-mi}; } \\
        \text{ right: defn of mutual information  }    
        \end{array}
        \Big].
    \end{align*}
    
    Observe: adding this inequality to our inductive hypothesis (IH$_k$) yields (IH$_{k+1}$)!
    So, by induction, the lemma holds for all $k$. \qedhere
\end{lproof}

\recall{theorem:sdef-le0}
\begin{lproof}
        \label{proof:sdef-le0}
    Suppose that $\mu \models \Diamond \Ar$, meaning that there is a witness $\nu(\X, \U)$ that extends $\mu$, and has properties (a-c) of \cref{defn:scompat}. 
    For each \arc\ a, since
    $\nu \models (\Src a, U_a) \tto \Tgt a$, 
    we have
    $\H_\nu(\Tgt a \mid \Src a, U_a) = 0$, and so
    \[
    \H_\mu(\Tgt a \mid \Src a) = \H_\nu(\Tgt a \mid \Src a, U_a) + \I_\nu(\Tgt a ; U_a \mid \Src a) = \I_\nu(\Tgt a; U_a \mid \Src a).
    \]
    Thus, we compute    
    \begin{align*}
        \sum_{a \in \Ar} \H_\mu(\Tgt a \mid \Src a) 
        &= \sum_{a \in \Ar} \I_\nu(U_a ; \Tgt a \mid \Src a) 
        \\
        &= \sum_{a \in \Ar} \I_\nu(U_a ; \Tgt a, \Src a) - \I_\nu(U_a ; \Src a)
            & \text{ by MI chain rule \eqref{eq:chain-rule-mi}} \\
        &\le \sum_{a \in \Ar} \I_\nu(U_a ; \Tgt a, \Src a) 
            & \text{ since $\I_\nu(U_a\,;\,\Src a) \ge 0$} 
        \\
        &\le \sum_{a \in \Ar} \I_\nu(U_a ; \X) 
            &\text{ since $\X \tto (\Src a, \Tgt a)$ }
        \\
        &\le \I_\nu(\X ; \U) + \sum_{a \in \Ar} \H_\nu( U_a ) - \H_\nu(\U)
            & \text{ by \cref{lem:Y-filter}}
        \\
        &= \I_\nu(\X ; \U) & \begin{array}{r}\text{ since $\U$ are independent}\\\text{ (per condition (b) of \cref{defn:scompat})} \end{array} \\
        &\le \H_\nu(\X) = \H_\mu(\X).
            &\text{(per condition (a) of \cref{defn:scompat})}
    \end{align*}
    Thus, $\IDef_{\!\Ar}(\mu) \le 0$. \qedhere
\end{lproof}

\recall{prop:sinc-nonneg-s2}
\begin{lproof}\label{proof:sinc-nonneg-s2}
    The first term in the definition of $\SIMInc$ be written as
    \[
        \Big( - \H_\nu(\mathcal U) + \sum_{a \in \Ar} \H_\nu(U_a)  \Big)
            = \Ex_{\nu} \Big[ \log \frac{\nu(\mathcal U)}{\prod_{a} \nu(U_a)} \Big]
    \]
    and is therefore the relative entropy between $\nu(\mathcal U)$ and 
    the independent product distribution $\prod_{a \in \Ar}\nu(U_a)$. 
    Thus, it is non-negative.
    The remaining terms of $\SIMInc_{\Ar}(\mu)$, are all conditional entropies, and hence non-negative as well.
    Thus $\SIMInc_{\Ar}(\mu) \ge 0$. 
    
    Now, suppose $\mu$ is s2-comaptible with $\Ar$, i.e., there exists
    some $\nu(\mathcal U, \X)$ such that 
    (a) $\nu(\X) = \mu(\X)$,
    (b) $\H_\nu(\Tgt a | \Src a, U_a) = 0$,
    and (d) $\{U_a\}_{a \in \Ar}$ are mutually independent.
    Then clearly $\nu$ satisfies the condition under the infemum,
    every $\H_\nu(\Tgt a | \Src a, U_a)$ is zero.
    It is also immediate that the final term is zero as well, because it equals
    $\kldiv{\nu(\mathcal U)}{\prod_a \nu(U_a)}$, and 
    $\nu(\mathcal U) = \prod_a \nu(U_a)$, per the definition of mutual independence.
    Thus, $\nu$ witnesses that $\SIMInc_{(\Ar, \lambda)} = 0$. 
    
    Conversely, suppose $\SIMInc_{(\Ar, \lambda)} = 0$.
    Because the feasible set is closed and bounded, as is the function,
    the infemum is achieved by some joint distribution $\nu(\X, \Ar)$ with marginal $\mu(\X)$.
    In this distribution $\nu$, we know that every $\H_\nu(\Tgt a | \Src a, U_a) = 0$
    and $\kldiv{\nu(\mathcal U)}{\prod_a \nu(U_a)} = 0$---
    because if any of these terms were positive, then the result
        would be positive as well.
    So $\nu$ satisfies (a) and (b) by definition. 
    And, because relative entropy is zero iff its arguments are identical
        we have $\nu(\mathcal U) = \prod_a \nu(U_a)$, so the $U_a$'s are mutually
        independent, and $\nu$ satisfies (d) as well.
\end{lproof}

\recall{theorem:siminc-idef-bounds}
\begin{lproof}
        \label{proof:siminc-idef-bounds}
    Part (a).  The left hand side of the theorem 
    ($\IDef_{\!\Ar}(\nu)\le\SIMInc_{\Ar}(\mu)$)
    is a strengthening of the argument used to prove \cref{theorem:sdef-le0}. 
    Specifically,
    let $\nu^*$ be a minimizer of 
        the optimization problem defining $\SIMInc$
    We calculate
    \allowdisplaybreaks
    \begin{align*}
        &\SIMInc_{\Ar}(\mu) - \IDef_{\!\Ar}(\mu) 
        \\
        &= \left( \,
        \sum_{a \in \Ar} \H_{\nu^*}(\Tgt a \mid \Src a, U_a) - \H_{\nu^*}(\mathcal U) + \sum_{a \in \Ar} \H_{\nu^*}(U_a)
        \right) - \left(\,
        \sum_{a \in \Ar}\H_\mu(\Tgt a \mid \Src a) - \H_\mu(\X) \right)\\
        &=
        \sum_{a \in \Ar} \Big( \H_{\nu^*}(\Tgt a \mid \Src a, U_a) - \H_{\nu^*}(\Tgt a \mid \Src a) \Big)
            + \H_\mu(\X)
            - \H_{\nu^*}(\mathcal U) + \sum_{a \in \Ar} \H_{\nu^*}(U_a)
        \\
        &=
        - \sum_{a \in \Ar} \I_{\nu^*}( \Tgt a ; U_a \mid \Src a)
            \qquad + \H_\mu(\X)
            - \H_{\nu^*}(\mathcal U) + \sum_{a \in \Ar} \H_{\nu^*}(U_a).
    \end{align*}
    The argument given in the first five lines of the proof of \cref{theorem:sdef-le0}, gives us a particularly convenient bound for the first group of terms on the left:
    \[
        \sum_{a \in \Ar} \I_{\nu^*} (U_a ; \Tgt a\mid \Src a) 
        \le \I_{\nu^*}(\X; \U) + \sum_{a \in \Ar} \H_{\nu^*}(U_a) - \H_{\nu^*}(\U).
    \]
    Substituting this into our previous expression, we have:
    \begin{align*}
        &\SIMInc_{\Ar}(\mu) - \IDef_{\!\Ar}(\mu) 
        \\
        &\ge - \Big(\I_{\nu^*}(\X; \U) + \sum_{a \in \Ar} \H_{\nu^*}(U_a) - \H_{\nu^*}(\U) \Big)
            + \H_\mu(\X) - \H_{\nu^*}(\U) + \sum_{a \in \Ar} \H_{\nu^*}(U_a)
        \\
        &= \H_\mu(\X) - \I_{\nu^*}(\X ; \U)
        \\
        &\ge 0. 
    \end{align*}
    The final inequality holds because of our assumption that the marginal $\nu^*(\X)$ equals $\mu(\X)$. 
    Thus, $\SIMInc_{\Ar}(\mu) \ge \IDef_{\!\Ar}(\mu)$, as proimised.

    We now turn to the right hand inequality, and part (b) of the theorem. 
    Recall that $\nu^*$ is defined to be a minimizer of the optimization problem defining $\SIMInc$. 
    For the right inequality $(\SIMInc_{\Ar}(\mu) \le \IDef_{\!\Ar^\dagger}(\nu))$ of part (a), observe that
    \begin{align*}
        \IDef_{\!\Ar^\dagger}(\nu)
            &=
                - \H_\nu(\X, \mathcal U)
                + \sum_{a \in \Ar} \H_{\nu}(U_a)
                + \sum_{a \in \Ar} \H_\nu(\Tgt a | \Src a, U_a)
                + \H_\nu(\X \mid \mathcal U) \\
            &= \Big( - \H_\nu(\mathcal U) + \sum_{a \in \Ar} \H_{\nu}(U_a) \Big)
            + \sum_{a \in \Ar} \H_\nu(\Tgt a | \Src a, U_a)
            \\&\ge \Big( - \H_{\nu^*}(\mathcal U) + \sum_{a \in \Ar} \H_{\nu^*}(U_a) \Big)
            + \sum_{a \in \Ar} \H_{\nu^*}(\Tgt a | \Src a, U_a)
            \\
            &= \SIMInc(\mu).
    \end{align*}
    This proves the right hand side of the inequality of part (a). 
    Moreover, because the one inequality holds with equality when $\nu = \nu^*$ is a minimizer of this quantity (subject to having marginal $\mu(\X)$) we have shown part (b) as well.

    \commentout{
    modulo the final claim about the size of the variables $\U$.
     To address that final missing piece, we claim that any minimizer $\nu(\U, \X)$ may be converted to another minimizing distribution $\hat \nu(\hat \U, \X)$ over the variables $\hat \U = \{ \hat U_a \}_{a \in \Ar}$, where $\V(U_a) = \{ \text{functions } \V(\Src a) \to \V(\Tgt a) \}$.
    To do this, we 
    
    \TODO
    }
\end{lproof}

\vfull{
\section{SIM-Equivalence}
    \label{sec:equivalence}
{\color{red} [now that we know this version of Theorem 2 is false, much of this discussion doesn't make sense.]
Applying \cref{theorem:func} with $\Ar' = \emptyset$
yeilds a quintessential special case: $\mu \models
\begin{tikzpicture}[center base]
    \node[dpad0] (X) {$X$};
    \node[dpad0,right=.5 of X] (Y) {$Y$};
    \draw[arr1,yshift=2px] (X) to[bend left=16] (Y);
    \draw[arr1] (X) to[bend right=16] (Y);
\end{tikzpicture}$ 
iff $Y$ is a function of $X$ according to $\mu$.
At first glance, this already seems to
    capture the essence of \cref{theorem:func};
    is it really meaningfully weaker? 
In fact it is; 
to illustrate, our next example is another graph that behaves the same way---but not in all contexts.

\begin{linked}{example}{det-fn}
     In the appendix, we prove
     $\mu \models \begin{tikzpicture}[center base]
         \node[dpad0] (X) {$X$};
         \node[dpad0,right of=X] (Y) {$Y$};
         \draw[arr1] (X) to[] (Y);
         \draw[arr1,<-] (Y) to[] +(0.7,0);
    \end{tikzpicture}$
    iff
     $Y$ is a function of $X$
    (according to $\mu$).
    But, in general, this graph says something
        distinct from 
        (and stronger than, as we will see in \cref{sec:monotone})
    the example above. 
    After an adding the \arc\ $\ed{}{\emptyset}{\{X\}}$ to both graphs, for example, they behave differently:
    every distribution $\mu$ satisfying $X \tto Y$
    also satisfies
    $\mu \models \begin{tikzpicture}[center base]
        \node[dpad0] (X) {$X$};
        \node[dpad0,right of=X] (Y) {$Y$};
        \draw[arr1] (X) to[bend left=15] (Y);
        \draw[arr1] (X) to[bend right=15] (Y);
        \draw[arr1,<-] (X) to[] +(-0.8,0);
    \end{tikzpicture}$,
    but only when $Y$ is a constant
    can it be the case that
    $\mu \models \begin{tikzpicture}[center base]
        \node[dpad0] (X) {$X$};
        \node[dpad0,right of=X] (Y) {$Y$};
        \draw[arr1] (X) to[] (Y);
        \draw[arr1,<-] (Y) to[] +(0.7,0);
        \draw[arr1,<-] (X) to[] +(-0.7,0);
   \end{tikzpicture}$.
\end{linked} 
}

\commentout{
    For a more general illustration, one can easily show that
    every \hgraph\ consisting of just one \arc\ is consistent with all probability measures $\mu \in \Delta \V\!\X$,
        even though different \arc s intuitively mean different things.
    Moreover, every \hgraph\ is a union of one-arc \hgraph s,
        and we have already seen that not all \hgraph s are equivalent.
    We explore role of union as a way of combining   
        qualitative PDGs further in \cref{sec:union}.
}

To distinguish between \hgraph s that are not interchangable,
    we clearly need a stronger notion of equivalence.
    
Given \hgraph s $\Ar_1$ and $\Ar_2$, 
we can form the combined \hgraph\ $\Ar_1 + \Ar_2$
that consists of the disjoint union of the two sets of \hyperarc s,
    and the union of their nodes.
We say that $\Ar$ and $\Ar'$ are \emph{(structurally) equivalent}
($\Ar \cong \Ar'$) if for every context $\Ar''$ and distribution $\mu$,  we have that
    $\mu \models \Diamond( \Ar + \Ar'')$ iff $\mu \models \Diamond( \Ar' + \Ar'')$.
By construction, structural equivalence ($\cong$) is itself
invariant to additional context:
     if $\Ar \cong \Ar'$ then
    $\Ar + \Ar'' \cong \Ar' + \Ar''$.
Our next result is a simple, intuitive, and particularly useful equivalence.

\begin{prop}
        \label{prop:equiv-factorizations-cnd}
        The following \hgraph s are equivalent:
    \[
        \begin{tikzpicture}[center base]
            \node[dpad0] (X) {$X$};
            \node[dpad0,above right=0.8em and -0.4em of X] (Z) {$Z$};
            \node[dpad0,right=1.2em of X] (Y) {$Y$};
            \mergearr[arr1] {Z}{X}{Y};
            \draw[arr1] (Z) to (X);
        \end{tikzpicture}
        ~~\cong~~
        \begin{tikzpicture}[center base]
            \node[dpad0] (X) {$X$};
            \node[dpad0,right=0.6 of X] (Y) {$Y$};
            \node[dpad0,anchor=center] (Z) at ($(X.east)!0.5!(Y.west) + (0,0.75)$){$Z$};
            \unmergearr[arr1] {Z}{X}{Y};
        \end{tikzpicture}
        ~~\cong~~
        \begin{tikzpicture}[center base]
            \node[dpad0] (X) {$X$};
            \node[dpad0,above left=0.8em and -0.4em of Y] (Z) {$Z$};
            \node[dpad0,right=1.2em of X] (Y) {$Y$};
            \mergearr[arr1] {Z}{Y}{X};
            \draw[arr1] (Z) to (Y);
        \end{tikzpicture}
    .
    \]
\end{prop}
These three \hgraph s correspond, respectively, to equivalent factorizations 
of a conditional probability measure
\[ 
    \def\Pr{P}
    \Pr(X|Z)\Pr(Y|X,Z) 
        = \Pr(X,Y|Z) 
        = \Pr(X|Y,Z)\Pr(Y|Z).
\]
\commentout{
    \begin{prop}
            \label{prop:equiv-factorizations}
            The following \hgraph s are equivalent:
        \[
            \begin{tikzpicture}[center base]
                \node[dpad0] (X) {$X$};
                \node[dpad0,right=1.1em of X] (Y) {$Y$};
                \draw[arr1] (X) to (Y);
                \draw[arr1,<-] (X) to +(-0.68,0);
            \end{tikzpicture}
            ~~\cong~~
            \begin{tikzpicture}[center base]
                \node[dpad0] (X) at (0,0) {$X$};
                \node[dpad0] (Y) at (1,0) {$Y$};
                \cunmergearr[arr1] {0.5,.8}{X}{Y}{0.5,.5}
            \end{tikzpicture}
            ~~\cong~~
            \begin{tikzpicture}[center base]
                \node[dpad0] (X) {$X$};
                \node[dpad0,right=1.1em of X] (Y) {$Y$};
                \draw[arr1] (Y) to (X);
                \draw[arr1,<-] (Y) to +(0.68,0);
            \end{tikzpicture}
        .
        \]
    \end{prop}
    These three \hgraph s also arise from equivalent BN structures.
    As such, they correspond, respectively, to equivalent factorizations 
    of a probability measure
    \[ \Pr(X)\Pr(Y|X) ~=~ \Pr(X,Y) ~=~ \Pr(X|Y)\Pr(Y). \]
}
\commentout{
We conjecture that, in a sense, \cref{prop:equiv-factorizations}
characterizes when two qualitative BNs satisfy the same set of
independencies.  More specifically, we conjecture that if 
two qualitative Bayesian Networks describe the same set of
independencies, then they can 
be provd equivalent using only instances of \cref{prop:equiv-factorizations}.
}
\cref{prop:equiv-factorizations-cnd}
provides a simple and useful way to relate \scibility\ of different \hgraph s. If we restrict to acyclic structures, for instance, we find: 

\begin{theorem}[\citealt{chickering-equiv-bns}]
        \label{theorem:bn-completeness}
    Any two qualitative Bayesian Networks that 
    represent the same independencies
    can be proven equivalent using only instances of \cref{prop:equiv-factorizations-cnd}
    (in which $X$, $Y$, $Z$ may be sets of variables).
\end{theorem}

\Cref{theorem:bn-completeness} 
is essentially a restatement of
main result of \citet{chickering-equiv-bns}, 
but it is simpler to state in terms of \dhygraph\ equivalences.
To state the result in its original form, one has to first
define an edge $X \to Y$ to be \emph{covered} in a graph $G$ iff $\Pa_G(Y) = \Pa_G(X) \cup \{X\}$;
then, the result states that all equivalent BN structures are related by a chain of reversed covered edges. 
Observe that this notion of covering is implicit in \cref{theorem:bn-completeness}.
\Cref{theorem:bn-completeness} is one demonstration of the usefulness of 
\cref{prop:equiv-factorizations-cnd},
but the latter applies far more broadly, to cyclic structures
and beyond. 
It becomes even more useful in tandem
    with the definition of monotonicity presented in \cref{sec:monotone},
    which is an analogue of implication. 
    
\commentout{
    Suppose we want
        to know that \scibility\ with $\Ar$ suffices to guarantee compatibility with $\Ar'$, but are not interested in the converse.
    As we will see in the next section, there is such a notion, and \cref{prop:equiv-factorizations-cnd} becomes even more useful in tandem with it.
}
}

\section{Monotonicity and Undirected Graphical Models}
    \label{sec:monotone}
    \label{appendix:undirected PGMs}

\commentout{
 The next two sections build up some theory that is necessary to understand the implications of our definition of \scibility\ for cyclic models---which turn out to be rather complicated. 
 }
\commentout{
    There are simple rules for manipulating
    PDGs,
    which can be used as an axiomatic proof system. 
}
\commentout{
    Holding beliefs $p(X)$ and $p(Y|X)$, 
    for example, is observationally equivalent to holding a single belief $p(X,Y)$ or beliefs $p(Y)$ and $p(Y|X)$.
}
\vfull{%
Monotonicity of PDG inconsistency \citep{one-true-loss}
is a powerful reasoning principle. Many important inequalities (e.g., the data processing inequality, relationships between statistical distances, the evidence lower bound, \ldots) can be proved using only a simple inference rule: ``more beliefs can only increase inconsistency''.
In this section, we develop and apply an anlogous principle for \scibility.
But first, we start with something simple.
}
The fact that (quantitative) PDG inconsistency  is monotonic
is a powerful reasoning principle that can be used to prove many important inequalities \citep{one-true-loss}. 
In this section, we develop a related principle for \scibility.
\vfull{
One classical representation of knowledge is a list of formulas
$[ \phi_1, \phi_2, \ldots, \phi_n]$
that one knows to be true.
This representation has a nice property:
learning an additional formula $\phi_{n+1}$
can only narrow the set of worlds one considers possible.
The same is true of \scibility. 
}%
\vfull{%
\begin{prop}
        \label{prop:mono}
    If $\Ar \subseteq \Ar'$ and 
    $\mu \models \Diamond\Ar'$, then $\mu \models \Diamond\Ar$.
\end{prop}
}%
Here is a direct but not very useful analague: if $\Ar \subseteq \Ar'$ and $\mu \models \Diamond\Ar'$, conclude $\mu \models \Diamond\Ar$. 
After all, if 
$\mu$ is consistent with a set of independent causal mechanisms, then surely 
it is
consistent with a causal picture
in which
only a subset of those mechanisms 
are
present and independent.  
{%
There is a sense in which 
    BNs and MRFs are also monotonic,
    but in the opposite direction:
    adding edges to a graph results in a weaker
    independence statement.
    We will soon see why.
}%

\begin{wrapfigure}[5]{o}{1.6cm}
    \begin{tikzcd}[row sep=3ex,column sep=0.2em]
        A \ar[rr,"f"] &\ar[d,squiggly]& B\ar[d,hook,gray]  \\
        A'  \ar[u,hook,gray]\ar[rr, dashed]&{\vphantom{a}}& B'
    \end{tikzcd}
\end{wrapfigure}
Since we use \emph{directed} hypergraphs,
     there is actually a finer notion of monotonicity at play. 
Inputs and ouputs play opposite roles, 
    and they are naturally monotonic in opposite directions. 
If there is an obvious way to regard an element of $B$ as an element of $B'$ (abbreviated $B \,{\color{gray}\hookrightarrow}\, B'$),
and $A' \,{\color{gray}\hookrightarrow}\, A$, 
then
a function $f : A \to B$ can be regarded as one of type $A'\to B'$.
This is depicted to the right.
The same principile applies in our setting.
If $\mat X$ and $\mat Z$ are sets of variables and $\mat X \subseteq \mat Z$, then 
$\V(\mat Z) {\color{gray}\hookrightarrow} \V(\mat X)$, by restriction. 
It follows, for example, that
    any mechanism by which $X$ determines $(Y,Y')$ can be viewed
        as a mechanism by which $(X,X')$ determines $Y$. 
The general phenomenon is captured by the following
    \unskip.

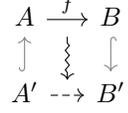
\begin{defn}
        \label{defn:weakening}
    \commentout{%
        Consider a directed hypergraph
        $\Ar = \{ \Src a \to \Tgt a \}_{a \in \Ar}$.
        If another graph $\Ar'$ can be obtained by adding sources to and removing targets from the \arc s of $\Ar$---that is, if $\Ar' = \{ \Src a{\!}' \to \Tgt a{\!}' \}_{a \in \Ar}$, with $\Tgt a{\!}' \subseteq \Tgt a$ and
         $\Src a{\!}' \supseteq \Src a$ for all $a \in \Ar$---then we
         say $\Ar'$ is a \emph{weakening} of $\Ar$ and
         write $\Ar \rightsquigarrow \Ar'$.
     }%
     If $\Ar = \{ \ed aST  \}_{a}$, 
     $\Ar' = \{ \ed{a'}{S'}{T'} \}_{a'}$,
     and there is an injective map $\iota : \Ar' \!\to\! \Ar$
     such that
     $\Tgt {a}{\!}' \subseteq \Tgt {\iota(a)}$ and
     $\Src {a}{\!}' \supseteq \Src {\iota(a)}$
     for all $a \in \Ar'$,
     then $\Ar'$ is a \emph{weakening} of $\Ar$
     (written $\Ar {\rightsquigarrow} \Ar'$).
\end{defn}

\vfull{
\scibility\ is monotonic with respect to weakening  ($\rightsquigarrow$).}
    
\begin{prop}
        \label{prop:strong-mono}
    If
    $\Ar \rightsquigarrow \Ar'$ 
    and $\mu \models \Diamond \Ar$, then  $\mu \models \Diamond \Ar'$.
    \commentout{%
        and $\mu$ is 
        \scible\ with $\Ar$,
        then $\mu$ is also 
        \scible\ with $\Ar'$. 
    }%
\end{prop}

\Cref{prop:strong-mono} is strictly stronger than the simple monotonicity mentioned at the beginning of the section,
because a \arc\ with no targets is vacuous, so
removing all targets of a \arc\ is equivalent to deleting it.
It also
explains why BNs and MRFs are arguably \emph{anti}-monotonic: adding $X \to Y$ to a graph $G$ means adding $X$ to the \emph{sources} the \arc\ whose target is $Y$, 
in $\Ar_G$.

As mentioned in the main body of the paper, 
the far more important consequence of this result is that it
helps us begin to understand what \scibility\ means for cyclic \hgraph s.
For the reader's convenience, we now restate the examples in the main text,
which are really about monotonicity..

\medskip
\textbf{\cref{example:xy-cycle}.}
    Every $\mu(X,Y)$ is
    \cible\ with
    \begin{tikzpicture}[center base]
        \node[dpadinline] (X) {$X$};
        \node[dpadinline,right=1.2em of X] (Y) {$Y$};
        \draw[arr1] (X) to[bend left=13] (Y);
        \draw[arr1] (Y) to[bend left=13] (X);
    \end{tikzpicture}.
    This is because this cycle is weaker than
    a \hgraph\ 
    that can already represent any distribution, i.e.,
    \begin{tikzpicture}[center base]
        \node[dpadinline] (X) {$X$};
        \node[dpadinline,right=1.0em of X] (Y) {$Y$};
        \draw[arr1] (X) to (Y);
        \draw[arr1,<-] (X) to +(-0.68,0);
    \end{tikzpicture}
    $~\rightsquigarrow~$
    \begin{tikzpicture}[center base]
        \node[dpadinline] (X) {$X$};
        \node[dpadinline,right=1.2em of X] (Y) {$Y$};
        \draw[arr1] (X) to[bend left=13] (Y);
        \draw[arr1] (Y) to[bend left=13] (X);
    \end{tikzpicture}~. \qedhere
\hfill$\triangle$.
\medskip

\settowidth{\cycleboxlen}{\usebox{\cyclebox}}
\begin{wrapfigure}[5]{o}{0.8\cycleboxlen}
    \vspace{-0.8em}
    \begin{tikzpicture}[center base, scale=0.8]
        \node[dpad0] (X) at (0:.8) {$X$};
        \node[dpad0] (Y) at (120:.8) {$Y$};
        \node[dpad0] (Z) at (-120:.8) {$Z$};
        \draw[arr2] (X) to 
            (Y);
        \draw[arr2] (Y) to
            (Z);
        \draw[arr2] (Z) to 
            (X);
    \end{tikzpicture}
\end{wrapfigure}
\textbf{\cref{example:xyz-cycle-1}.~} 
    What
    $\mu(X,Y,Z)$
    are \cible\ 
    with 
    the 3-cycle shown, on the right?
    By monotonicity,
    among them must be all distributions consistent with a linear chain ${\to}X{\to}Y{\to}Z$. Thus,  
        any distribution
    in which two variables are conditionally independent given the third
    is compatible with 
    the 3-cycle.
    Are there any distributions that are \emph{not} compatible with 
    this hypergraph? It is not obvious.
    We return to this
      in \cref{sec:pdgs}. 
\hfill$\triangle$

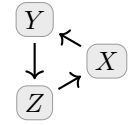
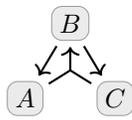
\begin{wrapfigure}[5]{i}
        {1.7cm}
    \vspace{-1em}
    \centering
    \begin{tikzpicture}[center base]
        \node[dpad0] (A) at (-0.6,0) {$A$};
        \node[dpad0] (B) at (0,1) {$B$};
        \node[dpad0] (C) at (0.6,0) {$C$};

        \draw[arr2,<-] (A) -- (B);
        \draw[arr2,<-] (C) -- (B);
        \mergearr[arr2] ACB
    \end{tikzpicture}
\end{wrapfigure}
Because \scibility\ applies to cyclic structures,  one might wonder if
    it also captures the independencies of undirected models 
    \unskip.
Undirected edges $A {-} B$ are commonly identified
with a (cylic) pair of directed edges $\{ A{\to}B, B{\to}A\}$,
as we have implicitly done in
defining $\Ar_G$. 
In this way, undirected graphs, too, naturally correspond to \dhygraph s.
For example, 
$G = A{-}B{-}C$
corresponds to the \hgraph\ 
\commentout{
\[
    \Ar_{G} = \Bigg\{
            \begin{array}{r@{}l} \{B\}&{\to}\{A\},\\ \{A,C\}&{\to}\{B\},\\ \{B\}&{\to}\{C\}
            \end{array}\Bigg\}
    =     \begin{tikzpicture}[center base]
        \node[dpad1] (A) at (-0.6,0) {$A$};
        \node[dpad1] (B) at (0,1) {$B$};
        \node[dpad1] (C) at (0.6,0) {$C$};

        \draw[arr2,<-] (A) -- (B);
        \draw[arr2,<-] (C) -- (B);
        \mergearr[arr2] ACB
    \end{tikzpicture}
    ~.
\]
}
$\Ar_G$ shown on the left.
Compatibility
with $\Ar_G$, however, does not coincide with any of the standard Markov properties
corresponding to $G$ \citep{koller2009probabilistic}.
This may appear to be a flaw in \cref{defn:scompat} (\scibility), but it is unavoidable.  
While both BNs and MRFs are monotonic, it is impossible to capture both classes with a monotonic definition.

\begin{theorem}\label{theorem:mrf-bn-monotone-impossible}
    It is possible to define a  relation  $\dotmodels$ 
        between distributions $\mu$ and \dhygraph s $\Ar$
    satisfying 
    any two, but not all three, of the
    following.
    \begin{description}[itemsep=0pt,parsep=0.3ex,topsep=0pt]
        \item [\rm(monotonicity)]
            If $\mu \dotmodels \Ar$
            and $\Ar \rightsquigarrow \Ar'$,
            then
            $\mu \dotmodels \Ar'$.
        \item [\rm(positive BN capture)]
            If $\mu$ satisfies the independencies
            $\mathcal I(G)$ of a dag $G$,
            then
            $\mu \dotmodels \Ar_G$.
        \item [\rm(negative MRF capture)]
            If $\mu \dotmodels \Ar_G$ for an undirected directed graph $G$,
            then $\mu$ has one of the Markov properties 
                with respect to $G$.
    \end{description}
\end{theorem}

\label{proof:mrf-bn-monotone-impossible}
The proof is
a direct and easy-to-visualize application
of monotonicity (\cref{prop:strong-mono}).
Assume montonicity and positive BN capture. 
Let $\muxor(A,B,C)$ be the joint distribution in which 
$A$ and $C$ are independent fair coins, and
$B = A \oplus C$ is their parity.
We then have:
\\$
\muxor
\models
\begin{tikzpicture}[center base,scale=0.8]
    \node[dpadinline] (A) at (-0.6,0) {$A$};
    \node[dpadinline] (B) at (0,1) {$B$};
    \node[dpadinline] (C) at (0.6,0) {$C$};

    \draw[arr2,<-] (A) -- +(0,1);
    \draw[arr2,<-] (C) -- +(0,1);
    \mergearr[arr2] ACB
\end{tikzpicture}
~~\rightsquigarrow~~
\begin{tikzpicture}[center base,scale=0.8]
    \node[dpadinline] (A) at (-0.6,0) {$A$};
    \node[dpadinline] (B) at (0,1) {$B$};
    \node[dpadinline] (C) at (0.6,0) {$C$};

    \draw[arr2,<-] (A) -- (B);
    \draw[arr2,<-] (C) -- (B);
    \mergearr[arr2] ACB
\end{tikzpicture}
~=~ \Ar_{A{-}B{-}C}.
$
~~
But
$\muxor \not\models A \CI C \mid B$.
\hfill\qedsymbol

We emphasize that \cref{theorem:mrf-bn-monotone-impossible} has implications for the qualitative semantics of \emph{any} graphical model (even if one were to reject the definition \scibility). 
We now look into the implications for some lesser-known graphical models, which may appear not to comply with \cref{theorem:mrf-bn-monotone-impossible}.

\paragraph{Dependecny Networks}
To readers familiar with \emph{dependency networks (DNs)} \citep{heckerman2000dependency},
\cref{theorem:mrf-bn-monotone-impossible} may raise some conceptual issues.
When $G$ is an undirected graph, $\Ar_G$ is the structure of a consistent DN.
The semantics of such a DN,
which intuitively describe an independent mechanism on each \arc,
coincide with the MRFs for $G$ (at least for positive distributions). 
In more detail, DN semantics are given by the fixed point of a markov chain that repeatedly generates independent samples along the \arc s of $\Ar_G$ for some (typically cyclic) directed graph $G$. The precise definition requires an order in which to do sampling. Although this choice doesn't matter for the ``consistent DNs'' that represent MRFs, it does in general. With a fixed sampling order, the DN is monotonic and captures MRFs, but can represent only BNs for which that order is a topological sort.

\vfull{
\begin{vnew}
\Cref{theorem:mrf-bn-monotone-impossible} shows that
    \scibility\ does not capture MRFs (at least, in the obvious way) at a purely observational level. 
Nevertheless,
    there is still a sense in which \scibility\ 
    captures MRFs \emph{causally}---that is, 
    if we \emph{intervene} instead of conditioning.

\begin{linked}{prop}{causal-mrf}
    Let $G$ be an undirected graph whose vertices correspond to variables $\X$. 
    \begin{enumerate}
        \item Let $\mu(\X)$ be a positive distribution (i.e., $\forall \mat x \in \V(\X).~\mu(\X{=}\mat x) > 0$).
        If $\mu \models \Diamond \Ar_G$, then 
        for every witness $\bar\mu$ and causal model $\mathcal M \in \PSEMs_{\Ar_G}(\bar\mu)$,
        whenever $\mat X, \mat Y, \mat Z \subseteq \X$ are such that  $\mat X \CI_G \mat Y \mid \mat Z$, 
        it is the case that
        $\bar\mu \models \mat X \CI \mat Y \mid \mathrm{do}_{\mathcal M}(\mat Z = \mat z)$.
        
        \item
        Convesely, there exists some distribution $\mu(\X)$
        If $\mu \not \models \Diamond \Ar_G$, then 
        
    \end{enumerate}
\end{linked}
\end{vnew}
}

\section{Information Theory, PDGs, and \SCibility}
\subsection{More Detailed Primer on Information Theory}
    \label{appendix:info-theory-primer}
We now expand on the fundemental information quantities introduced at the beginning of \cref{sec:info}.
Let $\mu$ be a probability distribution, and be $X,Y,Z$
    be (sets of) discrete random variables.
The \emph{entropy} of $X$ is the uncertainty in $X$, when it is distributed according to $\mu$, as measured by the number of bits of information needed (in expectation) needed to determine it, if the distribution $\mu$ is known.  It is given by 
\[
    \H_\mu(X) := \sum_{x \in \V(X)} \mu(X{=}x) \log \frac{1}{\mu(X{=}x)} \qquad= -\Ex_{\mu}[\log \mu(X)],
\]
and a few very important properties; chief among them, $\H_\mu(X)$ is non-negative, and equal to zero iff $X$ is a constant according to $\mu$. 
The ``joint entropy'' $\H(X,Y)$ is just the entropy of the combined variable $(X,Y)$ whose values are pairs $(x,y)$ for $x \in \V(X),y \in \V(Y)$; this is the same as the entropy of the variable $X \cup Y$ when $X$ and $Y$ are themselves sets of variables. 

The \emph{conditional entropy} of $Y$ given $X$
measures the uncertainty present in $Y$ if one knows the value of $X$
(think: the information in $Y$ but not $X$),
and is equivalently defined as any of the following three quantities:
\[
\H_\mu( Y | X) :=
        \quad
    \Ex_{\mu} [~\log \nicefrac1{\mu(Y | X)}~]
        \quad
    =\H_\mu(X,Y) - \H_\mu(X)
        \quad
    =\Ex_{x \sim \mu(X)} [~\H_{\mu \mid X{=}x}(Y)~]    
.
\]
The \emph{mutual information} $\I(X;Y)$,
and its conditional variant $\I(X;Y|Z)$, 
are given, respectively, by
\[
    \I_\mu(X;Y) :=
        \Ex_{\mu} \Big[ \log \frac{\mu(X,Y)}{\mu(X) \mu(Y)}\Big],
    \quad\text{and}\quad
    \I(X;Y|Z):= 
        \Ex_{\mu} \Big[ \log \frac{\mu(X,Y,Z)\mu(Z)}{\mu(X,Z) \mu(Y,Z)}\Big].
\]
The former is non-negative and equal to zero iff $\mu \models X \CI Y$, and the latter is non-negative and equal to zero iff $\mu \models X \CI Y \mid Z$. 
All of these quantities are purely ``structural'' or ``qualitative'' in the sense that they are invariant to relabelings of values, and 

Just as conditional entropy can be written as a linear combination of unconditional entropies, so too can conditional mutual information be written as a linear combination of unconditional mutual informations: $\I(X;Y|Z) = \I(X;(Y,Z)) - \I(X;Z)$.  
Thus conditional quantities are easily derived from the unconditional ones. But at the same time, the unconditional versions are clearly special cases of the conditional ones; for example, $\H_\mu(X)$ is clearly the special case of $\H(X|Z)$ when $Z$ is a constant (e.g., $Z = \emptyset$). 
Furthermore, entropy and mutual information are also interdefinable and generated by linear combinations of one another. 
It is easy to verify that 
$\I_\mu(X;Y) 
    = \H_\mu(X) + \H_\mu(Y) - \H(X,Y)
$
and 
$\I_\mu(X;Y|Z)
    = \H_\mu(X|Z) + \H_\mu(Y|Z) - \H(X,Y|Z)
$,
and thus mutual information is derived from entropy. 
Yet on the other hand, $\I_\mu(Y;Y) = \H_\mu(Y)$ and $\I_\mu(Y;Y|X) = \H_\mu(Y|X)$---thus entropy is a special case of mutual information. 
    
\subsection{Structural Deficiency: More Motivation, and Examples}
To build intuition for $\IDef$, which characterizes our bounds in \cref{sec:info},
	we now visualize the vector $\mat v_{\!\Ar}$ for various example hypergraphs. 

\relax%
	\definecolor{subfiglabelcolor}{RGB}{0,0,0}
	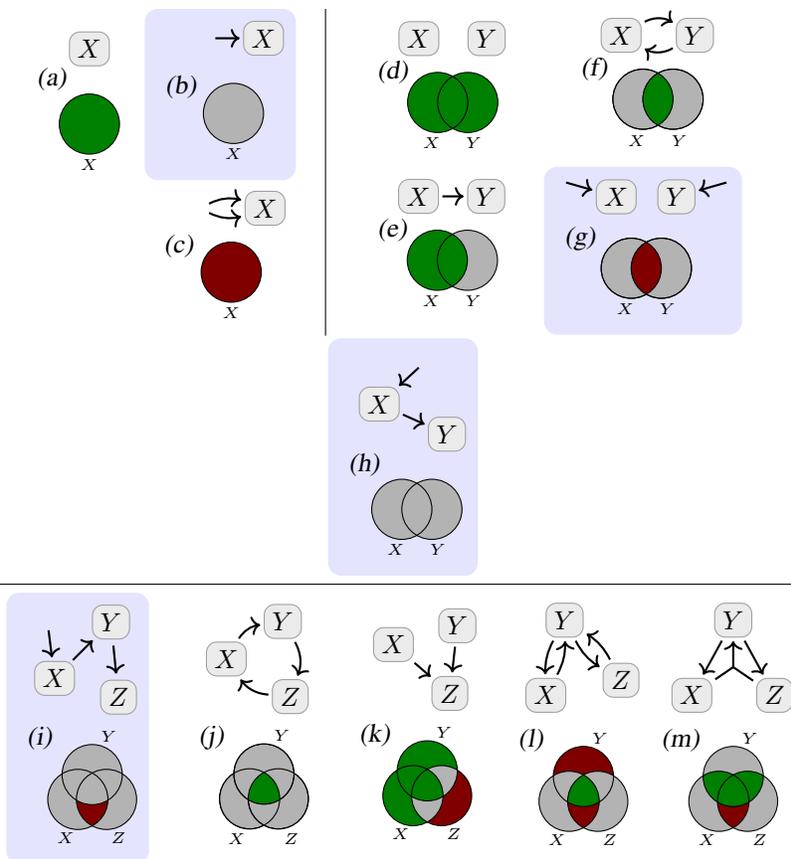
\begin{figure}
		\centering
	\def\vsize{0.4}
	\def\spacerlength{0.5em}
	\stepcounter{figure}
		\refstepcounter{subfigure}
		\begin{tikzpicture}\label{subfig:justX-0}
			\node[dpad0] (X) at (0,1){$X$};
			\draw[fill=green!50!black]  (0,0) circle (\vsize)  ++(-90:.22) node[label=below:\tiny$X$]{};
			\node at (-0.5, 0.6){\slshape\color{subfiglabelcolor}(\thesubfigure)};
		\end{tikzpicture}\!
	\begin{tabular}{c}
		\refstepcounter{subfigure}\label{subfig:justX-1}
		\begin{tikzpicture}[is bn]
			\node[] (1) at (-0.4, 1){};
			\node[dpad0] (X) at (0.4, 1){$X$};
			\draw[arr1] (1)  -- (X);
			\draw[fill=white!70!black]  (0,0) circle (\vsize) ++(-90:.22) node[label=below:\tiny$X$]{};
			\node at (-0.6,0.35){};
			\node at (-0.7, 0.35){\slshape\color{subfiglabelcolor}(\thesubfigure)};
		\end{tikzpicture}
		\\
		\refstepcounter{subfigure}\label{subfig:justX-2}
		\begin{tikzpicture}
			\node[] (1) at  (-0.45,.85){};
			\node[dpad0] (X) at  (0.45,.85){$X$};
			\draw[arr1] (1) to[bend left=20] (X);
			\draw[arr1] (1) to[bend right=20] (X);
			\draw[fill=red!50!black] (0,0) circle (\vsize) ++(-90:.22) node[label=below:\tiny$X$]{};
			\node at (-0.7, 0.35){\slshape\color{subfiglabelcolor}(\thesubfigure)};
		\end{tikzpicture}
	\end{tabular}%
	\hspace{\spacerlength}\vrule\hspace{\spacerlength}
		\begin{tabular}{c}
		\refstepcounter{subfigure}\label{subfig:justXY}
		\begin{tikzpicture}[]
			\node[dpad0] (X) at (-0.45,.85){$X$};
			\node[dpad0] (Y) at (0.45,.85){$Y$};
			\path[fill=green!50!black] (-0.2,0) circle (\vsize) ++(-110:.23) node[label=below:\tiny$X$]{};
			\path[fill=green!50!black] (0.2,0) circle (\vsize) ++(-70:.23) node[label=below:\tiny$Y$]{};
			\begin{scope}
				\clip (-0.2,0) circle (\vsize);
				\clip (0.2,0) circle (\vsize);
				\fill[green!50!black] (-1,-1) rectangle (3,3);
			\end{scope}
			\draw (-0.2,0) circle (\vsize);
			\draw (0.2,0) circle (\vsize);
			\node at (-0.8, 0.4){\slshape\color{subfiglabelcolor}(\thesubfigure)};
		\end{tikzpicture}\\[0.5em]
		\refstepcounter{subfigure}\label{subfig:XtoY}
		\begin{tikzpicture}[]
			\node[dpad0] (X) at (-0.45,0.85){$X$};
			\node[dpad0] (Y) at (0.45,0.85){$Y$};
			\draw[arr1] (X) to[] (Y);
			\path[fill=green!50!black] (-0.2,0) circle (\vsize) ++(-110:.23) node[label=below:\tiny$X$]{};
			\path[fill=white!70!black] (0.2,0) circle (\vsize) ++(-70:.23) node[label=below:\tiny$Y$]{};
			\begin{scope}
				\clip (-0.2,0) circle (\vsize);
				\clip (0.2,0) circle (\vsize);
				\fill[green!50!black] (-1,-1) rectangle (3,3);
			\end{scope}
			\draw (-0.2,0) circle (\vsize);
			\draw (0.2,0) circle (\vsize);
			\node at (-0.8, 0.4){\slshape\color{subfiglabelcolor}(\thesubfigure)};
		\end{tikzpicture}
	\end{tabular}%
	\begin{tabular}{c}
		\refstepcounter{subfigure}\label{subfig:XY-cycle}
		\begin{tikzpicture}[center base]
			\node[dpad0] (X) at (-0.49,0.85){$X$};
			\node[dpad0] (Y) at (0.49,0.85){$Y$};
			\draw[arr1] (X) to[bend left] (Y);
			\draw[arr1] (Y) to[bend left] (X);
			\draw[fill=white!70!black] (-0.2,0) circle (\vsize) ++(-110:.25) node[label=below:\tiny$X$]{};
			\draw[fill=white!70!black] (0.2,0) circle (\vsize) ++(-70:.25) node[label=below:\tiny$Y$]{};
			\begin{scope}
				\clip (-0.2,0) circle (\vsize);
				\clip (0.2,0) circle (\vsize);
				\fill[green!50!black] (-1,-1) rectangle (3,3);
			\end{scope}
			\draw (-0.2,0) circle (\vsize);
			\draw (0.2,0) circle (\vsize);
			\node at (-0.85, 0.4){\slshape\color{subfiglabelcolor}(\thesubfigure)};
		\end{tikzpicture}\\[2.5em]
	\refstepcounter{subfigure}\label{subfig:XYindep}
		\begin{tikzpicture}[center base, is bn]
			\node[dpad0] (X) at (-0.4,0.95){$X$};
			\node[dpad0] (Y) at (0.4,0.95){$Y$};
            \draw[arr1,<-] (X) -- +(-0.7,0.2);
            \draw[arr1,<-] (Y) -- +(0.7,0.2);
			\draw[fill=white!70!black] (-0.2,0) circle (\vsize) ++(-110:.23) node[label=below:\tiny$X$]{};
			\draw[fill=white!70!black] (0.2,0) circle (\vsize) ++(-70:.23) node[label=below:\tiny$Y$]{};
			\begin{scope}
				\clip (-0.2,0) circle (\vsize);
				\clip (0.2,0) circle (\vsize);
				\fill[red!50!black] (-1,-1) rectangle (3,3);
			\end{scope}
			\draw (-0.2,0) circle (\vsize);
			\draw (0.2,0) circle (\vsize);
			\node at (-0.88, 0.4){\slshape\color{subfiglabelcolor}(\thesubfigure)};
		\end{tikzpicture}
	\end{tabular}
	\hspace{\spacerlength}

		 \refstepcounter{subfigure}\label{subfig:1XY}
		\begin{tikzpicture}[center base, is bn]
			\node[] (1) at (0.15,2){};
			\node[dpad0] (X) at (-0.5,1.4){$X$};
			\node[dpad0] (Y) at (0.4,1){$Y$};
			\draw[arr0] (1) to[] (X);
			\draw[arr1] (X) to[] (Y);
			\path[fill=white!70!black] (-0.2,0) circle (\vsize) ++(-110:.23) node[label=below:\tiny$X$]{};
			\path[fill=white!70!black] (0.2,0) circle (\vsize) ++(-70:.23) node[label=below:\tiny$Y$]{};
			\begin{scope}
				\clip (-0.2,0) circle (\vsize);
				\clip (0.2,0) circle (\vsize);
			\end{scope}
			\draw (-0.2,0) circle (\vsize);
			\draw (0.2,0) circle (\vsize);
			\node at (-0.7, 0.6){\slshape\color{subfiglabelcolor}(\thesubfigure)};
		\end{tikzpicture}
    \\\smallskip\hrule\smallskip
		 \refstepcounter{subfigure}\label{subfig:1XYZ}
		\begin{tikzpicture}[center base,is bn]
			\node[] (1) at (-0.6,2.3){};
			\node[dpad0] (X) at (-0.5,1.5){$X$};
			\node[dpad0] (Y) at (0.25,2.25){$Y$};
			\node[dpad0] (Z) at (0.35,1.25){$Z$};
			\draw[arr1] (1) to (X);
			\draw[arr0] (X) to[] (Y);
			\draw[arr2] (Y) to[] (Z);
			\path[fill=white!70!black] (210:0.22) circle (\vsize) ++(-130:.25) node[label=below:\tiny$X$]{};
			\path[fill=white!70!black] (-30:0.22) circle (\vsize) ++(-50:.25) node[label=below:\tiny$Z$]{};
			\path[fill=white!70!black] (90:0.22) circle (\vsize) ++(40:.29) node[label=above:\tiny$Y$]{};
			\begin{scope}
                \clip (-30:0.22) circle (\vsize);
				\clip (210:0.22) circle (\vsize);
				\fill[red!50!black] (-1,-1) rectangle (3,3);
                \clip (90:0.22) circle (\vsize);
				\fill[white!70!black] (-1,-1) rectangle (3,3);
			\end{scope}
			\begin{scope}
				\draw[] (-30:0.22) circle (\vsize);
				\draw[] (210:0.22) circle (\vsize);
				\draw[] (90:0.22) circle (\vsize);
			\end{scope}
			\node at (-0.7, 0.7){\slshape\color{subfiglabelcolor}(\thesubfigure)};
		\end{tikzpicture}
		\hspace{3pt}
	\hspace{\spacerlength}%
		\refstepcounter{subfigure}\label{subfig:XYZ-cycle}
		\begin{tikzpicture}[center base]
			\node[dpad0] (X) at (-0.5,1.75){$X$};
			\node[dpad0] (Y) at (0.25,2.25){$Y$};
			\node[dpad0] (Z) at (0.35,1.25){$Z$};
			\draw[arr1] (X) to[bend left=25] (Y);
			\draw[arr1] (Y) to[bend left=25] (Z);
			\draw[arr1] (Z) to[bend left=25] (X);

			\draw[fill=white!70!black] (210:0.22) circle (\vsize) ++(-130:.27) node[label=below:\tiny$X$]{};
			\draw[fill=white!70!black] (-30:0.22) circle (\vsize) ++(-50:.27) node[label=below:\tiny$Z$]{};
			\draw[fill=white!70!black] (90:0.22) circle (\vsize) ++(40:.31) node[label=above:\tiny$Y$]{};

			\begin{scope}
				\clip (-30:0.22) circle (\vsize);
				\clip (210:0.22) circle (\vsize);
				\clip (90:0.22) circle (\vsize);
				\fill[green!50!black] (-1,-1) rectangle (3,3);
			\end{scope}
			\begin{scope}
				\draw[] (-30:0.22) circle (\vsize);
				\draw[] (210:0.22) circle (\vsize);
				\draw[] (90:0.22) circle (\vsize);
			\end{scope}
			\node at (-0.7, 0.7){\slshape\color{subfiglabelcolor}(\thesubfigure)};
		\end{tikzpicture}
	\hspace{3pt}
	\hspace{\spacerlength}%
		\refstepcounter{subfigure}\label{subfig:XZtoY}
		\begin{tikzpicture}[center base]
			\node[dpad0] (X) at (-0.45,1.9){$X$};
			\node[dpad0] (Z) at (0.3,1.25){$Z$};
			\node[dpad0] (Y) at (0.4,2.15){$Y$};
			\draw[arr0] (X) to[] (Z);
			\draw[arr1] (Y) to[] (Z);
			\path[fill=green!50!black] (210:0.22) circle (\vsize) ++(-130:.25) node[label=below:\tiny$X$]{};
			\path[fill=red!50!black] (-30:0.22) circle (\vsize) ++(-50:.25) node[label=below:\tiny$Z$]{};
			\path[fill=green!50!black] (90:0.22) circle (\vsize) ++(40:.29) node[label=above:\tiny$Y$]{};
			\begin{scope}
				\clip (-30:0.22) circle (\vsize);
				\clip (90:0.22) circle (\vsize);
				\fill[white!70!black] (-1,-1) rectangle (3,3);
			\end{scope}
			\begin{scope}
				\clip (-30:0.22) circle (\vsize);
				\clip (210:0.22) circle (\vsize);
				\fill[white!70!black] (-1,-1) rectangle (3,3);

				\clip (90:0.22) circle (\vsize);
				\fill[green!50!black] (-1,-1) rectangle (3,3);
			\end{scope}
			\draw[] (-30:0.22) circle (\vsize);
			\draw[] (210:0.22) circle (\vsize);
			\draw[] (90:0.22) circle (\vsize);
			\node at (-0.7, 0.7){\slshape\color{subfiglabelcolor}(\thesubfigure)};
		\end{tikzpicture}~
		\hspace{\spacerlength}%
			\refstepcounter{subfigure}\label{subfig:XYZ-bichain}
			\begin{tikzpicture}[center base]
				\node[dpad0] (X) at (-0.5,1.3){$X$};
				\node[dpad0] (Y) at (-0.25,2.3){$Y$};
				\node[dpad0] (Z) at (0.5,1.5){$Z$};
				\draw[arr1] (X) to[bend right=15] (Y);
				\draw[arr1] (Y) to[bend right=15] (X);
				\draw[arr1] (Y) to[bend right=15] (Z);
				\draw[arr1] (Z) to[bend right=15] (Y);
				\path[fill=white!70!black] (210:0.22) circle (\vsize) ++(-130:.25) node[label=below:\tiny$X$]{};
				\path[fill=white!70!black] (-30:0.22) circle (\vsize) ++(-50:.25) node[label=below:\tiny$Z$]{};
				\path[fill=red!50!black] (90:0.22) circle (\vsize) ++(40:.29) node[label=above:\tiny$Y$]{};
				\begin{scope}
					\clip (-30:0.22) circle (\vsize);
					\clip (90:0.22) circle (\vsize);
					\fill[white!70!black] (-1,-1) rectangle (3,3);
				\end{scope}
				\begin{scope}
					\clip (-30:0.22) circle (\vsize);
					\clip (210:0.22) circle (\vsize);
					\fill[red!50!black] (-1,-1) rectangle (3,3);
				\end{scope}
				\begin{scope}
					\clip (90:0.22) circle (\vsize);
					\clip (210:0.22) circle (\vsize);
					\fill[white!70!black] (-1,-1) rectangle (3,3);

					\clip (-30:0.22) circle (\vsize);
					\fill[green!50!black] (-1,-1) rectangle (3,3);
				\end{scope}
				\draw[] (-30:0.22) circle (\vsize);
				\draw[] (210:0.22) circle (\vsize);
				\draw[] (90:0.22) circle (\vsize);
				\node at (-0.7, 0.7){\slshape\color{subfiglabelcolor}(\thesubfigure)};
			\end{tikzpicture}
        \refstepcounter{subfigure}\label{subfig:X-Y-Z-undir}
        \begin{tikzpicture}[center base]
            \node[dpad0] (X) at (-0.55,1.3){$X$};
            \node[dpad0] (Y) at (0,2.3){$Y$};
            \node[dpad0] (Z) at (0.55,1.3){$Z$};
            \draw[arr1] (Y) to (X);
            \draw[arr1] (Y) to (Z);
            \mergearr[arr1] XZY
            \path[fill=white!70!black] (210:0.22) circle (\vsize) ++(-130:.25) node[label=below:\tiny$X$]{};
            \path[fill=white!70!black] (-30:0.22) circle (\vsize) ++(-50:.25) node[label=below:\tiny$Z$]{};
            \path[fill=white!70!black] (90:0.22) circle (\vsize) ++(40:.29) node[label=above:\tiny$Y$]{};
            \begin{scope}
                \clip (-30:0.22) circle (\vsize);
                \clip (90:0.22) circle (\vsize);
                \fill[white!70!black] (-1,-1) rectangle (3,3);
            \end{scope}
            \begin{scope}
                \clip (-30:0.22) circle (\vsize);
                \clip (210:0.22) circle (\vsize);
                \fill[red!50!black] (-1,-1) rectangle (3,3);
            \end{scope}
            \begin{scope}
                \clip (90:0.22) circle (\vsize);
                \clip (210:0.22) circle (\vsize);
                \fill[green!50!black] (-1,-1) rectangle (3,3);
            \end{scope}
            \begin{scope}
                \clip (-30:0.22) circle (\vsize);
                \clip (90:0.22) circle (\vsize);
                \fill[green!50!black] (-1,-1) rectangle (3,3);
            \end{scope}
            \draw[] (-30:0.22) circle (\vsize);
            \draw[] (210:0.22) circle (\vsize);
            \draw[] (90:0.22) circle (\vsize);
            \node at (-0.7, 0.7){\slshape\color{subfiglabelcolor}(\thesubfigure)};
        \end{tikzpicture}
	\addtocounter{figure}{-1} %
	\caption{
        Illustrations of the structural deficiency $\IDef_{\Ar}$ underneath
          drawn underneath their
		  associated hypergraphs $\{ G_i\}$. Each circle represents a
		  variable; an area in the intersection of circles $\{C_j\}$
		  but outside of circles $\{D_k\}$ corresponds to information
		  that is shared between all $C_j$'s, but not in any
		  $D_k$. Variation of a candidate distribution $\mu$ in a
		  green area makes its qualitative fit better (according to
		  $\IDef{}$), while variation in a red area makes its
		  qualitative fit worse; grey is neutral.
		  Only the boxed structures in blue,
		  whose $\IDef$ can be seen as measuring distance to a particular set of (conditional) independencies, are expressible as BNs.}
	\label{fig:info-diagram}
	\end{figure}
\begin{itemize}
	\item
Subfigures \ref{subfig:justX-0}, \ref{subfig:justX-1}, 
    and \ref{subfig:justX-2} show
    how adding \arc s makes distriutions more deterministic.
When $\Ar$ is the empty hypergraph, $\IDef$ reduces to negative entropy,
and so prefers distributions that are ``maximally uncertain''
 	(e.g., Subfigures \ref{subfig:justX-0} and \ref{subfig:justXY}).
For this empty
    but all distributions $\mu$ have negative $\IDef_{\Ar}(\mu) \le 0$.
In the definition of $\IDef$, each hyperarc $X\to Y$ is compiled to a ``cost'' $H(Y|X)$ 
    for uncertainty in $Y$ given $X$.
One can see this visually in \cref{fig:info-diagram}
    as a red crescent that's added to the information profile as we move 
    from \ref{subfig:justXY} to \ref{subfig:XtoY} to \ref{subfig:XY-cycle}.

\item
Some hypergraphs (see \Cref{subfig:justX-1,subfig:1XY}) are \emph{indiscriminate}, in the sense that every distribution gets the same score
(of zero, because a point mass $\delta$ always has $\SDef_{\Ar}(\delta) = 0$).
Such a graph has a structure such that \emph{any} distribution can be precisely encoded by the process in (b).
As shown here and also in \citet{pdg-aaai}, $\IDef$ can also indicate independencies and conditional independencies, illustrated respectively in Subfigures \ref{subfig:XYindep} and \ref{subfig:1XYZ}.

\item
For more complex structures, structural information deficiency $\IDef$ 
can represent more than independence and dependence. 
The cyclic structures in \cref{example:xy-cycle,example:xyz-cycle-1}, 
    correspond to the structural deficiencies pictured
    in Subfigures \ref{subfig:XY-cycle} and \ref{subfig:XYZ-cycle}, respectively,
    which are functions that encourage shared information between the three variables. 

\end{itemize}

\vfull{%
\subsection{Weights for SIM-Inc}
Given $\Ar$ and $|\Ar|+2$ positive weights $\boldsymbol\lambda = (\lambda^{\text{(a)}}, \lambda^{\text{(b)}}, \{\lambda^{\text{(c)}}_a\}_{a \in \Ar})$, define the function

\begin{align*}
    \SIMInc_{\Ar, \boldsymbol\lambda}(\mu) :=
    & \inf_{\nu(\U,\X)}
    \left\{
    \begin{array}{rl}
    &\lambda^{\text{(a)}}
     \kldiv[\big]{ \nu(\X) }{\mu(\X)}\\[1ex]
    + & \displaystyle
        \lambda^{\text{(b)}} \Big( - \H_\nu (\mathcal U) + 
    \sum_{a \in \Ar} \H_\nu (U_a)
         \Big) \\
    + & \displaystyle
        \sum_{\smash{a \in \Ar}}
        \lambda_a^{\text{(c)}} \H_\nu(\Tgt a | \Src a , U_a)
    \end{array}\right.
    .  \numberthis\label{eq:siminc-weighted}
\end{align*}

When $\boldsymbol\lambda = (\infty, 1, \mat 1)$, we get the analogous 
quantity defined in \eqref{eq:siminc} in the main text.

Here are some analogous results for this generalized version with weights. For a weighted hypergraph $(\Ar, \balpha)$, here is a strengthening of \cref{theorem:siminc-idef-bounds}, and the appropriate translateion of the hypergraph.
Given $(\Ar, \balpha)$ translate it to a new derandomized hypergraph 
$(\Ar, \balpha)^\dagger$ by replacing each weighted hyperarc 
\[
\begin{tikzpicture}[center base]
    \node[dpad0] (S) at (0,0) {$\Src a$};
    \node[dpad0] (T) at (1.6,0) {$\Tgt a$};
    \draw[arr2] (S) to node[above,pos=0.4]{$a$}
        node[below,pos=0.4,inner sep=1pt]{\color{gray}$\scriptstyle(\alpha_a)$}
        (T);
\end{tikzpicture}
\quad 
\text{with the pair of weighted \hyperarc s}
\quad
\begin{tikzpicture}
        [center base]
    \node[dpad0] (S) at (0,0) {$\Src a$};
    \node[dpad0] (T) at (1.4,0) {$\Tgt a$};
    \node[dpad0] (U) at (0.6,0.7) {$U_a$};
    \draw[arr,<-] (U) to node[above,pos=0.6]
        {$a_0$}
        node[below,pos=0.6,inner sep=1pt]{\color{gray}$\scriptstyle(\alpha_a)$}
         +(-1.2,0);
    \mergearr[arr1] SUT
    \node[below=2pt of center-SUT,xshift=-0.2em] (a1)
        {$a_1$};
    \node[below=0pt of a1, inner sep=1pt] {\color{gray}$\scriptstyle(\alpha_a)$};
\end{tikzpicture}~.
\]

\TODO
}%

\subsection{Counter-Examples to the Converse of Theorem \ref{theorem:sdef-le0}}
    \label{appendix:converse-sdef-le0}

In light of 
\cref{example:ditrichotomy}
and its connections to 
$\IDef$
through \cref{theorem:sdef-le0},
one might hope this criterion is not just a bound, but
    a precise characterization of the distributions that are \scible\ 
    with the 3-cycle. 
Unfortunately, it does not, and the converse of \cref{theorem:sdef-le0} is false.

\begin{example}
    Suppose $\mu(X,Y,Z) = \mathrm{Unif}(X,Z) \delta \mathrm{id} (Y|X)$ and $\Ar = \{\to X, \to Y\}$,
    where all variables are binary.
    Then $\IDef_{\!\Ar}(\mu) \,{=}\, 0$, but $X$ and $Y$ are not independent.
    \qedhere
\end{example}
Here is another counter-example, of a very different kind.
\begin{example}
    Suppose  $A, B, C$ are binary variables.
    It can be shown by enumeration (see appendix) that
    no distribution supported on seven of the eight
    possible joint settings of of $\V(A,B,C)$ can be
    \scible\ with the 3-cycle $\Ar_{3\circ}$. Yet it is easy
    to find examples of such distributions $\mu$ that have positive
    interaction information $\I(A;B;C)$,
    and thus $\IDef_{\mu}(\Ar_{3\circ}) \le 0$ for such distributions.
\end{example}

\section{\SCibility\ Constructions and Counterexamples}
    \label{sec:func-counterexamples}

We now give a counterexample to a 
simpler previously conjectured strengthening of \cref{theorem:func},
in which part (a) is an if-and-only-if. 
This may be surprising.
In the unconditional case, it is true that, two arcs $\{ \ed1{}X, \ed2{}X \}$ precisely encode that $X$ is a constant, as illustrated by \cref{example:two-edge-det}.
The following, slightly more general result, 
    is an immediate correlary of \cref{theorem:func}(c).

\begin{prop}
    $\mu \models \Diamond \Ar \sqcup \{ \ed1{}X, \ed2{}X \}$ if and only if $\mu \models \Diamond \Ar$ and
    $\mu \models \emptyset \tto X$. 
\end{prop}

One can be forgiven for imagining that the conditional case would be analogous---that \scibility\ with a \hgraph\ that has two parallel arcs from $X$ to $Y$ would imply that $Y$ is a function of $X$. But this is not the case.
Furthermore, our counterexample also shows that 
neither of the two properties we consider in the main text
    (requiring that $\Ar$ is \partl, or that the \scibility\ with $\mu$ is even) 
    are enough to ensure this. 
That is, there are \partl\ graphs $\Ar$ such that $\mu \emodels \Ar$ but $\mu \not\models \Diamond \Ar \sqcup \{ \ed1XY, \ed2XY \}$.

\begin{example}
        \label{ex:counterexample-func-simple}
    We will construct a witness of SIM-compatibility
    for the \hgraph\ 
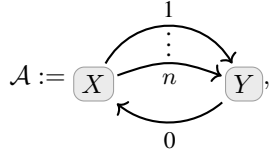
    \[
        \Ar := 
        \begin{tikzpicture}[center base]
            \node[dpad0] (X) {$X$};
            \node[dpad0] (Y) at (2,0) {$Y$};
            \draw[arr1] (X) to[bend left=60,looseness=1.2] node[below]{$\vphantom{|}\smash{\scriptstyle\vdots}$} 
                node[above]{\small 1} (Y);
            \draw[arr1] (X) to[bend left=20,looseness=1.2] node[below]{\small $n$} (Y);
            \draw[arr1] (Y) to [bend left=40] node[below]{\small 0}(X);
        \end{tikzpicture},
    \]  
    in which $Y$ is \emph{not} a function of $X$, 
    which for $n=3$ will disprove the analogue of \cref{theorem:func} for the \partl\ context $\Ar'$ equal to the 2-cycle. 
    
    Let  $\U = (U_0, U_1, \ldots, U_n)$ be a vector of $n$ mutually independent random coins, and $A$ is one more independent random coin. 
    For notational convenience, define the random vector $\mat U := (U_0, \ldots, U_n)$ consisting of all variables $U_i$ except for $U_0$. 
    Then, define variables $X$ and $Y$ according to:
    \begin{align*}
        X &:= 
            ( A \oplus U_1 , \ldots,  A \oplus U_n ,  ~~ U_0 \oplus U_1, U_0 \oplus U_2, \ldots,  U_0 \oplus U_n ) 
            \\
            &= (A \oplus \mat U,~~ U_0 \oplus \mat U)
            \\
        Y &:= (A, U_0 \oplus \mat U ) = (A, ~U_0 \oplus U_1, U_0 \oplus U_2, \ldots,  U_0 \oplus U_n ) 
            ,
    \end{align*}    
    where and the operation $Z \oplus \mat V$ is element-wise xor (or addition in $\mathbb F_2^n$), after implicitly converting the scalar $Z$ to a vector by taking $n$ copies of it. Call the resulting distribution 
    $\nu(X, Y, \U)$. 
    
    It we now show that $\nu$ witnesses that its marginal on $X,Y$ is 
        \scible\ with $\Ar$, which is straightforward.
    \begin{enumerate}[label=(\alph*), start=2]
        \item $\U$ are mutually independent by assumption;
        \item[(c.0)] $Y = (A, \mat B)$ and $U_0$ determine $X$ according to:
        \begin{align*}
            g(A, \mat B, U_0) &= (A \oplus U_0 \oplus \mat B,\; \mat B) \\
                &= (A \oplus U_0 \oplus U_0 \oplus \mat U,\;  U_0 \oplus \mat U) & \text{since } \mat B =U_0 \oplus \mat U \\
                &= (A \oplus \mat U, U_0 \oplus \mat U) = X
        \end{align*}
        \item[(c.1--$n$)] for $i \in \{1, \ldots, n\}$, $U_i$ and $X = (\mat V, \mat B)$ together determine $Y$ according to
        \begin{align*}
            f_i(\mat V, \mat B, U_i) := (V_i \oplus U_i,\; \mat B) 
                = (A \oplus U_i \oplus U_i,\; U_0 \oplus \mat U) = Y.
        \end{align*}
    \end{enumerate}
    In addition, this distribution $\nu(\U, X,Y)$ satisfies condition
    \begin{enumerate}[label=(\alph*), start=4]
    \item $\nu(X,Y \mid \U) = \frac{1}{2} \mathbbm1[ g(Y, U_0) = X] \prod_{i = 1}^n \mathbbm1[ f_i(X, U_i) = Y]$,
    since, for all joint settings of $\U$, there are two possible values of $(X,Y)$, corresponding to the two values of $A$, and both happen with probability $\frac12$. 
    \end{enumerate}
    
    Thus, we have constructed a distribution that witnessing the fact that $\mu(X,Y) \emodels \Ar$.  
    
    Yet, observe that $X$ alone does not determine $Y$ in this distribution, because $X$ alone is not enough to determine $A$ (without also knowing some $U_i$). 
        
    For those who are interested, 
    observe that the bound of \cref{theorem:sdef-le0} tells that
    we must satisfy 
    \begin{align*}
    0 \ge
        \IDef_{\!\Ar}(\mu)
        &= - \H_\mu(X,Y) + n \H_\mu(Y \mid  X) + \H_\mu(X \mid Y) \\
        &= - \I_\mu(X;Y) + (n-1) \H_\mu(Y \mid X)  
    \end{align*}
    Indeed, this distribution has information profile 
    \[
        \H(X \mid Y) = 1\,\text{bit},\qquad
        \I(X ; Y) = n\,\text{bits},\qquad
        \H(Y \mid X) = 1\,\text{bit},
    \]
    and so $\IDef_{\!\Ar}(\mu) = -1\,\text{bit}$. 
    Intuitively, this one missing bit corresponds to the value of $A$ that is not determined by the structure of $\Ar$.
\end{example}

\section{From Causal Models to Witnesses}
    \label{appendix:sem2witness}
    
We now return to the ``easy'' direction of the correspondence between \scibility\ witnesses and causal models, mentioned at the beginning of \cref{sec:witness-to-causal-model}.
Given a (generalized) randomized PSEM $\cal M$, we now show that distributions $\nu \in \SD{\cal M}$, are \scibility\ witness showing that the marginals of $\nu$ are \scible\ with the hypergraph $\Ar_{\cal M}$. 
More formally:

\commentout{
\begin{prop}{causal-model-structure-capture}
    If $M$ is a SEM, then
    $\mu(\enV,\U) \models \Diamond \Ar_{M}$ 
        iff
        there exists $M'$ such that $\Pa_M = \Pa_{M'}$ and
        $\mu \in \SD{M'}$.
\end{linked}
}

\begin{prop}
    If $\mathcal M = (M
        {=} (\U,\enV,\mathcal F)
        , P)$ is a randomized PSEM, then
    every $\nu \in \SD{\mathcal M}$ 
    witnesses
    the \scibility\ of its marginal on its exogenous variables, with the dependency structure of $\mathcal M$.
    That is,
    for all $\nu \in \SD{\cal M}$ 
    and $\mathcal Y \subseteq \U \cup \V$, 
    $\nu(\mathcal Y) \models \Diamond \Ar_{\cal M}$. 
\end{prop}

The proof is straightforward: by definition, if $\nu \in \SD{\cal M}$, 
then it must satisfy the equations, and so automatically fulfills condition (c). 
Condition (a) is also satisfied trivially, by assumption: the distribution we're considering is defined to be a marginal of $\nu$.
Finally, (b) is also satisfied by construction: we assumed that $\U_{\Ar} = \{ U_a \}_{a \in \Ar}$ are independent.

\commentout{
\section{Combining Qualitative Dependency Graphs}
    \label{sec:union}
    
In a BN or an MRF, the choice of whether
or not to place an edge between $A$ and $B$
does not just depend on the relationship between $A$ and $B$---it also
    depends on what other variables are present in the graph.
For example, suppose $X$ is a variable that mediates
    the interaction between $A$ and $B$.

In a ``coarse'' model that does not contain $X$, it would be appropriate
to include an edge $A{\to}B$, but in a ``fine'' model that does, that edge
would be in appropriate, and instead one should include edges $A{\to}X{\to}B$.
Although the situation is similar for our case, the reasons
for this modeling choice are more transparent: if the model already contains $A{\to}X{\to}B$,
one should not also include the \arc\ $A{\to}B$ because it is not independent of the two mechanisms already present.    

[...] It follows that notion of \scibility\ cannot be
defined inductively over unions of independent mechanisms. 
}

\vfull{%

\section{An Algorithm for Finding Witnesses: The Null Value Construction}
    \label{sec:null}

We have now built up a body of examples, but it is still not clear
    how to compute \scibility.
In other words, it is still not clear how to solve the decision problem: given $\mu$ and $\Ar$, determine whether or not $\mu \models \Diamond \Ar$.
In this section, we discuss one approach to this problem. 

If you start with a distribution $\nu(\X)$,
it's not at all obvious how to extend it with a conditional distribution
$\nu(\U|\X)$ such that the variables $\U$ are \emph{unconditionally}
dependent, given that they cannot be independent of $\X$.
It seems that the only way to ensure this unconditional independence
is to start with a distribution $\nu(\U) = \prod_{a} \nu(U_a)$
and then figure out how to extend it to the variables $\X$.

To begin, for each $a \in \Ar$, take $U_a$ to be a response variable,
taking values $\V(U_a) = (\V \Tgt a) ^{(\V \Src a)}$, just 
as in \cref{sec:causal}.
But now we run into a problem: without carefully selecting the
supports of the distributions over $\U$, it is entirely possible
that there will be some joint setting $\mat u \in \V \U$
occurs with positive problability, but represents a collection
    of functions that has no fixed point.  
For example, take the graph

\begin{example}[5, continued]
    By randomly selecting distributions $\Pr(U_1),
    \Pr(U_2)$, and $\Pr(U_3)$ (see \cref{sec:null}), one finds that the set of distributions that are consistent with this 3-cycle has larger dimension than the set of distributions that factorize according to $\Pr(X,Y,Z) \propto \phi_1(X,Y) \phi_2(Y,Z) \phi_3(Z,X)$.%
        \footnote{see appendix for details.}
    Thus, our definition does not coincide with the corresponding factor graph.
\end{example}

\begin{conj}
    If $\mu_0 \models \Diamond \Ar$, and $\mu'$ lies on the path $\mu(t)$ of
    gradient flow minimizing $\IDef_{\!\Ar}(\mu')$, starting at $\mu(0) = \mu_0$, then $\mu' \models \Diamond \Ar$.
\end{conj}

The following has emperical support.

\begin{conj}
    The distribution (s?)
    $\hat \mu := \arg\min_{\mu : \mu \models \Diamond \Ar} \kldiv{\mu}{\hat\mu}$
    have the same information profile as $\mu$. 
\end{conj}

\section{Even Structural Compatibility}

\subsection{Even \SCibility}

If $\mathcal M$ is a cyclic or \subpartl\ PSEM, then $\SD{\mathcal M}$ may 
contain many distributions. Still, so long as it is non-empty, there is
still a unique distribution that, arguably, best describes the distribution of the PSEM (in the absence of interventions)---namely, the one that, for any given value $\mat u \in \V(\U)$, treats all ``fixed-points''
\ifvfull
    $\mat x \in \mathcal F(\mat u)$
\else
    $\{\mat x \in \V(\X) : \forall a \in \Ar. f_a(\Src a(\mat x), u_a) = \Src a(\mat x) \}$
\fi
of the equations $\mathcal F$ symmetrically.
For example, if $\mathcal M$ has 
no exogenous variables ($\U = \emptyset$),
endogenous variables $\X = [X_1, \ldots, X_n]$ that are all binary, 
and equations 
$f_{X_i}(\X\setminus X_i) = X_{(i+1)\% n}$,
then $\SD{\mathcal M}$ is a 1-dimensional specturm of distributions
    supported on the two points $\{ (0, \ldots, 0), (1, \ldots, 1) \}$. 
The distribution that gives the two an equal weight of $\frac12$ is somehow special, in that it is the unique one that does not break the symmetry by preferring either $(0, \ldots, 0)$ or $(1, \ldots, 1)$. 
This intuition is made precise, and generalized to \scibility\ witnesses (rather than PSEMs), by the following definition.

\begin{defn}
        \label{defn:esim-compat}
    We say a witness $\nu(\U,\X)$ to $\mu \models \Diamond \Ar$ is \emph{even}, iff
    it satisfies properties (a,b) of \cref{defn:scompat}, and also
    the following strengthening of property (c):
    \begin{enumerate}[start=4,label={(\alph*)}, nosep]
    \item $\displaystyle
        \nu(\X \mid \U) \propto \mathbbm{1}\Big[ \bigwedge_{a \in \Ar}
            f_a(\Src a, U_a) = \Tgt a \Big],
    $
    \end{enumerate}
    for some set
    $
    \mathcal F = 
    \{ f_a : \V(\Src a, U_a) \to \V(\Tgt a) \}_{a \in \Ar}
    $
     of equations.
     In this case, we say $\mu$ is \emph{evenly} \scible\ (\escible) with $\Ar$,
     write $\mu \emodels \Ar$, and 
    call the pair $(\nu, \mathcal F)$ a witness of E\scibility.
\end{defn}

E\scibility\ 
clearly implies \scibility, 
and thus is a stricter notion. 
Furthermore, E\scibility\ witnesses have an even sharper relationship to causal models:
A witness $(\bar\mu(\U, \X), \mathcal F)$ to E\scibility, 
can be equivalently viewed as a PSEM $\mathcal M = (\U,\X, \mathcal F, \nu(\U))$, 
because the rest of the distribution $\nu(\X \mid \U)$ is determined by $\mathcal F$ and property (d). 

\begin{prop}
    \begin{itemize}
        \item There is a 1-1 correspondence between E\scibility\ witnesses and GRPSEMs $\cal M$ in which $\SD{\cal M} \ne \emptyset$. 
    \end{itemize}
\end{prop}
\begin{proof}
    1. Given a E\scibility\ witness $(\nu(\X, \U_\Ar), \mathcal F)$, 
    
    by \cref{prop:gen-sim-compat-means-arise}
    $\PSEMsA(\nu)$
\end{proof}

\commentout{
\begin{linked}{theorem}{esem-properties}
    \label{theorem:esem-properties}
    \begin{enumerate}[wide,label={(\alph*)}]
    \item

    \item
    Conversely, if $\mathcal M = (\X,\U,\mathcal F, P)$ is a (generalized) PSEM, 
    then there is a unique distribution $\nu(\U,\X) \in \SD{\mathcal M}$
    with marginal $\nu(\X) = P$,
    and that also satisfies property (d) of \cref{defn:esim-compat} for the equations $\mathcal F$.
    \item
    If $\nu$ is a witness to $\mu \emodels \Ar$, then 
    $\nu$ and its unique corresponding causal model $\mathcal M$
    ascribe the same probabilities to causal formulas. That is,
    $
        \forall \varphi \in \mathcal L(\X,\U). ~
        \Pr_{\nu}(\varphi) = P(\varphi).
    $
    \end{enumerate}
\end{linked}
}%
Thus, PSEMs are in direct 1-1 correspondence with E\scibility\ witnesses when \hgraph s $\Ar = \Ar_G$ for some graph $G$.

We now verify that various results from the main text extend to E\Scibility. 
\begin{itemize}
    \item[{[\cref{theorem:bns}]}]
    When $G$ is acyclic
    (and, more generally, when $\mathcal M \models \U \tto \X$), 
    the extra condition (d) holds trivially, and so
    E\scibility\ coincides with \scibility, and the analogue of \cref{theorem:bns} in which $\models$ is replaced with $\emodels$ also holds. 
    
    \item [{$\vdots$}]
        
    \item[{[\cref{prop:strong-mono}]}]
    As we show in the \hyperref[proof:strong-mono]{the proof of \cref*{prop:strong-mono}}, E\Scibility\ is also monotonic with respect to weakening ($\rightsquigarrow$). 
\end{itemize}

The original scoring function $\IDef$ is related to even \scibility. 

\begin{conj}
    There is a constant $\kappa = \kappa(\Ar, \V)$ depending on the \hgraph\ $\Ar$ and the possible values $\V$ that the variables can take, such that
    \[
        \IDef_{\!\Ar}(\mu) \le \kappa
        \quad\iff\quad
        \mu \emodels \Ar  
    \] 
\end{conj}
\begin{proof}
    \textbf{($\impliedby$).}
    Suppose $\mu \emodels \Ar$. Then there is some witness $\bar\mu$ extending $\mu$ to independent variables $\U_\Ar = \{ U_a \}_{a \in \Ar}$. 

    \textbf{($\implies$).}
    Suppose that $\IDef_{\!\Ar}(\mu) \le \kappa$. 
\end{proof}

\subsection{ESIM Compatibility Scoring Rules}
We have now seen that, $\IDef_{\!\Ar^\dagger}$ measures distance from 
    being a witness to \scibility\ (\cref{theorem:siminc-idef-bounds}(b)).
Modulo a constant offset or limit, $\IDef_{\!\Ar}$, i.e., the original scoring function applied to the original \hgraph\ 

Let's now repeat the same approach as the previous section, by explicitly constructing a scoring function for E\scibility.
Extend our previous weight vector by one entry, so that $\boldsymbol\lambda = (\lambda^{\text{(a)}},\lambda^{\text{(b)}},\{\lambda^{\text{(c)}}_a\}_{a \in \Ar}, \lambda^{\text{(d)}})$
\unskip.
\begin{align*}
    \mathrm{E}\mskip-1mu\SIMInc_{\Ar, \boldsymbol\lambda}(\mu) 
        &:= 
        \inf_{\nu(\X\!,\,\U)} 
            \blacksquare
        \\ + \lambda^{\text{(d)}}
        \inf_{\mathcal F}& \Ex_{\mat u \sim \nu(\U)}
            \Big[ 
            \kldiv[\big]{ \nu(\X \mid \U{=}\mat u) }{ \mathrm{Unif}[{\mathcal F(\mat u)}] }
            \Big],
\end{align*}
where $\blacksquare$ consists of everything but the infemum from \cref{eq:siminc},
$\mathcal F$ ranges over sets of equations along $\Ar$ (as in \cref{defn:SEM,defn:esim-compat}), and $\mathrm{Unif}[{\mathcal F(\mat u)}]$ is the uniform distribution over joint settings of $\X$ that are consistent with the equations after fixing context $\U = \mat u$. 
Recall that the key step of constructing $\Ar^\dagger$ was to add the hyperarc $\U \to \X$. But for even compatbility, we want to effectively do the opposite---that is, subject to satisfying the other constraints, we want to incentivize, rather than penalize the conditional entropy $\H(\X|\,\U)$. This is made precise by the following proposition.

\begin{prop}
    \begin{enumerate}[label={\normalfont(\alph*)},wide]
    \item For all $\boldsymbol{\lambda} > \mat 0$, $\mathrm{E}\mskip-1mu\SIMInc_{\Ar, \boldsymbol\lambda}(\mu) \ge 0$ with equality iff $\mu \emodels \Ar$.
    \commentout{
    \item 
    $\mathrm{E}\mskip-1mu\SIMInc_{\Ar}(\mu) = \IDef_{\!\Ar}(\mu) + \kappa_{\Ar}(\mu)$,
    where $\kappa_{\Ar}(\mu)$ 
    }
    \end{enumerate}
\end{prop}

In other words, $\IDef$ itself already measures distance from \emph{even \scibility}, once we find the appropriate constant to make it non-negative. 
Although has an enormous benefit of not requiring an infemum. 

\TODO[ TODO:  There's an issue here I still need to finish working out. ]
\[
    \IDef_{\!\Ar}(\mu)
    \le \SIMInc_{\Ar}(\mu) 
    \le \IDef_{\!\Ar^\dagger}(\nu^*)
    \le \IDef_{\!\Ar}(\mu) + \kappa()
    .
\]

$\kappa_\Ar$ is a (possibly infinite) piecewise constant function of $\mu$ with finitely many pieces,
and finite when $\mu \models \Diamond \Ar$. 

\subsection{Complete Derandomization for Cyclic Models}
    \label{ssec:full-derandomize}
We have seen that a number of properties of causal models are simpler when $\mathcal M \models \U \tto \X$. 
In some sense, the job of a causal model is model $\X$ by adding variables $\U$ that account for any uncertainty. 
When $\mathcal M \not\models \U \tto \X$, this job is in some sense incomplete; cycles can create a new source of uncertainty. 
In this subsection, we explore the effects of adding one more variable $U_0$ to account for the remaining uncertainty, by explaining it as randomenss.
Technically speaking, this means looking into one way of converting a 
    randomized PSEM $\cal M$ to one in which $\U \tto \X$.

Our construction is parameterized by a PSEM $\cal M$ and a choice of $\nu \in \SD{\mathcal M}$. In brief, we use a natural construction explained in the next subsection (\ref{sec:cpd-derandomize}) to obtain a ``maximally independent'' derandomization of $\nu(\X \mid \U)$.
The result is a new generalized randomized PSEM, which we call $\derind{\mathcal M}{\nu}$,
    differing from $\cal M$ in that it 
    has one extra \arc\ $\U \to \X$, and, correspondingly, an extra variable $U_0$, and an extra equation $f_0 : \U \to \X$.
This new causal model has two important properties: 
\begin{enumerate}
    \item 
    Only $\nu$ can arise from $\derind{\cal M}\nu$
    \hfill (i.e., $\SD{\derind{\cal M}\nu} = \{ \nu \}$), 
    and 
    \item 
    the exogenous variables determine the values endogenous ones
    \hfill(i.e.,
    $\derind{\cal M}\nu \models (\U,U_0) \tto \X$
    ). 
\end{enumerate}

\subsection*{Constructing the Causal Model $\derind{\cal M}{\nu}$}
We now apply the general technique above to obtain the causal model 
$\derind{\cal M}{\nu}$ discussed in \cref{ssec:full-derandomize}. 
Concretely, let 
$\V(U_0) := \prod_{\mat u \in \V(\U)} \mathcal F(\mat u)$
consist of all functions from $\V(\U)$ to $\V(\X)$ consistent with the equations $\mathcal F$,
and add an equation
$\X = U_0(\mat u) = f_{0}(U_0, \mat u)$
along the \arc\ $\U \to \X$.

Given $\nu \in \SD{\mathcal M}$,
define
$P(U_0 {=} f) := \prod_{\mat u \in \V(\U)} \nu(f(\mat u) \mid \mat u)$.
As shown more generally in \cref{sec:},

\begin{enumerate}
\item
 this indeed a probability distribution, and
 \item 
 the joint distribution $P(\U, U_0) = P(\U)P(U_0)$ 
 extends uniquely along the function defined by $U_0$, to a distribution $P(\U, U_0, \X)$, and the marginal of that distribution on ($\U, \X$) equals $\nu$.
\end{enumerate}
\commentout{%
Moreover, this distribution has marginal $P(\U, \X)$, because
\begin{align*}
    &\sum_{f \in \V(U_0)} P(\U, \X, U_0{=}f)
    \\&= 
    \sum_{f \in \V(U_0)} P(\U, f) \mathbbm1[\X = f(\U)]
    \\
    &= P(\U) \sum_{f \in \V(U_0)} \mathbbm1[\X = f(\U)] \prod_{\mat u \in \V(\U)} \nu(f(\mat u) \mid \mat u)
    \\
    &=  P(\U) \sum_{f \in \V(U_0)} \mathbbm1[\X = f(\U)] \nu(f(\U) \mid \U) \prod_{\mat u \ne \U} \nu(f(\mat u) \mid \mat u)
    \\
    &=  P(\U) \sum_{f \in \V(U_0)} \mathbbm1[\X = f(\U)] \nu(\X \mid \U) \prod_{\mat u \ne \U} \nu(f(\mat u) \mid \mat u)
    \\
    &=  P(\U) \nu(\X \mid \U) \sum_{f \in \V(U_0)} \mathbbm1[\X = f(\U)] \prod_{\mat u \ne \U} \nu(f(\mat u) \mid \mat u)
    \\
    &=  P(\U) \nu(\X \mid \U) \sum_{f \in \prod_{\mat u'}} \mathbbm1[\X = f(\U)] \prod_{\mat u \ne \U} \nu(f(\mat u) \mid \mat u)
    \\
    &= P(\U) \nu(\X \mid \U)
\end{align*}}%

We must also be careful about how to respect interventions. In the most general form, an intervention of the form $f_a \gets g$, for some function $g : \V(\Src a) \to \V(\Tgt a)$, not only modifies the equation $f_a$ by setting it equal to $g$, but also modifies the equation $f_0$ according to:
\[
    f_0 ( \mat u ) := f_0(\mat u)[ ]
\]

\TODO

For instance, if $\mathcal M$ is a PSEM, then to perform an intervention $\mat X {\gets} \mat x$ on the causal model $\derind {\cal M}\nu$,
we mean not only to replace the equation fo $f_X$, but also modify 
\[
    f_0^{\text{new}}(\mat u) := 
    f_0^{\text{old}}(\mat u)[ \mat X \mapsto \mat x]
\]

\TODO

\begin{prop}
    \begin{enumerate}
    \item $\Pr_{\derind{\mathcal M}{\nu}}(\varphi)  = 
     \Pr_{\cal M}(\varphi)$
        
    \end{enumerate}
\end{prop}
}%

\end{document}